\documentclass[letterpaper, 11pt]{article}
\pdfoutput=1
\usepackage[left=1in, right=1in, top=1in, bottom=1in]{geometry}
\usepackage[normalem]{ulem}
\usepackage{graphicx}
\usepackage{cite}
\usepackage{bm}
\usepackage{amsmath}
\usepackage{amssymb}
\usepackage{amsfonts}
\usepackage{amsthm}
\usepackage{cancel}
\usepackage{comment}
\usepackage{mathtools}
\usepackage{hyperref}
\usepackage{array}
\usepackage{braket}
\usepackage{wrapfig}
\usepackage{dsfont}
\usepackage{tikz}
\usepackage{pgfplots}
\usepgfplotslibrary{fillbetween}
\usepackage[export]{adjustbox}
\usepackage{caption}
\usepackage{subcaption}

\usetikzlibrary{calc, decorations.pathmorphing,arrows.meta}

\hypersetup{colorlinks=true, urlcolor=[rgb]{0,0,0.5}, citecolor=[rgb]{0.5,0,0}, linkcolor=[rgb]{0,0,0.4}}
\pgfplotsset{compat=1.18}

\newtheorem{theorem}{Theorem}
\newenvironment{itheorem}[2]
{\vspace{0.5\baselineskip} \noindent \textbf{#1}~(#2)\textbf{.}~\itshape }{}
\newtheorem{definition}{Definition}

\newtheorem{corollary}{Corollary}
\newtheorem{proposition}{Proposition}

\newtheorem{conjecture}{Conjecture}
\newtheorem{remark}{Remark}

\def\autorefapp#1{\hyperref[#1]{Appendix~\ref{#1}}}

\renewcommand{\title}[1]{\vbox{\center\bf{\LARGE #1}}\vspace{5mm}}
\renewcommand{\author}[1]{\vbox{\center{#1}}\vspace{5mm}}
\newcommand{\address}[1]{\vbox{\center\em#1}}

\newcommand\emails[1]{\begingroup\renewcommand\thefootnote{}\footnote{#1}\addtocounter{footnote}{-1}\endgroup}

\def\ep{\varepsilon}

\def\vphi{\varphi}
\def\tr{{\rm tr}}
\def\op{{\cal O}}
\def\C{\mathbb{C}}

\def\iden{\mathbb{I}}

\def\and{\quad {\rm and} \quad}

\def\ra{\rightarrow}
\def\CC{{\cal C}}

\def\fn{\footnotesize}

\def\ketbra#1{ |{#1}\rangle\!\langle{#1}| }

\DeclareMathOperator*{\Ex}{\mathbb{E}}

\DeclareMathOperator*{\pr}{{\rm Pr}}

\def\dist{\mathcal{D}}
\def\bwrqc{\nu_{\rm bw}}
\def\pwrqc{\nu_{\rm pw}}

\def\poly{\text{poly}}

\def\q{q}
\def\na{n_A}
\def\nb{n_B}
\def\da{D_A}
\def\db{D_B}
\def\err{\delta}
\def\CC{\mathcal{C}_{\err}}

\def\cxl{r}
\def\ssa{\rho_A}
\def\F{\mathcal{F}}
\def\rank{{\rm rank}}

\numberwithin{equation}{section}
\allowdisplaybreaks

\begin{document}

\title{Sharp Transitions for Subsystem Complexity}

\author{Yale Fan,${}^{a,b,c}$ Nicholas Hunter-Jones,${}^{c,d}$ Andreas Karch,${}^c$ Shivan Mittal${}^c$}

\address{
${}^a$Department of Physics, University of Idaho, Moscow, ID 83844 \\[6pt]
${}^b$Center for Computing Research, Sandia National Laboratories, Albuquerque, NM 87185 \\[6pt]
${}^c$Department of Physics, University of Texas at Austin, Austin, TX 78712 \\[6pt]
${}^d$Department of Computer Science, University of Texas at Austin, Austin, TX 78712
}
\emails{\hspace*{-5mm} Emails: \href{mailto:yalefan@gmail.com}{\tt yalefan@gmail.com},
\href{mailto:nickrhj@utexas.edu}{\tt nickrhj@utexas.edu},
\href{mailto:karcha@utexas.edu}{\tt karcha@utexas.edu},
\href{mailto:shivan@utexas.edu}{\tt shivan@utexas.edu}.}

\begin{abstract}
The circuit complexity of time-evolved pure quantum states grows linearly in time for an exponentially long time. This behavior has been proven in certain models, is conjectured to hold for generic quantum many-body systems, and is believed to be dual to the long-time growth of black hole interiors in AdS/CFT. Achieving a similar understanding for mixed states remains an important problem. In this work, we study the circuit complexity of time-evolved subsystems of pure quantum states. We find that for greater-than-half subsystem sizes, the complexity grows linearly in time for an exponentially long time, similarly to that of the full state. However, for less-than-half subsystem sizes, the complexity rises and then falls, returning to low complexity as the subsystem equilibrates. Notably, the transition between these two regimes occurs sharply at half system size. We use holographic duality to map out this picture of subsystem complexity dynamics and rigorously prove the existence of the sharp transition in random quantum circuits. Furthermore, we use holography to predict features of complexity growth at finite temperature that lie beyond the reach of techniques based on random quantum circuits. In particular, at finite temperature, we argue for an additional sharp transition at a critical less-than-half subsystem size. Below this critical value, the subsystem complexity saturates nearly instantaneously rather than exhibiting a rise and fall. This novel phenomenon, as well as an analogous transition above half system size, provides a target for future studies based on rigorous methods.
\end{abstract}

\section{Introduction}

Holographic duality equates the dynamics of certain strongly coupled quantum field theories with that of classical theories of gravity in one higher dimension \cite{Maldacena:1997re}.  As such, holography has emerged as a practical tool for translating geometric problems in gravity into algebraic problems in quantum mechanics.  By way of holography, classical gravity can provide a shortcut to difficult calculations in quantum many-body physics and can also function as an oracle, suggesting sometimes heuristic results that provide targets for rigorous methods based on quantum information theory.

A growing body of evidence suggests that holography makes universal predictions about strongly interacting quantum many-body systems, e.g., those that are sufficiently chaotic (namely, those with spectral statistics described by random matrix theory).  For example, holographic reasoning has found great success in the study of both static and dynamical properties of entanglement entropy.  In holography, the entanglement entropy of a boundary subregion goes like the area of the minimal-area extremal surface satisfying a homology constraint \cite{Ryu:2006bv, Ryu:2006ef, Hubeny:2007xt}.  Thus the entanglement structure of holographic states mimics that of quantum states prepared by certain tensor networks \cite{Swingle:2009bg}.  Moreover, competition between extremal surfaces leads the entanglement entropy to grow and then saturate at the thermalization time, reproducing the expected time evolution of entanglement entropy in generic quantum systems \cite{Hartman:2013qma}.

However, entanglement entropy alone is a limited probe of many-body physics.  A significantly more subtle quantity is quantum circuit complexity (or simply ``complexity''), which measures the difficulty of creating and manipulating quantum states \cite{Aaronson:2016vto}.  Via quantum Hamiltonian complexity, it can be viewed as a fundamental property of physical systems \cite{Gharibian:2015rnl}.

Complexity underlies a great two-way flow of insight between quantum gravity and many-body physics.  Studying the time evolution of complexity in quantum systems has yielded clues as to how to extend the holographic dictionary beyond black hole horizons \cite{Susskind:2014moa, Bouland:2019pvu}, shedding light on the similarities and differences between complexity and entanglement entropy as long-time probes of the black hole interior \cite{Hartman:2013qma, Brown:2017jil}.  Brown and Susskind conjectured that the circuit complexity of a pure state in a holographic system---or a strongly-interacting many-body system more broadly---grows linearly in time before saturating at a time exponential in the system size.  While the circuit complexity of specific states or unitaries is difficult to compute, this conjecture has been proven under certain assumptions, including the case where the system evolves via a local random quantum circuit (RQC) \cite{brandao2021models, Haferkamp:2021uxo, chen2024incompressibility}.  The argument is that the set of circuits constructed from random 2-local gates converges to a unitary $k$-design in depth linear in $k$, and the complexity of a $k$-design element is lower-bounded by $k$.

Our goal in this paper is to formulate and prove a version of the Brown-Susskind conjecture for subsystems and mixed states. This is a natural problem from both fundamental and practical points of view. For instance, bulk locality in holography provides a natural geometric framework for computing properties of subsystems.  This fact has been exploited to great effect in the study of entanglement entropy.  It is believed that all information about a boundary subsystem is encoded in a corresponding bulk subregion, the entanglement wedge \cite{Czech:2012bh}.  The entanglement wedge allows us to probe more subtle dynamical properties of subsystems than entropy, yet its full potential remains untapped in the realm of complexity. Moreover, understanding the dynamics of quantum circuit complexity in open quantum systems is a problem of practical value, with implications for the preparation and simulation of mixed states as well as for open-system approaches to universal quantum computation \cite{Aharonov:1998zf, Briegel:2009inn, Verstraete:2009zet, McGinley:2025fpi}.

We describe the implications of the AdS/CFT correspondence for the time evolution of quantum circuit complexity in subsystems of a quantum many-body system, and we find agreement with expectations from random quantum circuits.

We first show that holography predicts universal behavior for subsystem complexity, where the complexity of a mixed state is defined by purification.  In particular, it predicts a sharp transition at half size as well as the behavior before and after the transition.  Previous work has defined a notion of purification complexity for mixed states and studied its relation to holographic subregion complexity (based on the ``complexity = volume'' (CV) \cite{Susskind:2014rva, Stanford:2014jda} and ``complexity = action'' (CA) \cite{Brown:2015bva, Brown:2015lvg} proposals and their variants) \cite{Alishahiha:2015rta, Ben-Ami:2016qex, Abt:2017pmf, Agon:2018zso, Caceres:2018blh, Caceres:2019pgf}, but has not studied its dynamics.  On the other hand, the time dependence of holographic complexity \cite{Carmi:2017jqz, Chen:2018mcc, Auzzi:2019mah} has not, thus far, been studied for proper subregions that comprise a finite fraction of the boundary.  We aim to fill both of these gaps.  While many potential geometrizations of complexity have been proposed on the holographic side, we argue that all of these definitions exhibit the following feature: holographic subsystem complexity grows for exponentially long times if the subsystem is at least half the size of the whole system, and otherwise saturates rapidly to a constant value.  This phenomenon is robust against ambigui\-ties in the definition of holographic complexity because the difference in behavior with respect to subsystem size is driven by the entanglement wedge, the subregion of the gravitational spacetime that contains all of the information about the subsystem density matrix.  For subsystems smaller than half the total system size, the entanglement wedge does not include the growing region of the bulk spacetime that gives rise to linear growth of \emph{any} subsystem-related quantity.

We then argue that this holographic understanding is reflected in chaotic quantum many-body systems more generally. RQCs are a solvable model of local quantum-chaotic dynamics in which quantities concerning the spread of local information and the onset of thermalization can be computed precisely. The growth of circuit complexity in this model has been well-understood, where lower bounds rigorously establish that the complexity grows linearly in time for an exponentially long time. This statement holds for the circuit complexity of the state (prepared from an initial product state) as well as that of the unitary evolution itself. In this work, we examine the circuit complexity of subsystems of states time-evolved by random quantum circuits. Similar to the picture understood in AdS/CFT, we find that the subsystem complexity for random circuits also displays distinctly different behaviors for subsystems greater than and less than half system sizes. We consider $n$-qubit random quantum circuits in one spatial dimension, where the evolution is given by a depth-$t$ circuit consisting of $t$ layers of randomly chosen gates. Each gate is chosen randomly with respect to the Haar measure on the 2-site unitary group $U(4)$, and the gates are arranged in a brickwork fashion. More precisely, we consider two similar 1D RQCs, which we term {\it brickwork} and {\it patchwork} RQCs. We strongly believe that the qualitative features of subsystem complexity growth should not depend on these choices of random circuits, and should hold for any connected architecture and universal gate set. The reason for our focus on these two models is that they allow certain quantities to be computed very precisely.

Prior work has studied the time evolution of several quantum information-theoretic properties of subsystems, although from different perspectives than ours.  Examples include the time evolution of subsystem complexity in spin chains (which, in some cases, qualitatively resembles that of entanglement entropy) \cite{DiGiulio:2021oal, DiGiulio:2021noo}; subsystem complexity phase transitions due to measurements \cite{Jian:2023mdh}; the dynamics of subsystem information capacity in random quantum circuits \cite{Chen:2024abj}; and the dynamics of entanglement asymmetry (defined in terms of a R\'enyi entropy) in RQCs as a function of subsystem size, for which a qualitative change in behavior is found when the subsystem size crosses half the system size \cite{Ares:2025ljj}.

This paper is organized as follows. In \autoref{sec:motivation}, we present the basic holographic construction and describe the various sharp transitions. In \autoref{sec:holography}, we give a quantitative description of the relevant minimal surfaces, leaving technical details to \autoref{app:holodetails}. In \autoref{sec:rqcs}, we derive rigorous bounds on subsystem complexity in random unitary circuits, which exhibit behavior consistent with holographic predictions, leaving technical details to \autoref{app:rqccomp}. In \autoref{sec:discussion}, we summarize our results and outline some future prospects. In \autoref{app:hawkingpage}, we provide some background material about the holographic description of the eternal black hole background relevant for our discussion.

While this paper was in preparation, we learned that complementary results have been obtained in the concurrent work \cite{Haah:2025hyf}.

\section{Motivation: Observations in Holographic Systems} \label{sec:motivation}

Our study is motivated by highly fruitful analogies between a black hole of entropy $S$ and a random quantum circuit on $n$ qudits ($S = n\log q$ for qudits of local dimension $q$): 
\begin{itemize}
\item On one hand, the volume of the interior of a black hole (or ``Einstein-Rosen bridge'') should grow linearly with time for a time exponential in $S$ \cite{Susskind:2014moa, Bouland:2019pvu}.\footnote{An expectation for the rate of this linear growth is given in \cite{Susskind:2014moa}.}  Classically, it would grow forever, but at times of order $e^S$, we would expect classical gravity to break down and exponentially suppressed geometries to contribute significantly.  Indeed, the finiteness of the quantum gravity Hilbert space requires that the growth saturate, and it has been argued to do so on the basis of nonperturbative effects in gravity (e.g., in \cite{Iliesiu:2021ari, Balasubramanian:2022gmo}).\footnote{In the large-$S$ limit, we can work purely in classical gravity, where a boundary state is dual to a single geometry.}
\item On the other hand, we would expect the complexity of an $n$-qudit quantum state prepared via a random quantum circuit to grow linearly with time (circuit depth) until it saturates at a time exponential in $n$.  One possible approach to formulating and proving this statement is given in \cite{Haferkamp:2021uxo}.
\end{itemize}
We ask: do these analogies extend to \emph{subsystems} of $n$ qudits?

This is a natural question from the holographic point of view.  Holography yields simple prescriptions for computing properties of subsystems of a non-gravitational quantum system, most notably their entanglement entropy via the Ryu-Takayanagi (RT) formula \cite{Ryu:2006bv, Ryu:2006ef}.  However, entanglement entropy evolves very differently than complexity over time \cite{Hartman:2013qma} (\autoref{fig:HMtransitionredone}).  If one chooses identical subregions of linear size $L$ on each side of an eternal black hole, then there exist two candidate extremal-area bulk surfaces ending on their boundary, which we refer to as RT surfaces: a connected surface that traverses the Einstein-Rosen (ER) bridge and a disconnected surface with two caps.  Typically, the area of the connected surface is initially minimal and grows linearly with time, but at a time of order $L$ (the thermalization time), it exceeds that of the disconnected surface and quantities based on RT surfaces become time-independent.  The volume of the ER bridge continues to grow even after the entropy stops growing.  Hence entanglement entropy-based quantities cannot capture the exponentially long linear growth of the eternal black hole, motivating its dual description as a complexity \cite{Susskind:2014moa}.

\begin{figure}[!htb]
\centering
\includegraphics[width=0.8\textwidth]{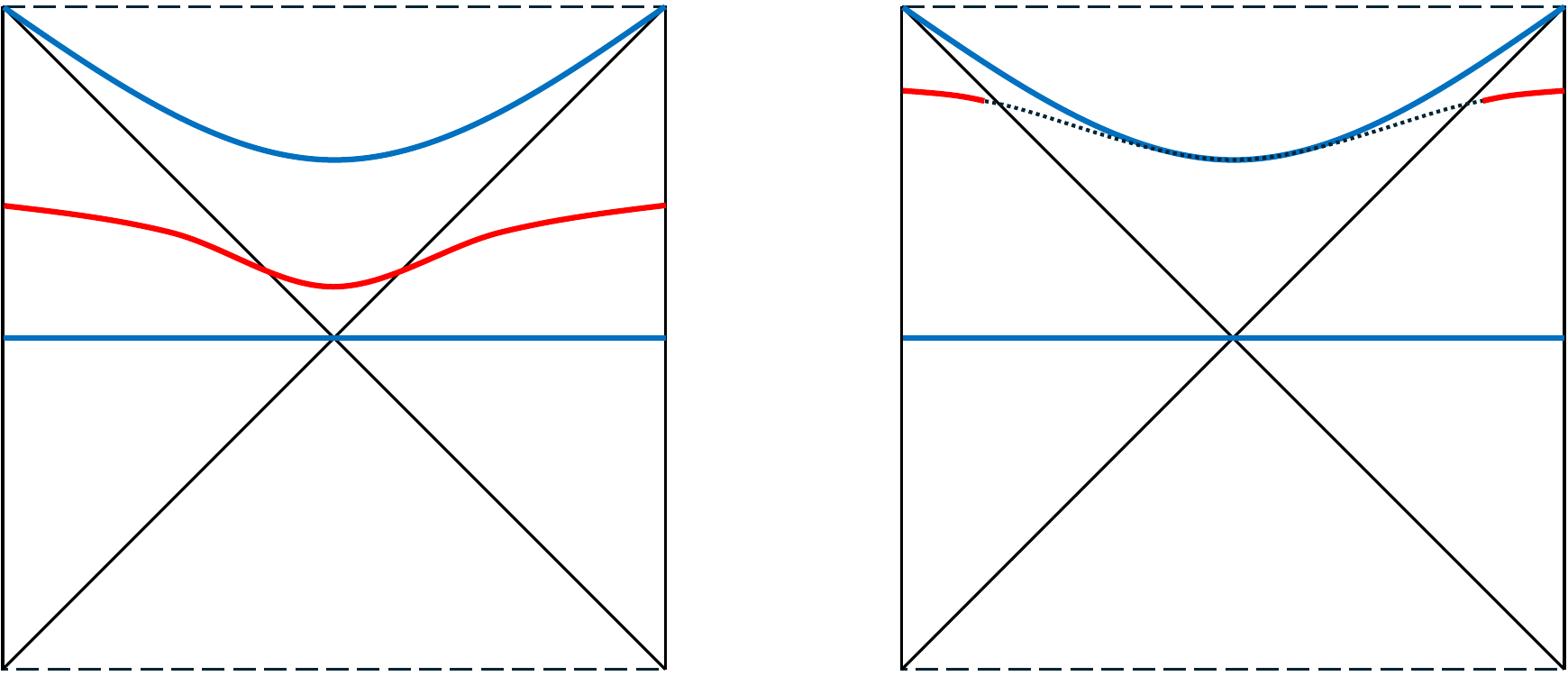}
\caption{Hartman-Maldacena transition.  Left: early-time (connected) extremal surface crossing the ER bridge.  Right: late-time (disconnected) extremal surface terminating outside the horizon.  The two surfaces exchange dominance at the thermalization time.}
\label{fig:HMtransitionredone}
\end{figure}

What does holography predict for the complexity growth of a subsystem?  Any property of a subsystem should be derivable from its reduced density matrix.  In holography, the reduced density matrix of a subregion is encoded in its bulk entanglement wedge \cite{Czech:2012bh}, which is bounded by its RT surface.\footnote{We are being loose with terminology.  More precisely, the entanglement wedge is the bulk domain of dependence of a subregion of a bulk Cauchy slice bounded by the boundary subregion and its RT surface.  Note also that there exist many distinct notions of subregion duality, as summarized in, e.g., \cite{Bao:2024hwy}.}  This means that any holographic prescription for computing subsystem complexity should be sensitive to the Hartman-Maldacena (HM) transition between connected and disconnected RT surfaces.  Nonetheless, the behavior of holographic subsystem complexity need not mirror the behavior of holographic entanglement entropy because the former depends on the entire entanglement wedge and not only on the RT surface.

To see this, consider the boundary subregion $I_L\cup I_R$ in the two-sided setup, where $I_{L, R}$ are intervals and each boundary has the topology of a circle, for simplicity.  Regardless of the interval sizes, the RT surface undergoes a transition in time.  Before the transition, the interior of the RT surface (i.e., the entanglement wedge) threads the ER bridge.  However, after the transition:
\begin{itemize}
\item If the intervals are less than half the size of each boundary, then the entanglement wedge is the ``inside'' of the disconnected RT surface.  Hence the entanglement wedge is disconnected and does not traverse the ER bridge.
\item If the intervals are more than half the size of each boundary, then the entanglement wedge is the ``outside'' of the disconnected RT surface.  Hence the entanglement wedge is connected and includes the ER bridge.
\end{itemize}
Qualitatively, therefore, holography predicts that complexity grows for exponentially long times only if the subsystem is at least half the size of the whole system.

This conclusion accords with intuition from quantum mechanics: for subsystems of size less than half that of the total system, we expect the reduced density matrix to quickly become close to maximally mixed, and the maximally mixed state has low complexity. (Here, we have in mind approximate circuit complexity, which is robust under infinitesimal perturbations of the state; it may be the case that small deviations from the maximally mixed state result in large fluctuations in exact circuit complexity \cite{Haferkamp:2021uxo}.) We are led to the somewhat surprising conclusion that a sharp transition occurs from polynomially long growth to exponentially long growth of subsystem complexity at exactly half the system size.

In addition to this transition at half system size, we demonstrate that the holographic theory exhibits two more sharp transitions:
\begin{itemize}
\item At finite temperatures, there exists a second critical subsystem size, less than half, below which the complexity saturates essentially instantaneously but above which the complexity saturates in finite time. This transition disappears in the infinite-temperature limit.  Likewise, there exists a critical greater-than-half subsystem size, symmetric about half system size with respect to the previous one, above which the complexity grows for an exponentially long time at a constant rate but below which the complexity grows at two different rates.
\item In the range of subsystem sizes where the complexity saturates in finite (polynomially long) time, it drops suddenly to a constant value at the saturation time after the initial linear rise.  Likewise, in the regime where the complexity grows at two different rates, it jumps upward discontinuously at the (polynomially long) time when the slope of the linear growth increases.
\end{itemize}
The various transitions in both space and time are summarized in \autoref{figtable}.

\begin{table}[!htb]
\centering
\definecolor{db}{rgb}{0.271, 0.459, 0.706}
\begin{tabular}{|c|c|c|c|} \hline
Subsystem Size & Early Times & Late Times & Complexity vs.\ Time \\ \hline
$p < p_\text{crit}$ & \includegraphics[scale=0.17, valign=c]{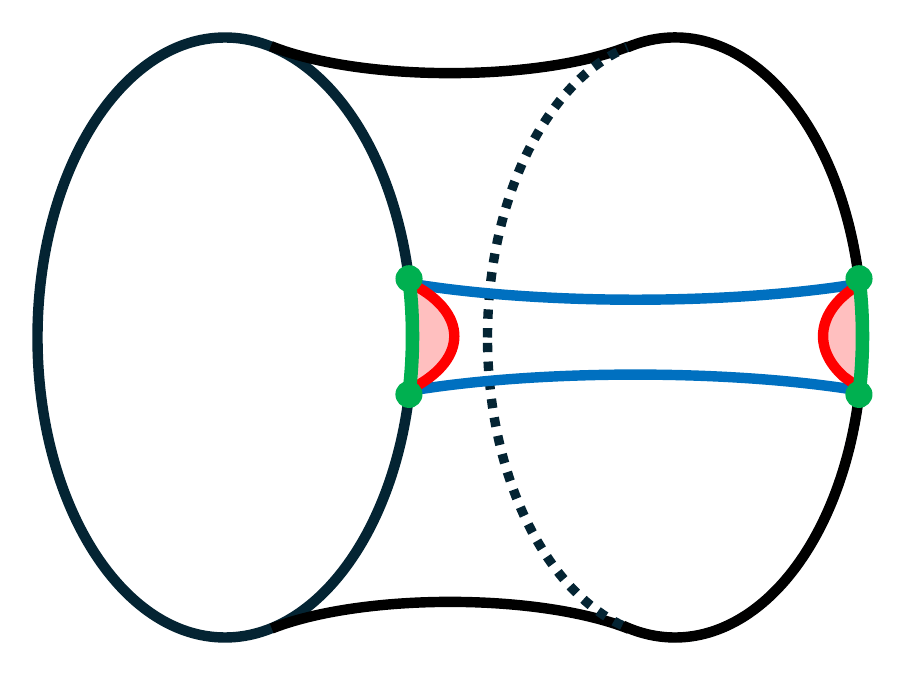} & \includegraphics[scale=0.17, valign=c]{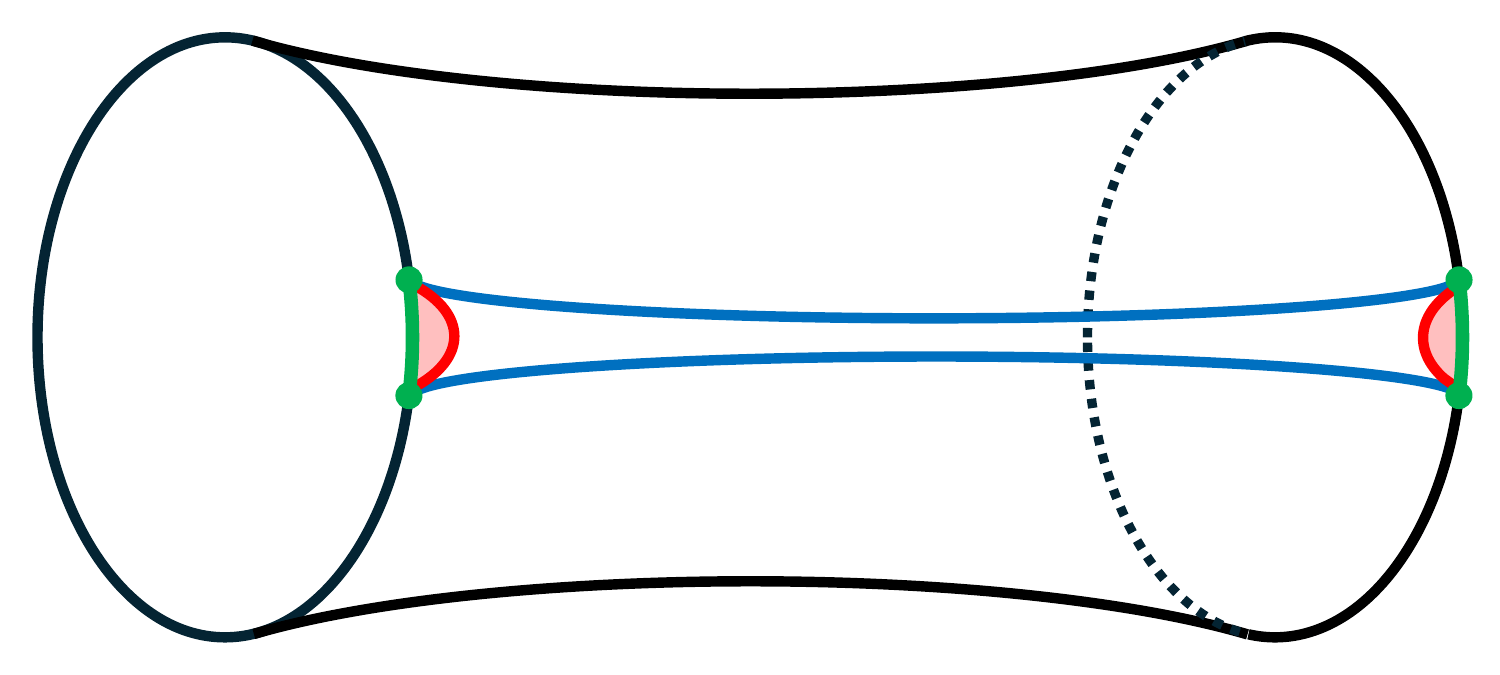} & 
\hspace*{-6pt}
\scalebox{0.8}{
\begin{tikzpicture}[baseline=12mm]
    \draw[thick,->] (0,0) -- (5,0);
    \draw[thick,->] (0,0) -- (0,3);
    \node at (2.5,-0.25) {{\small time}};
    \node[rotate=90] at (-0.3,1.5) {{\small complexity}};
    \draw[line width=1.2pt,red] (0,0.5) -- (4.9,0.5);
    \node at (0,3) {};
\end{tikzpicture}}~
\\ \hline
$p_\text{crit} < p < \dfrac{1}{2}$ & \includegraphics[scale=0.17, valign=c]{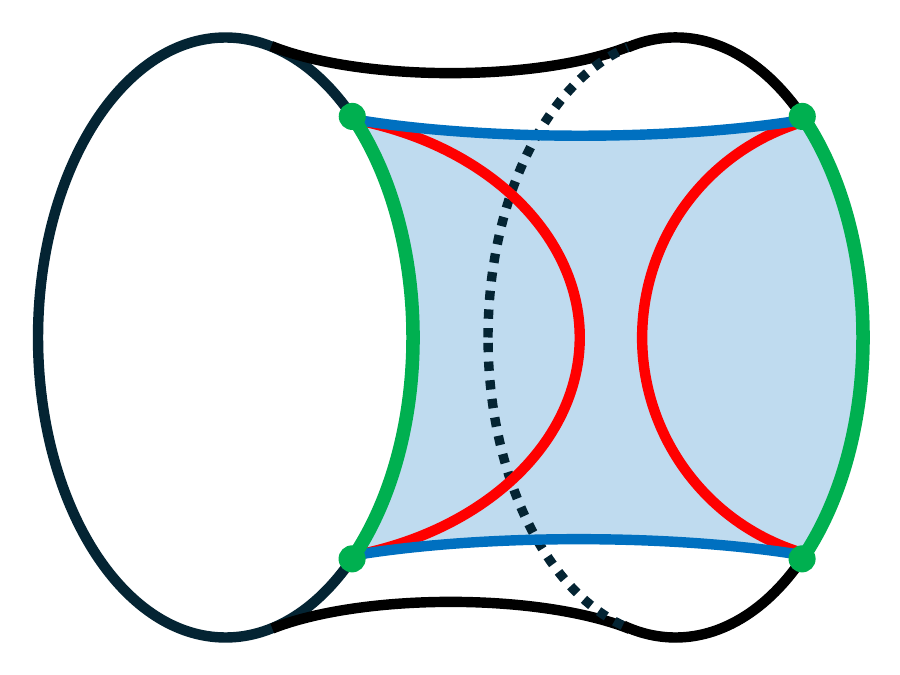} & \includegraphics[scale=0.17, valign=c]{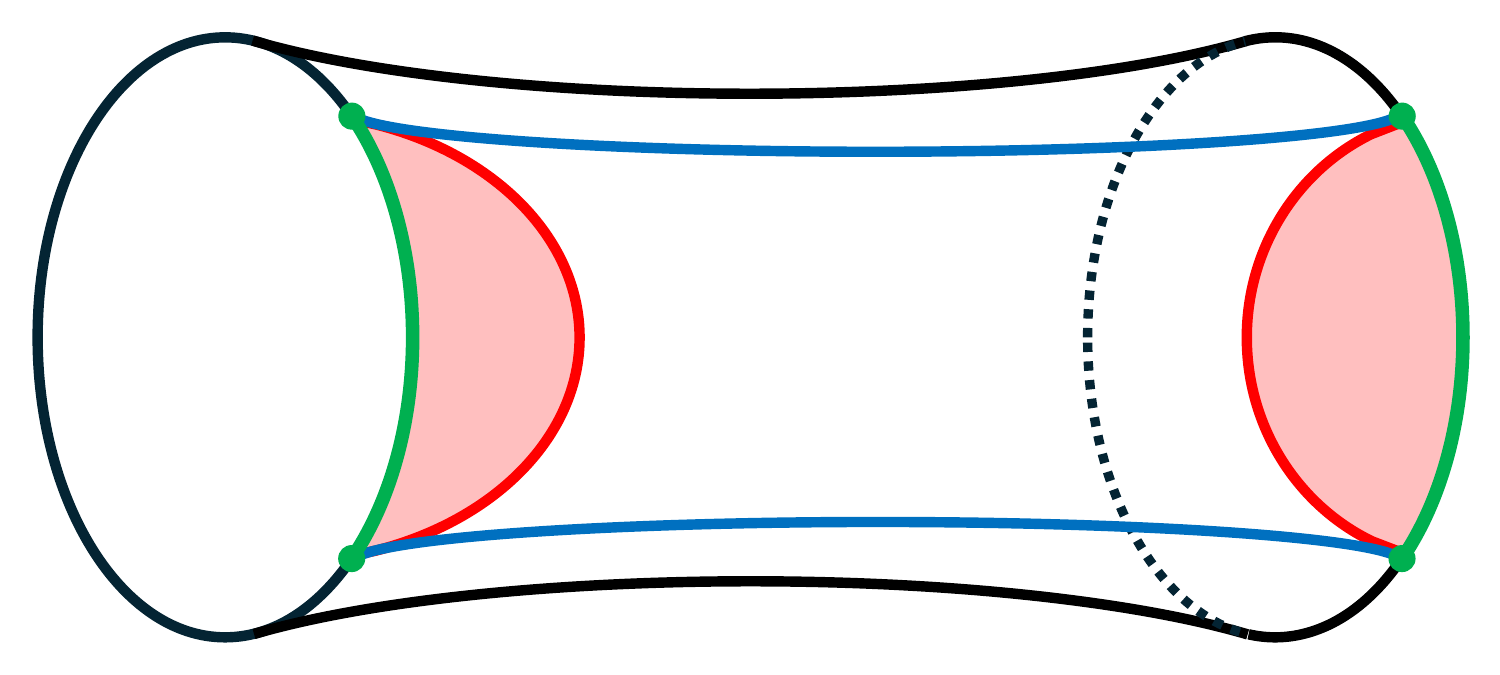} & 
\hspace*{-6pt}
\scalebox{0.8}{
\begin{tikzpicture}[baseline=12mm]
    \draw[thick,->] (0,0) -- (5,0);
    \draw[thick,->] (0,0) -- (0,3);
    \node at (2.5,-0.25) {{\small time}};
    \node[rotate=90] at (-0.3,1.5) {{\small complexity}};
    \draw[line width=1.2pt,db] (0,0.5) -- (2.1,2);
    \draw[line width=1.2pt,dashed,black] (2.1,2) -- (2.1,0.8);
    \draw[line width=1.2pt,db,black!10!red] (2.1,0.8) -- (4.9,0.8);
    \node at (0,3) {};
\end{tikzpicture}}~
\\ \hline
$\dfrac{1}{2} < p < 1 - p_\text{crit}$ & \includegraphics[scale=0.17, valign=c]{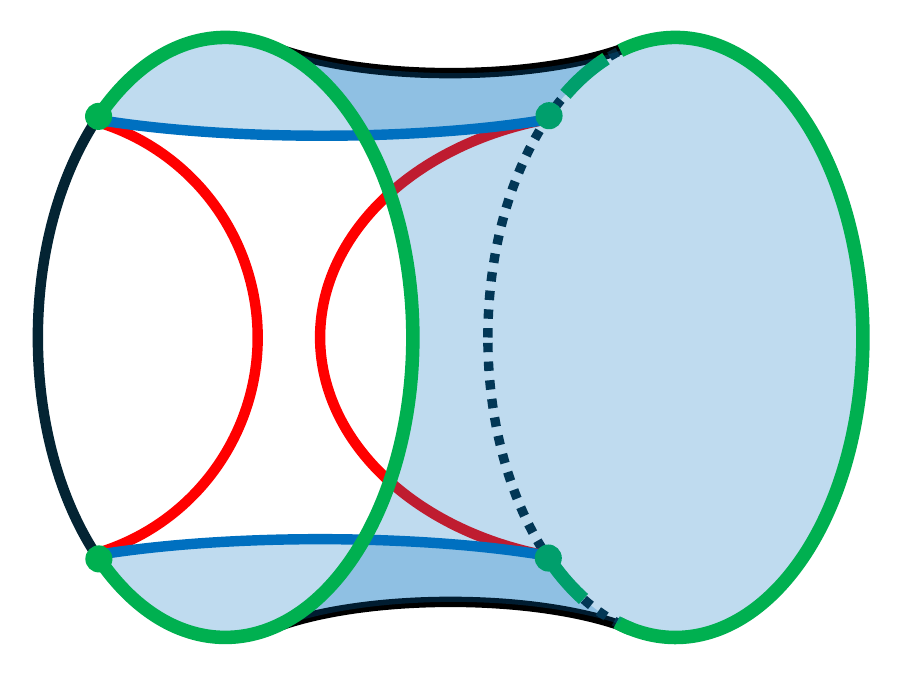} & \includegraphics[scale=0.17, valign=c]{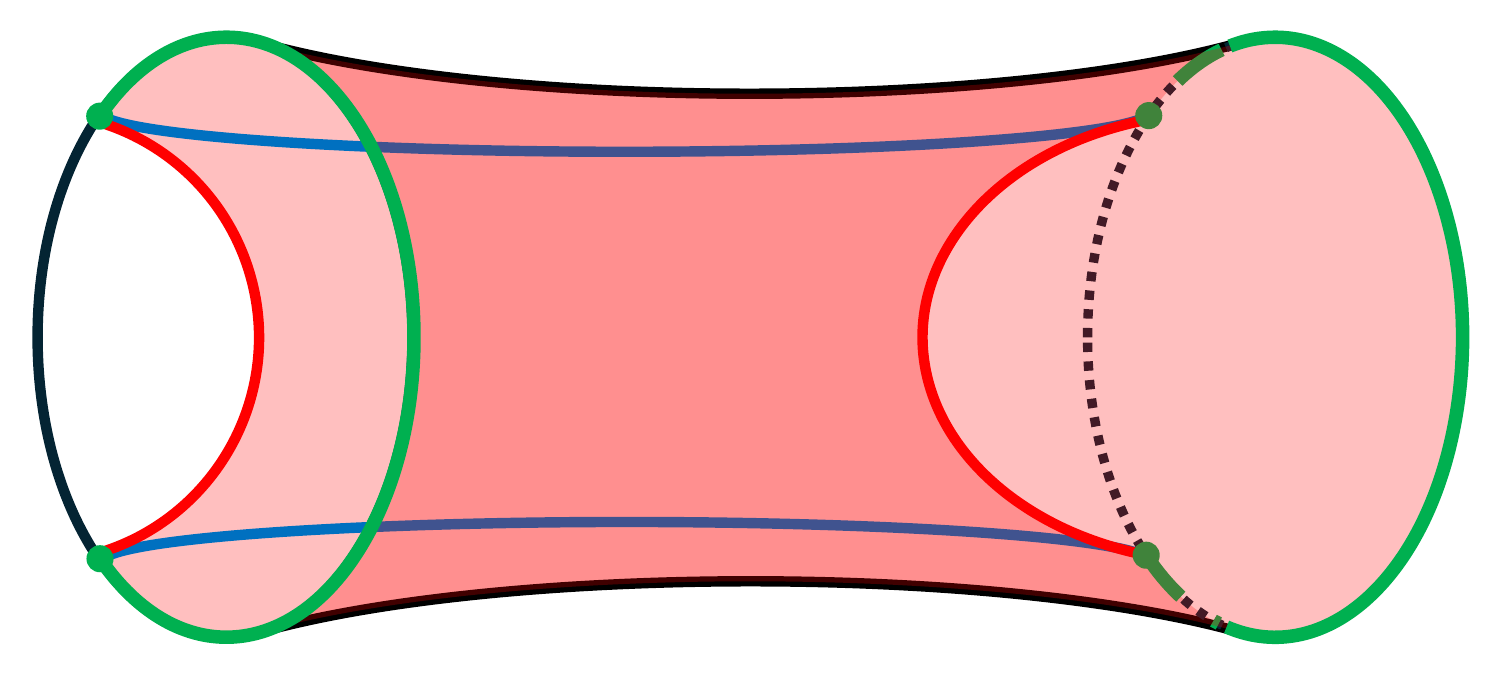} & 
\hspace*{-6pt}
\scalebox{0.8}{
\begin{tikzpicture}[baseline=12mm]
    \draw[thick,-{Bar[width=4pt]}] (0,0) -- (2.2,0);
    \draw[thick,{Bar[width=4pt]}->] (2.8,0) -- (5,0);
    \node at (2.53,0) {$\cdots$};
    \draw[thick,-{Bar[width=4pt]}] (0,0) -- (0,1.2);
    \draw[thick,{Bar[width=4pt]}->] (0,1.8) -- (0,3);
    \node at (0,1.6) {$\vdots$};
    \node at (2.5,-0.25) {{\small time}};
    \node[rotate=90] at (-0.3,1.5) {{\small complexity}};
    \draw[line width=1.2pt,db] (0,0.2) -- (1,0.5);
    \draw[line width=1.2pt,dashed,black] (1,0.5) -- (1,1.01875);
    \draw[line width=1.2pt,db,black!10!red] (1,1.01875) -- (3.2,2.6) -- (4.9,2.6);
    \node at (0,3) {};
\end{tikzpicture}}~
\\ \hline
$p > 1 - p_\text{crit}$ & \includegraphics[scale=0.17, valign=c]{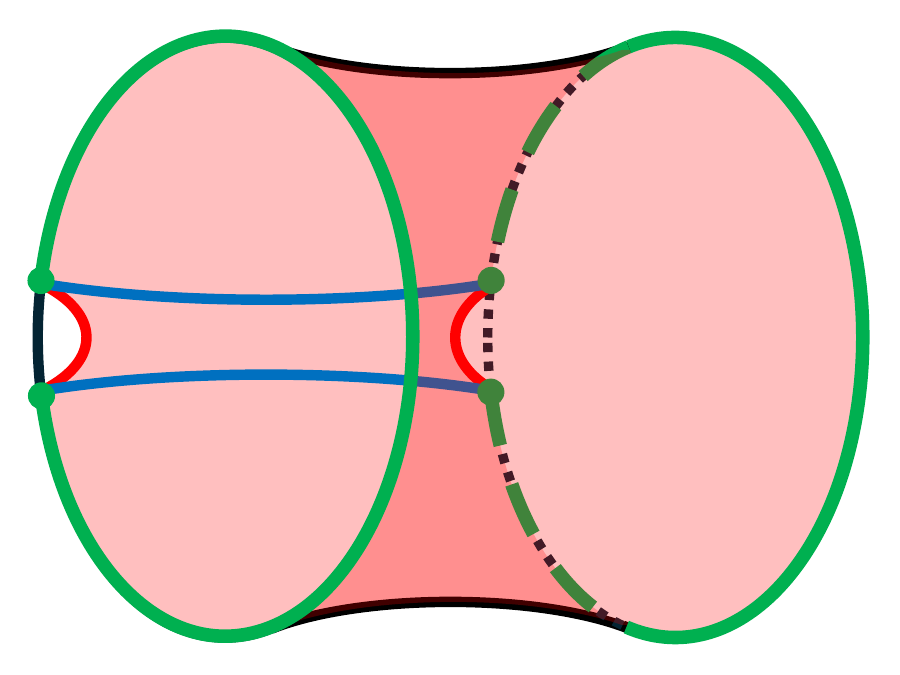} & \includegraphics[scale=0.17, valign=c]{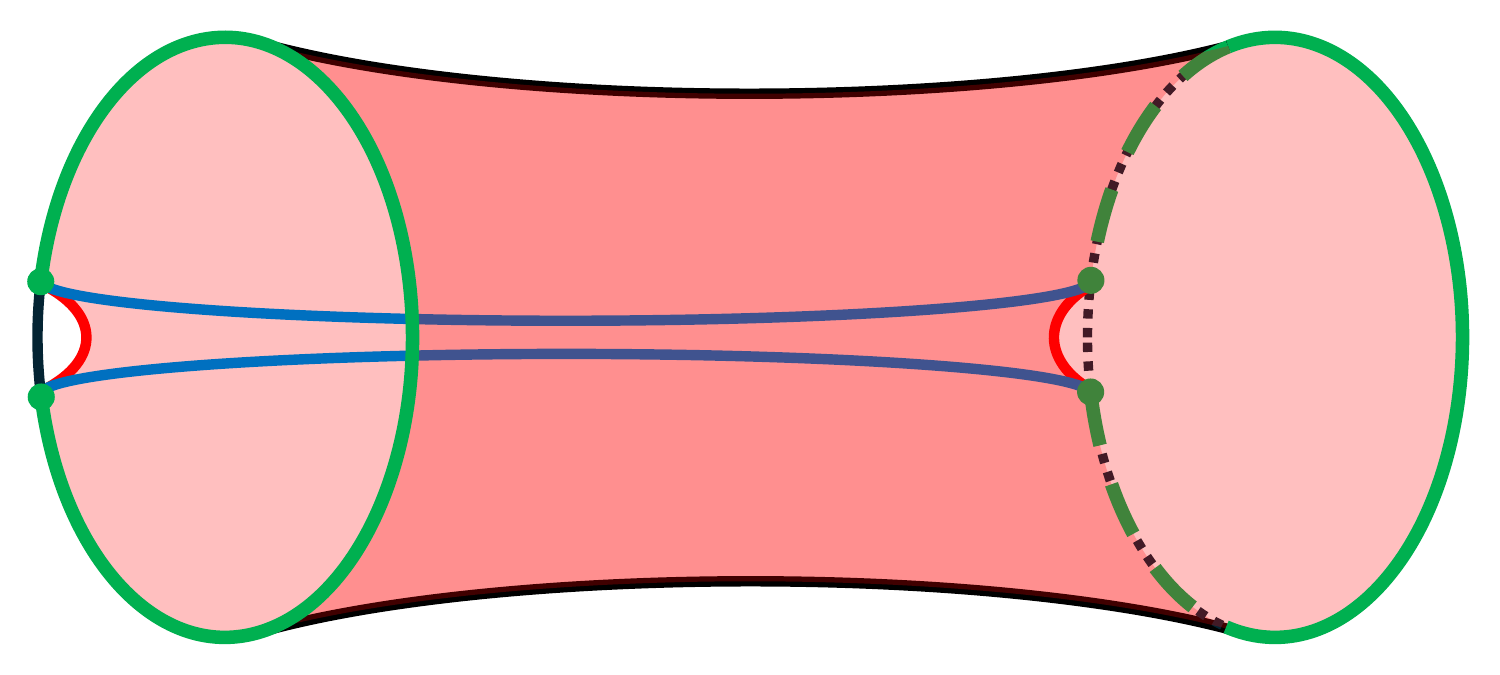} & 
\hspace*{-6pt}
\scalebox{0.8}{
\begin{tikzpicture}[baseline=12mm]
    \draw[thick,-{Bar[width=4pt]}] (0,0) -- (2.2,0);
    \draw[thick,{Bar[width=4pt]}->] (2.8,0) -- (5,0);
    \node at (2.53,0) {$\cdots$};
    \draw[thick,-{Bar[width=4pt]}] (0,0) -- (0,1.2);
    \draw[thick,{Bar[width=4pt]}->] (0,1.8) -- (0,3);
    \node at (0,1.6) {$\vdots$};
    \node at (2.5,-0.25) {{\small time}};
    \node[rotate=90] at (-0.3,1.5) {{\small complexity}};
    \draw[line width=1.2pt,db,black!10!red] (0,0.5) -- (3.2,2.8) -- (4.9,2.8);
    \node at (0,3) {};
\end{tikzpicture}}~
\\ \hline
\end{tabular}
\caption{Complexity growth for different subsystem sizes $p$ ($0 < p < 1$), where the critical parameter $p_\text{crit}$ ($0 < p_\text{crit} < 1/2$) decreases with temperature.  Quantum complexity transitions in both space and time are mapped, via holography, to the behavior of minimal surfaces in negatively curved spaces (in this case, ``wormholes'' with two circular boundaries).  Boundary subsystems are drawn as green intervals.  Candidate minimal surfaces are drawn with red and blue lines: red lines indicate the static extremal surface, while blue lines indicate the growing extremal surface.  Shaded red or blue regions indicate the volume of the dominant (minimal-area) extremal surface, which measures the complexity of the subsystem.  The red and blue surfaces may exchange dominance as the wormhole grows in length, leading to holographic predictions for complexity dynamics.  First row: red shaded region is time-independent.  Second row: shaded region transitions from blue (time-dependent) to red (time-independent).  Third row: shaded region transitions from blue (time-dependent) to red (also time-dependent, but with a faster rate of growth).  Fourth row: red shaded region dominates from the beginning and grows with time because it includes the wormhole (ER bridge).}
\label{figtable}
\end{table}

To extract quantitative predictions for subsystem complexity growth from gravity, we use the CV proposal \cite{Susskind:2014rva, Stanford:2014jda} for convenience.  To model a quantum system of finite spatial extent, we put the dual field theory on a compact space by considering two-sided spherical AdS black holes and ball-shaped subregions of variable size.  We find that holography predicts that the complexity of a sufficiently small subsystem saturates after a polynomially short time.

On the random quantum circuit side, we adopt a definition of subsystem complexity that en\-tails minimizing over all purifications with a polynomial number of ancillas (other definitions are proposed in, e.g., \cite{Agon:2018zso}).  We find that the above holographic expectations are mirrored in the be\-hav\-ior of subsystems of random quantum circuits, which provides additional evidence that gravity makes universal predictions for generic quantum systems.\footnote{Nonetheless, random quantum circuits of two-qudit gates may still be imperfect models of black hole dynamics.  For instance, it has been argued that such models cannot produce computationally pseudorandom quantum states (states indistinguishable from Haar-random by an observer with polynomial resources) in the $\log n$ depth needed to qualify black holes as the fastest scramblers in nature \cite{Chamon:2023lsv}.}

\section{Subsystem Complexity Dynamics in Holography} \label{sec:holography}

We now describe the basic holographic mechanism in more detail.  Consider an eternal AdS black hole, which is dual to a thermofield double (TFD) state on the boundary \cite{Maldacena:2001kr}, with complementary regions $A, B$ comprising the left boundary and corresponding regions $A', B'$ comprising the right boundary.  The inside of the black hole takes the form of a wormhole extending between the left and right horizons.  Consider the two-sided region $R_A = A\cup A'$ and its complement $R_B = B\cup B'$.  Because the TFD state is pure, $R_A$ and $R_B$ have the same RT surface.  We focus on late times, when this RT surface does not thread the wormhole.

Let $\gamma_A$ and $\gamma_B$ be the RT surfaces for regions $A$ and $B$ individually.  If $\gamma_A$ is shorter than $\gamma_B$ (see \autoref{fig:RTfigureredone}, adapted from Figure 5 of \cite{Ryu:2006ef}), then the RT surface of the previous paragraph is $\gamma_A\cup \gamma_{A'}$ and the entanglement wedge of $R_A$ does not include the black hole interior: it has two connected components.  However, the entanglement wedge of $R_B$ is connected and does include the black hole interior. It is the same as the entanglement wedge of the global TFD state minus the pieces belonging to the entanglement wedge of $R_A$.  Thus the complexity of the state associated to $R_A$ has saturated, while that of the state associated to $R_B$ continues growing linearly with time.  This conclusion makes sense because the complexity of a subsystem that is only slightly smaller than the whole system should behave almost identically to that of the whole system.  What may not have been obvious from the beginning is the existence of a sharp transition at half system size.

\begin{figure}[!htb]
\centering
\includegraphics[width=0.8\textwidth]{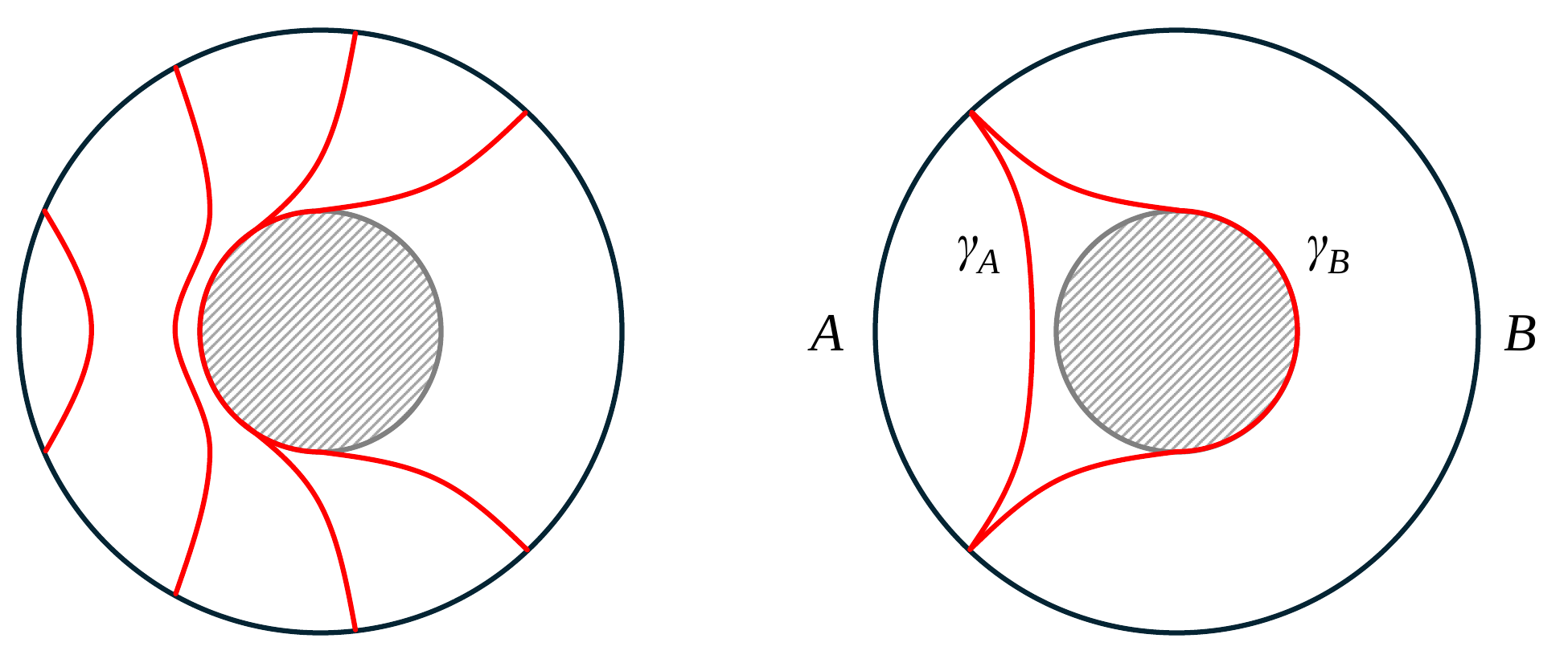}
\caption{The one-sided BTZ black hole \cite{Banados:1992wn} (black circle = boundary, gray circle = horizon).  Left: RT surfaces for various subregion sizes.  Right: RT surfaces for complementary regions $A$ and $B$.  In the two-sided case, only the smaller surface $\gamma_A$ is relevant.}
\label{fig:RTfigureredone}
\end{figure}

In general, the geometry dual to a mixed CFT state has a horizon; there are two different RT surfaces in light of the associated homology constraint, reflecting the fact that the von Neumann entropies differ for the two subsystems in a bipartition.  The key point is that the TFD state is pure even though it looks thermal from the point of view of a single boundary, so the RT surfaces for the regions $R_A$ and $R_B$ (which have support on both boundaries) are the same even though the RT surfaces for $A$ and $B$ (which have support on a single boundary) differ.  In other words, if we consider the single-sided region $B$, then only the surface $\gamma_B$ satisfies the constraint of being homologous to $B$, but if we consider the two-sided region $R_B$, then the topology of the wormhole allows the surface $\gamma_A\cup \gamma_{A'}$ to satisfy the constraint that it be homologous to $R_B$.

In the following discussion, we flesh out this general picture by explicitly calculating the relevant RT surfaces as functions of time and subsystem size. One intriguing finding is that a second transition in fact occurs at a critical subsystem size below half the system size. At very small subsystem sizes, the complexity saturates almost instantaneously, whereas at subsystem sizes between the critical size and half that of the full system, the complexity grows for a time that is of order the subsystem size (and that does not scale with the local qubit density, which is very large in systems with holographic duals). As argued on general grounds, above half system size, the complexity grows linearly for a time that is exponential in the local qubit density and hence in the total number of qubits.

\subsection{Setup}

We consider the eternal (large) black hole in AdS$_{d+1}$, which is dual to the TFD state of two copies of the CFT$_d$ \cite{Maldacena:2001kr}.\footnote{The precise CFT dual to an eternal AdS black hole depends on the internal space.  The black hole geometry only captures macroscopic thermodynamic quantities, not the CFT spectrum (black hole microstates).}  We adapt the calculation of \cite{Hartman:2013qma} to the case of a spherical boundary and examine several aspects in greater detail.\footnote{Similar calculations can be found in \cite{Ryu:2006ef, Alishahiha:2015rta, Carmi:2016wjl, Agon:2018zso, Caceres:2018blh, Caceres:2019pgf, Geng:2020fxl}.  The holographic subsystem CV and CA prescriptions are introduced in \cite{Alishahiha:2015rta} and \cite{Carmi:2016wjl}, respectively: \cite{Alishahiha:2015rta} studies the CV proposal in the static single-sided setting, while \cite{Carmi:2016wjl} generalizes it to time-dependent geometries.  \cite{Agon:2018zso} raises issues with CV for subsystems in the context of static two-sided AdS black holes where the subsystem is a constant-time slice of a \emph{single} boundary.}

First consider Euclidean AdS$_{d+1}$ with $d\geq 2$ \cite{Witten:1998zw}.  Let the asymptotic boundary be $S^{d-1}_\ell\times S^1_\beta$, where $\ell$ is the radius of the $S^{d-1}$ and $\beta$ is the circumference of the $S^1$.  The AdS-Schwarzschild solution of mass $M$ takes the form
\begin{equation}
ds^2 = f(r)\, dt_E^2 + \frac{dr^2}{f(r)} + r^2\, d\Omega_{d-1}^2, \qquad f(r) = 1 + \frac{r^2}{\ell^2} - \frac{\mu}{r^{d-2}}, \qquad t_E\sim t_E + \beta,
\end{equation}
where $t_E$ is Euclidean time, $\mu = 16\pi G_N M/(d - 1)\omega_{d-1}$, and
\begin{equation}
\omega_{d-1} = \operatorname{vol}(S^{d-1}_1) = \frac{2\pi^{d/2}}{\Gamma(\frac{d}{2})}
\end{equation}
is the volume of the unit $(d - 1)$-sphere.  Note that in $d + 1$ bulk dimensions, the mass dimension of $G_N$ is $[G_N] = 1 - d$.  The metric obeys the thermal boundary condition
\begin{equation}
ds^2\to \frac{r^2}{\ell^2}\, dt_E^2 + \frac{\ell^2}{r^2}\, dr^2 + r^2\, d\Omega_{d-1}^2, \qquad t_E\sim t_E + \beta
\label{thermalBC}
\end{equation}
as $r\to\infty$.  The horizon is the outermost solution to $f(r_h) = 0$ (in terms of $\ell, \mu$), and smoothness at $r = r_h$ sets $\beta(r_h)$:
\begin{equation}
\beta(r_h) := \frac{4\pi}{f'(r_h)} = \frac{4\pi\ell^2 r_h}{dr_h^2 + (d - 2)\ell^2}.
\end{equation}
(When $d = 2$, the horizon radius is $r_h = \ell\sqrt{\mu - 1}$, making clear the black hole threshold $\mu > 1$ in AdS$_3$.) Note that the maximum $\beta$ (minimum temperature) and corresponding $r_h$ are
\begin{equation}
\beta_\text{max} = \frac{2\pi\ell}{\sqrt{d(d - 2)}}, \qquad r_\ast = \sqrt{\frac{d - 2}{d}}\ell.
\end{equation}
For any other value of $\beta$, there are two different black hole solutions: small ($r_h < r_\ast$) and large ($r_h > r_\ast$).  We are interested in the large black hole, as it is thermodynamically preferred.  The coordinates $(r, t_E)$ with $r\geq r_h$ and $t_E\in [0, \beta)$ parametrize a disk (cigar).  The origin is smooth and corresponds to the Euclidean horizon.  The thermal circle is contractible.  In thermal AdS, $(r, t_E)$ parametrize a cylinder rather than a disk.  The thermal circle is non-contractible.

To obtain the time-evolved TFD state (with time going up on both sides, in contrast to the Schwarzschild time), we consider the Lorentzian black hole.  In this case, the metric is
\begin{equation}
ds^2 = -f(r)\, dt^2 + \frac{dr^2}{f(r)} + r^2\, d\Omega_{d-1}^2,
\label{lormetric}
\end{equation}
with $\beta$ given in terms of $r_h$ as before (although the time direction is no longer periodic).\footnote{We use $(r, t)$ coordinates.  To make contact with the $(\rho, t)$ coordinates of \cite{Hartman:2013qma}, we write $\rho = \int_{r_h}^r dr'/\sqrt{f(r')}$, so that the horizon lies at $\rho = 0$.  Then we have
\[
ds^2 = -g(\rho)^2\, dt^2 + h(\rho)^2\, d\Omega_{d-1}^2 + d\rho^2, \qquad g(\rho) = \sqrt{f(r(\rho))}, \qquad h(\rho) = r(\rho).
\]
Near $\rho = 0$, we have $r(\rho) = r_h + \frac{1}{4}f'(r_h)\rho^2 + O(\rho^4)$, where $r'(\rho) = \sqrt{f(r)}$ implies that $r^{(n)}(0) = 0$ for $n$ odd.  Near $r = r_h$, we have $f(r_h + \epsilon) = f'(r_h)\epsilon + O(\epsilon^2)$.  Hence
\[
g(\rho) = \frac{1}{2}f'(r_h)\rho + O(\rho^3) = \frac{2\pi}{\beta}\rho + O(\rho^3), \qquad h(\rho) = r_h + O(\rho^2),
\]
as in \cite{Hartman:2013qma} (where they set $\beta = 2\pi$).}

To orient ourselves, before considering subregions, we first consider spatial slices of maximal volume that are anchored on the entire spherical boundary.  Such slices are rotationally invariant.  If we specify such a slice by a function $r(t)$, then the induced metric is
\begin{equation}
ds^2 = \left[-f(r) + \frac{\dot{r}^2}{f(r)}\right]dt^2 + r^2\, d\Omega_{d-1}^2
\end{equation}
and the volume functional is
\begin{equation}
V[r(t)] = \omega_{d-1}\int dt\, r^{d-1}\sqrt{-f(r) + \frac{\dot{r}^2}{f(r)}}.
\label{volume}
\end{equation}
Since the integrand lacks explicit $t$-dependence, the following ``energy'' is conserved:
\begin{equation}
\frac{r^{d-1}f(r)}{\sqrt{-f(r) + \frac{\dot{r}^2}{f(r)}}} = \text{constant}.
\end{equation}
The final ($t = \infty$) slice is located at constant $r = r_f$, and this value can be found by maximizing $r^{2(d-1)}|f(r)|$ (as in \cite{Susskind:2014moa}).  While this codimension-one slice (wormhole) is connected for all time, the codimension-two surface of interest should undergo a transition from connected to disconnected.

We wish to consider not the total volume of a spatial slice, but the volume of the part of that slice enclosed by the RT surface.  In general, this means that we wish to consider the intersection of the entanglement wedge with the maximal slice that is anchored on the boundary subregion.  However, for a time-independent geometry, there is a preferred time slice; because the exterior is static, an RT surface outside the black hole (e.g., at late times) is straightforward to calculate. So our first task is to understand the RT surface at early and late times; all complexity measures will then follow.

It is convenient to think of $S^{d-1}$ as an $S^{d-2}$ that is fibered over an interval and that shrinks to zero size at the endpoints.  Concretely, we can write the metric of the unit $S^{d-1}$ as
\begin{equation}
d\Omega_{d-1}^2 = d\theta^2 + \sin^2\theta\, d\Omega_{d-2}^2
\label{spheremetric}
\end{equation}
where $\theta\in [0, \pi]$.  Consider the spatial subregion $\theta\in [0, \theta_0]$ for some fixed $\theta_0$.  This subregion is a ball of variable size; its boundary extends along the $S^{d-2}$ at $\theta = \theta_0$.  We consider both static and growing RT (area-extremizing) surfaces anchored to this boundary, which are rotationally symmetric with respect to the $S^{d-2}$.  Spherical symmetry guarantees that the transverse slices are spheres, but the size of the transverse sphere may vary with $r$.

\subsection{Area Comparison}

The transition in subregion complexity is driven by a transition in the RT surface.  To determine the correct RT surface at any given instant in time, we need to carefully compare the areas of the two candidate surfaces.  For a given subregion size, we numerically estimate the transition time at which the area of the (growing) connected RT surface overtakes the area of the (static) disconnected RT surface.  This dictates how the volume evolves with time:
\begin{itemize}
\item If the subregion size is less than half the total size, then the volume enclosed stops growing and saturates to the volume enclosed by the disconnected caps.
\item If the subregion size is greater than half the total size, then the volume enclosed keeps growing as the total volume of the spatial slice minus the volume enclosed by the disconnected caps.
\end{itemize}
We focus on comparing the (divergent) areas to identify the minimal-area surface.  See \autoref{app:holodetails} for details.

The areas involve integrating to the boundary of AdS and need to be regulated.  To compare the areas of the dynamic (HM) and static solutions and thereby obtain the transition time as a function of $\theta_0$, it suffices to integrate up to a cutoff on the radius: the divergences cancel in the difference of the two areas.  The difference between the two areas should approach a finite value as the cutoff is taken to infinity.  Assigning a finite and cutoff-independent value to the area of each surface individually would require holographic renormalization, or subtracting local counterterms depending on the induced metric but not on the extrinsic curvature of some cutoff surface \cite{deHaro:2000vlm, Graham:1999pm}.

We parametrize the static surface, which lives at constant $t$, by a function $\theta(r)$ that asymptotes to $\theta_0$ as $r\to\infty$.  We parametrize the growing surface by two functions $t(r)$ and $\theta(r)$, where the latter approaches $\theta_0$ and the former approaches a fixed boundary time $t_b$ as $r\to\infty$.  Each surface is determined by minimizing the corresponding area functional that follows from \eqref{lormetric} and \eqref{spheremetric}.  The static surface and the growing surface coincide for asymptotically large $r$, as required for the difference between their areas to be well-defined as the radial cutoff goes to infinity.  More precisely, we can expand $t(r)$ and $\theta(r)$ for the growing surface as power series in $1/r$ near the boundary ($r\to\infty$).  We expect on general grounds \cite{Graham:1999pm}, and can verify through the equations of motion at leading order in $r$, that both series start at order $1/r^2$:
\begin{equation}
t(r) = t_b + O(r^{-2}), \qquad \theta(r) = \theta_0 + O(r^{-2}).
\end{equation}
Hence $t'(r) = O(r^{-3})$ and $\theta'(r) = O(r^{-3})$, which implies that the integrands of the static area functional \eqref{areafunctionalstaticr} and the growing area functional \eqref{areafunctionalgrowingr} have the same asymptotic behavior as $r\to\infty$.  Note that for $d\geq 2$, we have
\begin{equation}
f(r) = \frac{r^2}{\ell^2} + O(1).
\end{equation}
The static surface is described by the special case where $t(r)$ is constant.

For the sake of numerics, we find it more convenient to impose boundary conditions at finite $r$ than at $r = \infty$.  For the static surface, this amounts to specifying $r_0$, the value of $r$ at the innermost point $\theta = 0$.  For the growing surface, this amounts to specifying $r_s$ and $\theta_s$, the values of $r$ and $\theta$ at the symmetric point $t = 0$.

In principle, we can determine how the difference between the HM surface area and the static surface area behaves as a function of both the boundary time $t_b$ and the boundary angle $\theta_0$.  This requires numerically scanning over many values of $(r_s, \theta_s)$ to obtain the corresponding $t_b$, $\theta_0$, and HM surface area.  One can empirically determine the boundaries of the parameter space of $(r_s, \theta_s)$ (i.e., the space of allowed values) by scanning over $r_s$ starting from slightly below $r_h$, where each value of $r_s$ has a corresponding range of $\theta_s$.  Given $\theta_0$, we can compute $r_0$ for the corresponding static surface as well as the corresponding RT surface area.

In practice, since we lack an analytical method to fix the value of $\theta_0$ as $r_s$ and $\theta_s$ are varied, computing the full time dependence of the area difference for a given $\theta_0$ requires sampling many data points and taking a one-dimensional slice through a two-dimensional dataset.  It is far more feasible simply to compute the transition time (the time at which the area difference vanishes) as a function of $\theta_0$.

The angle $\theta_0\in [0, \pi]$ determines the subsystem size $p$ (a fraction $0 < p < 1$), where $\theta_0 = \pi/2$ corresponds to $p = 1/2$.  We consider primarily $\theta_0$ between 0 and $\pi/2$.  In summary, holography suggests that there exists a critical subsystem size $p_\text{crit} < 1/2$ such that:
\begin{itemize}
\item For $0 < p < p_\text{crit}$, the complexity saturates essentially instantaneously.
\item For $p_\text{crit} < p < 1/2$, the complexity grows linearly and then saturates after a short time that grows sublinearly with $p$ (in particular, as a concave function of $p$).
\item For $1/2 < p < 1$, the complexity grows linearly for an exponentially long time.  However, as indicated in \autoref{figtable}, both the rate of the linear growth and the value of the complexity jump discontinuously for $1/2 < p < 1 - p_\text{crit}$.
\end{itemize}
The existence of a critical fraction $p_\text{crit}$ for which the transition time vanishes is clear from geometry: smaller subregions correspond to smaller static surfaces, so sufficiently small subregions have a static surface with smaller area than the growing HM surface for all times.  The above conclusions are subject to $1/N$ corrections that are exponentially small in the system size.  See \autoref{fig:saturationvsp}.

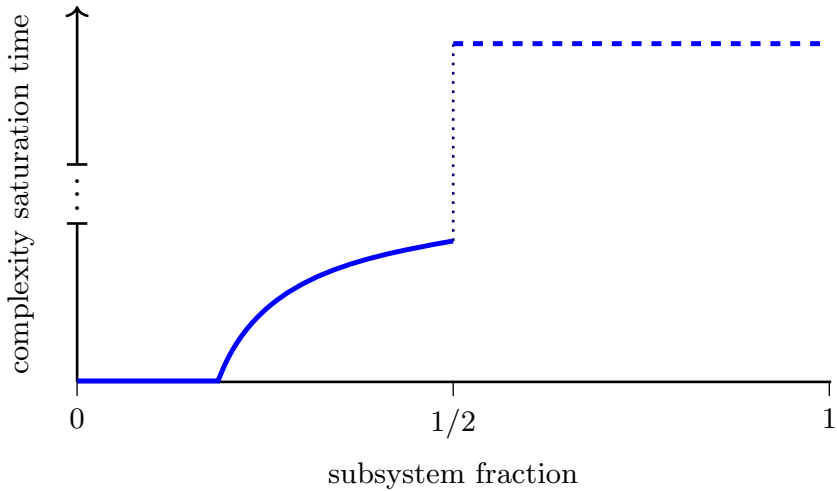
\begin{figure}[!htb]
\centering
\definecolor{db}{rgb}{0,0,0.4}
\definecolor{mb}{rgb}{0.37,0.51,0.71}
\hspace{-1 cm}
\scalebox{1.25}
{
\begin{tikzpicture}
    \draw[thick] (0,0) -- (8.02,0);
    \draw[thick,-|] (0,0) -- (0,1.7);
    \draw[thick,|->] (0,2.3) -- (0,4);
    \node at (0,2.1) {$\vdots$};
    \node at (4,-1) {{\fn subsystem fraction}};
    \node[rotate=90] at (-0.6,2) {{\fn complexity saturation time}};
    \draw[blue,line width=1.5pt] (0,0.01) -- (1.5,0.01) to[in=190,out=70] (4,1.5);
    \draw[thick,dotted,color=db] (4,1.5) -- (4,3.6);
    \draw[line width=1.5pt,dashed,blue] (4,3.6) -- (8,3.6);
    \foreach \x/\l in {0/0,4/{1/2},8/1}
    {\draw (\x,0) -- (\x,-0.15);
    \node[anchor=north] at (\x,-0.15) {\fn \l};}
\end{tikzpicture}
}
\caption{Schematic dependence of saturation time on subsystem fraction $p$ for the entire range of $p$.  The dashed top line should be understood as an approximate expectation, and merely indicates an exponential separation between the saturation times for $p < 1/2$ and $p > 1/2$ (the saturation time should in fact continue to increase with $p$ for $p > 1/2$, which can be understood more precisely on the quantum information side).}
\label{fig:saturationvsp}
\end{figure}

Importantly, the value of $p_\text{crit}$ depends on $\beta$ (see \autoref{fig:saturationvstemp}).  It would be interesting to understand whether the curves at finite $\beta$ reflect complexity dynamics in random quantum circuits with conserved quantities (e.g., those comprised of gates with continuous symmetries).  We leave this conjecture for the quantum information community.

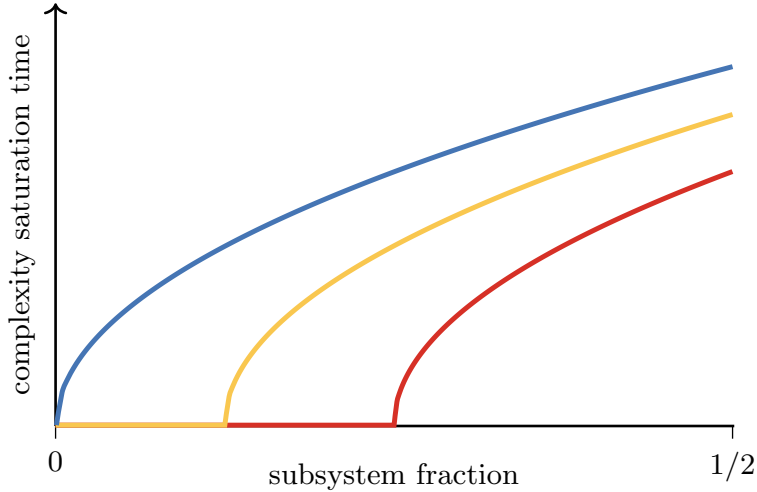
\begin{figure}[!htb]
\centering
\definecolor{lt}{rgb}{0.84, 0.19, 0.15}
\definecolor{mt}{rgb}{0.976, 0.780, 0.310}
\definecolor{ht}{rgb}{0.271, 0.459, 0.706}
\scalebox{1.25}
{
\begin{tikzpicture}[scale=0.9]
    \draw[thick] (0,0) -- (8.02,0);
    \draw[thick,->] (0,0) -- (0,5);
    \node at (4,-0.6) {{\fn subsystem fraction}};
    \node[rotate=90] at (-0.4,2.5) {{\fn complexity saturation time}};
    \foreach \x/\l in {0/0,8/{1/2}}
    {\draw (\x,0) -- (\x,-0.15);
    \node[anchor=north] at (\x,-0.15) {\fn \l};}
    \draw[line width=1.5pt,lt,smooth] (0,0.01) -- plot[domain=4:8,samples=120] (\x,{1.50*sqrt(\x-4)+0.01});
    \draw[line width=1.5pt,mt,smooth] (0,0.01) -- plot[domain=2:8,samples=120] (\x,{1.50*sqrt(\x-2)+0.01});
    \draw[line width=1.5pt,ht,smooth] plot[domain=0.01:8,samples=120] (\x,{1.50*sqrt(\x-0.01)+0.01});
\end{tikzpicture}
}
\caption{Schematic dependence of saturation time on temperature.  The different curves correspond to infinite temperature (blue) and successively lower finite temperatures (orange, red).  The blue curve ($\beta = 0$) is well-understood in both random quantum circuits and holography: it describes a $p^{1/(d - 1)}$ power law and starts at $p = 0$ (for instance, in 1D RQCs with $d = 2$, this curve would be linear).  The other curves have finite $p_\text{crit}$, and their shape is conjectural.  They may reflect the behavior of, e.g., $U(1)$ charge-conserving random circuits.}
\label{fig:saturationvstemp}
\end{figure}

We now describe our results in more (analytical and numerical) detail.

\subsubsection{Analytics}

When the radius of the boundary sphere is large compared to $\beta$, one can use simple dimensional analysis to estimate the following:
\begin{itemize}
\item The dependence of the critical subsystem size on the black hole parameters $\ell$ and $\mu$, which determine the size and the temperature of the dual quantum system (see \eqref{predictioncritical}).
\item The form of the saturation time as a function of subsystem size $p$ in the quickly saturating regime ($p < 1/2$), which is a simple fractional power law (see \eqref{predictionsaturation}).
\end{itemize}
Recall that we are interested in the AdS$_{d+1}$-Schwarzschild black hole with $d\geq 2$ and asymptotic boundary $S^{d-1}_\ell \linebreak[1] \times \linebreak[1] \mathbb{R}$:
\begin{equation}
ds^2 = -f(r)\, dt^2 + \frac{dr^2}{f(r)} + r^2\, d\Omega_{d-1}^2, \qquad f(r) = 1 + \frac{r^2}{\ell^2} - \frac{\mu}{r^{d-2}}.
\end{equation}
In terms of $\ell, \mu$, the horizon radius $r_h$ is the outermost solution to $f(r_h) = 0$, and
\begin{equation}
\beta = \frac{4\pi\ell^2 r_h}{dr_h^2 + (d - 2)\ell^2}.
\end{equation}
The boundary metric is not fixed by the bulk metric: we have the freedom to multiply it by a non-unique \emph{defining function} \cite{Witten:1998zw, Graham:1999pm} with a double zero at infinity.  A sensible choice (and the one that we make) is $\ell^2/r^2$.  The ``number of qubits'' in the dual quantum system equals the volume of the boundary sphere $S^{d-1}_\ell$ in Planck units, which, for fixed $G_N$, is proportional to $\ell^{d-1}$. (Since $G_N$ is an overall prefactor in the action, it does not appear in the equations of motion and is independent of the parameters of the solution.  Since $G_N\propto \ell_p^{d-1}$, the number of qubits is proportional to $\ell^{d-1}/G_N$, or simply $\ell^{d-1}$.  Hence Newton's constant $G_N$ gives the local ``qubit density.'') Writing the metric of the unit sphere $S^{d-1}_1$ as before,
\begin{equation}
d\Omega_{d-1}^2 = d\theta^2 + \sin^2\theta\, d\Omega_{d-2}^2
\end{equation}
where $\theta\in [0, \pi]$, we consider the ball given by $\theta\in [0, \theta_0]$ for some fixed $\theta_0$.  There are thus three length scales in the problem:
\begin{itemize}
\item $\ell$, the curvature radius of AdS (which, via our choice of defining function, equals the radius of the boundary sphere).
\item $L$, the linear size of the subsystem (which equals $\ell$ times a dimensionless function of $\theta_0$).
\item $\beta$, the temperature of the black hole.
\end{itemize}
Only two dimensionless ratios of these scales are meaningful, as our theory is conformal:
\begin{itemize}
\item The dimensionless ratio $\beta/\ell$ (or, alternatively, $\mu^{1/(d - 2)}/\ell$) is solely a property of the black hole and unrelated to the subsystem size.  Indeed, setting $r = \ell r_0$ and $t = \ell t_0$, the metric becomes
\begin{equation}
ds^2 = \ell^2\left[-f(r_0)\, dt_0^2 + \frac{dr_0^2}{f(r_0)} + r_0^2\, d\Omega_{d-1}^2\right], \qquad f(r_0) = 1 + r_0^2 - \frac{\mu/\ell^{d-2}}{r_0^{d-2}},
\end{equation}
which (up to an overall factor of $\ell^2$) is a function of $\mu^{1/(d - 2)}/\ell$.
\item The dimensionless ratio $L/\ell$ characterizes the subsystem size.
\end{itemize}
In the limit of small $\beta/\ell$ (high temperature and large volume), the black hole solution becomes planar.  The Hawking-Page transition (see \autoref{app:hawkingpage} for a review) occurs at the critical value $\beta/\ell = \frac{2\pi}{d - 1}$; for $\beta/\ell$ above the critical value, the black hole is no longer thermodynamically relevant.  For any fixed $\beta/\ell\in (0, \frac{2\pi}{d - 1})$ (between the flat-space limit and the Hawking-Page threshold), we aim to determine the saturation time of complexity as a function of $L/\ell$ and, in particular, the critical value of $L/\ell$ (the subsystem fraction at which the saturation time vanishes).

To begin, we consider the planar black hole of \cite{Hartman:2013qma}, which has one fewer length scale and hence only one dimensionless ratio, $\beta/L$.  Following \cite{Hartman:2013qma}, consider the union of two identical subregions $R$ on each asymptotic boundary of a planar two-sided black hole, and let $R$ (whose boundary is assumed to be sufficiently regular) have linear size $L$:
\begin{itemize}
\item If $L\gg \beta$, then the minimal-area extremal surface crosses the wormhole at early times, and its area grows linearly with boundary time $t_b$ with a coefficient proportional to $\operatorname{vol}(\partial R)$.
\item After a time $t_b\sim L$ (the thermalization time), the minimal-area extremal surface becomes the static surface that does not cross the horizon.  More precisely, the thermalization time is $L/2$ (from considering light cones propagating inward from $\partial R$).
\end{itemize}
In the case of the planar black hole, one can tune the subsystem size $L$ arbitrarily.  In our case, however, the size of $L$ is limited by $\ell$, so one needs a large black hole to achieve a parametric separation between $L$ and $\beta$.  We focus on the regime of small $\beta/\ell$ by taking $\mu$ large, so that
\begin{equation}
r_h\approx (\mu\ell^2)^{1/d}, \qquad \beta\approx \frac{4\pi\ell^2}{d(\mu\ell^2)^{1/d}}.
\end{equation}
Since $r_h = \ell$ is the critical value corresponding to the Hawking-Page transition, we require that $r_h\gg \ell$, which is equivalent to $\mu\gg \ell^{d-2}$.  For small $\beta/\ell$, the critical angle should still occur when $L\sim \beta$ (as in \cite{Hartman:2013qma}), i.e., when
\begin{equation}
\frac{L}{\ell}\sim \frac{\beta}{\ell}.
\end{equation}
As $\beta/\ell$ increases, this relation may be corrected to $L/\ell\sim f(\beta/\ell)$ for some function $f$.

The relevant formulas for the spherical black hole are as follows.  The volume of the spatial subregion of $S^{d-1}_\ell$ corresponding to the boundary angle $\theta_0$ is
\begin{equation}
\omega_{d-2}\ell^{d-1}\int_0^{\theta_0} d\theta\, \sin^{d-2}\theta,
\end{equation}
so the fraction of the entire system that the subsystem occupies is
\begin{equation}
p(\theta_0) = \frac{\omega_{d-2}}{\omega_{d-1}}\int_0^{\theta_0} d\theta\, \sin^{d-2}\theta = \frac{\Gamma(\frac{d}{2})}{\pi^{1/2}\Gamma(\frac{d - 1}{2})}\int_0^{\theta_0} d\theta\, \sin^{d-2}\theta.
\end{equation}
Therefore, in our case, $R$ is a $B^{d-1}$ with volume
\begin{equation}
\operatorname{vol}(R) = p(\theta_0)\omega_{d-1}\ell^{d-1},
\end{equation}
and $\partial R$ is an $S^{d-2}$ of radius $\sin\theta_0$ and volume
\begin{equation}
\operatorname{vol}(\partial R) = \omega_{d-2}\sin^{d-2}\theta_0.
\end{equation}
We restrict our attention to the range $p < 1/2$, for which the complexity saturates at the thermalization time.  The linear size of $R$ goes like
\begin{equation}
\boxed{L\sim p(\theta_0)^{\frac{1}{d - 1}}\ell,}
\label{predictionsaturation}
\end{equation}
which, as a purely geometrical quantity, is independent of $G_N$ and hence of the large local qubit density. From \eqref{predictionsaturation}, we see that for $L\gg \beta$, holography predicts that the complexity saturates at a time that grows as $p^{1/(d - 1)}$.  In particular, regardless of the subsystem fraction $p$ (provided that $p < 1/2$), the complexity will have saturated by a time of order
\begin{equation}
\ell\sim (\text{number of qubits})^{\frac{1}{d - 1}},
\end{equation}
which is polynomial (indeed, sublinear) in the number of qubits and depends on the dimensionality of the arrangement of qubits.  By contrast, for $p > 1/2$, the complexity saturates after a time \emph{exponential} in the number of qubits.

On the other hand, for $L\lesssim \beta$, the static surface always dominates and holography predicts that thermalization will occur essentially instantaneously.  We see that $L\sim \beta$ when
\begin{equation}
\boxed{p(\theta_0)\sim \left(\frac{\beta}{\ell}\right)^{d-1}.}
\label{predictioncritical}
\end{equation}
This relation defines the critical angle (below which complexity saturates ``instantaneously''), which \emph{decreases} with temperature.  Although one might expect the complexity to saturate more quickly as the temperature increases, this relation suggests that the threshold for saturation may also increase with temperature.

\subsubsection{Numerics}

Our numerical analysis allows us to go beyond the planar limit, i.e., the predictions of \cite{Hartman:2013qma}.  We consider $d = 2, 3, 4$, with a focus on $d = 3$ (the smallest ``non-special'' dimension).

For numerics, we set $\ell = 1$ without loss of generality.  Note that if $1/\ell^2 = \mu - 1$, then $r_h = 1$ for any $d$.  Thus setting $\ell = 1$ and $\mu = 2$ yields a black hole at the Hawking-Page threshold for any $d$; this is the smallest large black hole.

\begin{figure}[!htb]
\centering
\includegraphics[width=0.7\textwidth]{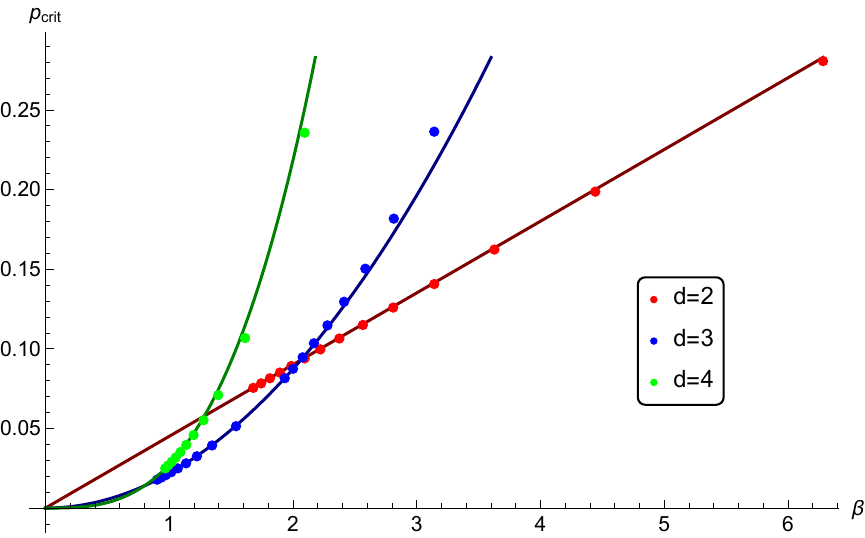}
\caption{$p_\text{crit}$ versus $\beta$ for $d = 2, 3, 4$ with fits (in darker colors) to linear, quadratic, and cubic power laws, respectively.  Deviations from power-law behavior are visible at large $\beta$ (small $\mu$): the data points are higher for odd $d$ and lower for even $d$.}
\label{fig:criticalpvsbeta}
\end{figure}

We examine $p_\text{crit}$ versus $\beta$ in \autoref{fig:criticalpvsbeta}.  To compute $p_\text{crit}$ for a given $\beta$, we set $r_s$ to be slightly below $r_h$ (which results in a slight overestimate for $p_\text{crit}$) and scan over $\theta_s$.  We also impose that the absolute value of the area difference be less than some tolerance:
\begin{itemize}
\item For $d = 2$, we consider $\mu$ from 2 to 15 in steps of 1 ($r_h$ from 1 to approximately 3.7) and use a tolerance of 0.0002.
\item For $d = 3$, we consider $\mu$ from 2 to 10 in steps of 1 and from 10 to 100 in steps of 10 ($r_h$ from 1 to approximately 4.6) and use a tolerance of 0.005.
\item For $d = 4$, we consider $\mu$ from 2 to 100 in steps of roughly 10 ($r_h$ from 1 to approximately 3.1) and use a tolerance of 0.005.
\end{itemize}
The numerical routine for the static RT surface breaks down below a minimum $\theta_0$, which makes going to large $\mu$ difficult because the critical angle at large $\mu$ might fall below that minimum value.  We check that $p_\text{crit}\sim \beta^{d-1}$ in all cases.  Deviations from this power law for small $\mu$ are noticeable when $d = 3$ and $d = 4$, but hardly so when $d = 2$.

We examine the $p$-dependence of the saturation time as follows:
\begin{itemize}
\item In \autoref{fig:saturationvspford234}, we plot the transition time versus subsystem size for black holes at the Hawking-Page threshold ($\ell = 1$ and $\mu = 2$) for $d = 2, 3, 4$.  We find that $p_\text{crit}$ becomes smaller and that the curve rises more slowly as $d$ increases.
\item For $d = 3$, we also consider larger values of $\mu$ (\autoref{fig:saturationvspford3}).  For $\mu = 2$, we can explore a large range of $r_s$, but for $\mu = 10$ and $\mu = 100$, we cannot take $r_s$ far below $r_h$ numerically.
\end{itemize}
For $p$ near $p_\text{crit}$, the curves are not simple power laws.  While the saturation time should go like $p^{1/(d - 1)}$ for $p\gg p_\text{crit}$, this regime is not accessible numerically.  On one hand, since $p < 1/2$, seeing this power-law behavior would require making $p_\text{crit}$ small, which requires taking $\mu$ large.  On the other hand, taking $p$ large is numerically challenging for large $\mu$.

\begin{figure}[!htb]
\centering
\includegraphics[width=0.7\textwidth]{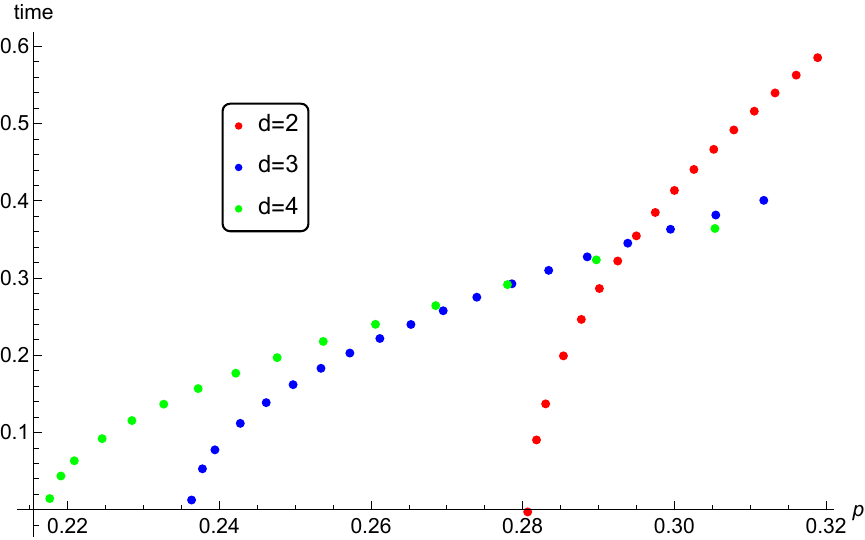}
\caption{Saturation time versus subsystem fraction $p$ for black holes at the Hawking-Page threshold and $p < 1/2$.}
\label{fig:saturationvspford234}
\end{figure}

\begin{figure}[!htb]
\centering
\includegraphics[width=0.7\textwidth]{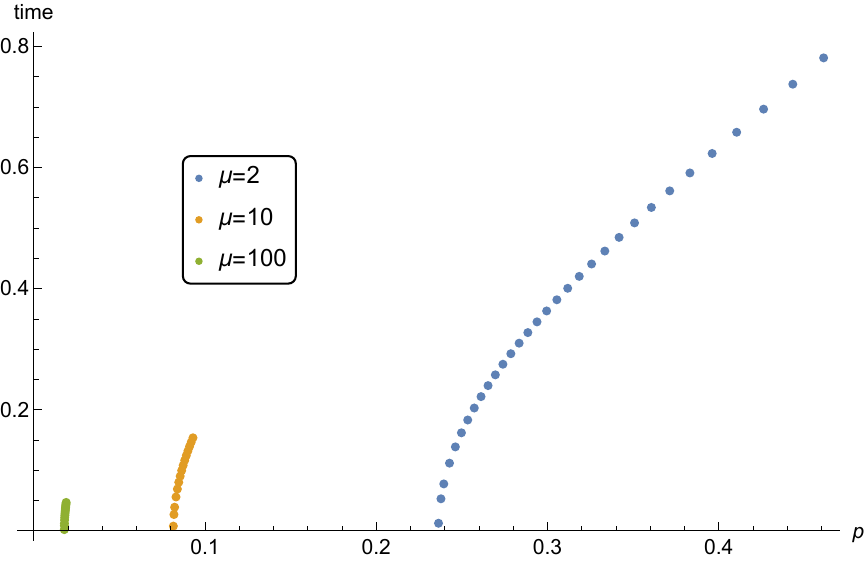}
\caption{Saturation time versus subsystem size for $d = 3$, $\ell = 1$, and $\mu = 2, 10, 100$.  Larger $\mu$ corresponds to higher temperature (smaller $\beta$).  Below $\mu = 2$ (the Hawking-Page threshold), the black hole is no longer a good model for thermalization in the boundary theory (nonetheless, it remains interesting to ask about the behavior of the boundary theory at such low temperatures).}
\label{fig:saturationvspford3}
\end{figure}

\subsection{Volume Transition}

The area (complexity) transition time is a universal and robust prediction from holography, in contrast to the precise values of the complexity itself.  Unlike entanglement entropy, which is quantified by area, the complexity (for which we use enclosed volume as a proxy) changes discontinuously at the transition.  In the regime $p_\text{crit} < p < 1/2$, the complexity grows linearly and then drops down to a small constant value at the transition.

Let us refer to the growing surface as the HM surface and to the static surface as the RT surface.  The HM surface approaches the RT surface asymptotically, while the RT surface caps off outside the horizon.  Assuming that the RT surface is always bounded by the HM surface, the RT surface will always enclose a smaller volume.  Then the constant value of the ``sawtooth'' (the rightmost panel in the second row of \autoref{figtable}) will be lower than the peak of the linear rise.

At the transition, we would expect any parametric separation between the growing and static volumes to be at most polynomial in $\ell$, not exponential.  The precise multiplicative constant depends on the definition of complexity (this is true in both quantum information theory and in holography).  Based on ``complexity = volume'' and the aforementioned inclusion property of volumes, we surmise that most reasonable complexity measures will exhibit a drop.  We summarize our expectations in \autoref{figtable}.

\subsection{Interpretation in Quantum Mechanics}

So far, we have discussed geometrizations of complexity, but we have yet to define complexity for subsystems of ordinary quantum systems.  A natural way to define the complexity (the ``purification complexity'') of a mixed state is by minimizing the unitary circuit complexity over all possible purifications.  By this definition, the complexity of a subsystem is upper-bounded by the complexity of the whole system.  If the size of a subsystem is at least half the size of the whole system, then the most ``efficient'' purification (the one that requires the fewest ancillas) is generically the state of the whole system.  On the other hand, if the size of a subsystem is less than half the size of the whole system, then the most efficient purification is generically a state on twice as many qubits.

As an alternative to optimizing over purifications, focusing on a particular purification with a fixed number of ancillas yields a notion of subsystem complexity that is easier to compute in practice.  Consider a subsystem $A$ and an auxiliary system $B$, with total Hilbert space $\mathcal{H}_A\otimes \mathcal{H}_B$.  A general state of this bipartite system is $\rho_{AB}$.  Given a density matrix $\rho_A = \tr_B\rho_{AB}\in \operatorname{End}(\mathcal{H}_A)$, its \emph{canonical purification} is given by the pure state $|\sqrt{\rho_A}\rangle\in \mathcal{H}_A\otimes \mathcal{H}_A^\ast$, which satisfies
\begin{equation}
\tr_{A^\ast}|\sqrt{\rho_A}\rangle\langle\sqrt{\rho_A}| = \rho_A.
\end{equation}
If $\rho_A = \sum_i \lambda_i|i\rangle\langle i|$, then we have the Schmidt decomposition $|\sqrt{\rho_A}\rangle = \sum_i \sqrt{\lambda_i}|i\rangle|i^\ast\rangle$.  The complexity of the canonical purification in the doubled Hilbert space provides an upper bound on the purification complexity.  The canonical purification, of which the TFD state is an example, has a simple geometric realization \cite{Dutta:2019gen}.  To construct the spacetime dual to $|\sqrt{\rho_A}\rangle$, we glue the entanglement wedge of a boundary state $\rho_A$ to its CPT conjugate across a quantum extremal surface.

Our results turn out to be largely agnostic of which definition of complexity we use in quantum mechanics.

We have also discussed what temperature means geometrically (in the bulk of a two-sided black hole), but we have yet to discuss its precise meaning on the boundary.  What is the interpretation of $\beta$ for a generic many-body system?  We are considering a non-equilibrium state rather than a thermal state.  The two-sided black hole Hilbert space is $\mathcal{H}_\text{CFT}^{\otimes 2}$.  Our choice of Hamiltonian acting on the doubled system is one for which the TFD state evolves in time (the phases do not cancel), so it is not an equilibrium state of the entire system:
\begin{equation}
|\text{TFD}(t_b)\rangle = \frac{1}{\sqrt{Z(\beta)}}\sum_n e^{-\beta E_n/2 - 2iE_n t_b}|E_n\rangle\otimes |E_n\rangle.
\end{equation}
However, the thermal state that we get by tracing out one side is still time-independent.  This time-dependent TFD state is a model for thermalization.\footnote{An alternative model is a quantum quench \cite{Chen:2018mcc, Auzzi:2019mah}.}  From the point of view of a time-dependent pure state in $\mathcal{H}_\text{CFT}$ undergoing thermalization, we interpret $\beta$ as the value for which the state at long times ``looks'' like a thermal state at inverse temperature $\beta$ (in the sense of reproducing the average energy of the corresponding Boltzmann distribution).\footnote{Of course, $\mathcal{H}_\text{CFT}$ is infinite-dimensional.  Consider the spectrum of the CFT quantized on $S^{d-1}_\ell$.  To get a finite-dimensional Hilbert space, we might cut off the CFT spectrum with an $\ell$-independent UV cutoff; as $\ell$ increases, the number of states below the cutoff increases because the level spacings decrease.}

We now comment on the implications of our results for generic quantum systems.  For convenience, we summarize the translation between discrete and continuous systems in \autoref{discretevscontinuous}.

\begin{table}[ht!]
    \centering
    \begin{tabular}{|l|l||l|l|}
        \hline
          \multicolumn{2}{|c||}{Discrete} & \multicolumn{2}{c|}{Continuous} \\
         \hline
         \hline
        $n$ & total number of qudits & $S$ & total entropy \\
        \hline
         $\q$ & local qudit dimension & $\sim S/\ell^{d-1}$ &local entropy density \\
         \hline
        $\na$, $\nb=n-\na$ & number of qudits in subsystem & $\sim S(L/\ell)^{d-1}$ & subsystem entropy \\
         \hline
    \end{tabular}
    \caption{Characteristic quantities in discrete and continuous systems. The $\sim$ symbol suppresses factors of $\pi$ associated with volumes of unit spheres and fractions thereof (they are displayed in the main text). Here, $\ell$ is the total system size and $L$ is the physical size of the subsystem. The subsystem dimensions are $\da =  q^{\na}$ and $\db = q^{\nb}$.}
    \label{discretevscontinuous}
\end{table}

As a simple check, we compare the complexity saturation time in holography and in RQCs for $p < 1/2$ at $\beta = 0$.  To define a local temperature, one needs a conserved quantity.  RQCs are local but do not conserve energy, unlike time-independent Hamiltonian evolution, so they equilibrate to infinite temperature ($\beta = 0$).  For $p < 1/2$, the saturation time in an RQC corresponds to the scrambling time, which we expect to be of order $n_A$ to the power of the inverse spatial dimension of the circuit geometry (by locality).  On the other hand, in holography, the scaling law \eqref{predictionsaturation} is exact at $\beta = 0$ and there is no critical subsystem size.  Since $p = n_A/(n_A + n_B)$ and $n_A + n_B\propto \ell^{d-1}$, the saturation time goes like
\begin{equation}
p^{1/(d - 1)}\ell\propto n_A^{1/(d - 1)},
\end{equation}
precisely as in RQCs \cite{harrow2023approximate}. Note that holography suggests that, for $n\gg 1$, the saturation time depends only on the fraction $p$.  This is a consequence of conformal symmetry.  Holography also allows one to quantify the behavior of the complexity before the saturation time.

Mixed states that are close to the maximally mixed state in trace distance have (purifications of) low complexity.  Under random evolution, we expect that the complexity of a subsystem that is less than half the size of the whole system saturates quickly to a small value because its reduced density matrix quickly approaches the maximally mixed state.  Holography suggests that for less-than-half subsystems, the complexity moreover exhibits a downward jump once the subsystem equilibrates.

\section{Subsystem Complexity Dynamics in Random Quantum Circuits} \label{sec:rqcs}

Thus far, we have performed classical gravity calculations to identify three sharp geometric ``phase transitions'' that map, via holography, to sharp transitions in the complexity of reduced states of a non-gravitational quantum field theory in one fewer dimension. We now imagine discretizing such a quantum field theory and representing it by $n$ qubits in a 1D layout that evolve in time by random quantum circuits with nearest-neighbor two-qubit gates. We use tools from quantum information to rigorously define a natural measure of state complexity for subsystems and prove the existence of sharp transitions (as in \autoref{figtable}) in the subsystem complexity of 1D random quantum circuits. Specifically, we prove one sharp transition as a function of the subsystem size and discuss evidence for another sharp transition in time.

The growth of circuit complexity, both for the evolved state and for the circuit itself, is well-understood in RQCs. Rigorous lower bounds on the circuit complexity \cite{brandao2021models, Haferkamp:2021uxo, chen2024incompressibility} have now established that complexity grows (essentially) linearly in time for a time exponential in the system size. Here, we study complexity growth and saturation for subsystems of states evolved by random quantum circuits. Our definition of subsystem complexity is the natural one: the number of gates required to construct a good approximation to the target state starting from a fixed fiducial state on a number of ancillas at most polynomial in the system size. We define this complexity measure more precisely in \autoref{def:complexity}. Before we give an overview of our results, we try to argue why we might expect a complexity transition in the first place. In the following, $O(f(x))$ means asymptotically upper-bounded by $f(x)$ and $\Omega(f(x))$ means asymptotically lower-bounded by $f(x)$.

Here is a simple argument for a complexity transition. First, we observe that the complexity of typical pure states in the Hilbert space \cite{knill1995approximation} and typical output states of $O(\exp(n))$-time random quantum circuits \cite{brandao2021models} is at least $\Omega(\exp(n))$. Next, consider the state $\rho_A$ on a subsystem of fraction $p$ ($0 < p < 1$) of an $n$-qubit Haar-random state. It is well-known that the typical trace distance between $\rho_A$ and the maximally mixed state on $pn$ qubits is given by \cite{lubkin78, lloyd1988, page1993}
\begin{equation}
\dist\left(\rho_A, \frac{\iden_A}{\da}\right)\leq \frac{1}{2}\sqrt{\da\tr\rho_A^2 - 1} \approx \frac{1}{2}\sqrt{\da \left(\frac{\da + \db}{\da\db + 1}\right) - 1} \approx 2^{(p - 1/2)n - 1},
\end{equation}
which is exponentially small in $n$ when $p < 1/2$. The first inequality follows from the generalized mean inequality, the second approximate equality follows from the average value of $\tr\rho_A^2$ for an $n$-qubit Haar-random pure state, and the third approximate equality follows from taking $n\gg 1$. The subsystem state being very close to maximally mixed implies that we may approximate it by preparing the maximally mixed state. Since we expect the maximally mixed state to be of low ($O(pn)$) complexity by any measure of complexity, the subsystem complexity is small for $p<1/2$. The behavior of the complexity for subsystems of Haar-random states at the two extreme sizes $p \linebreak[1] = \linebreak[1] 1$ and $p < 1/2$ suggests that it undergoes a transition with respect to subsystem size. Furthermore, the same bound on the trace distance---and the same conclusion---holds for $p<1/2$ subsystems of states evolved by unitaries drawn from an approximate unitary $2$-design, which is substantially weaker than requiring that the pure state be completely Haar-random.

However, this argument alone cannot be used to infer a \emph{sharp} transition in complexity. Should we expect there to be a sharp transition in any complexity or entanglement measure defined on a finite-dimensional quantum system? The authors of \cite{aubrun2014entanglement} prove that the separability of subsystems of typical states undergoes a sharp transition (from non-separable to separable) as the ratio of the subsystem size to the total system size varies from 1 to a fraction less than 1. In the case of qudits, this fraction is approximately $2/5$. This is an example of a sharp transition in finite-dimensional quantum systems. Perhaps, then, it is not unreasonable to expect a sharp transition in complexity (which can also be thought of as a multiparty entanglement measure) as a function of subsystem size fraction $p$.

We now state our results and depict them schematically in \autoref{fig:qi_expectation}. We prove that $n$-qubit states evolved by random quantum circuits (see \autoref{def:bwrqc} and \autoref{def:pwrqc}) in one dimension with periodic boundary conditions obey the following.
\begin{itemize}
    \item \autoref{fig:qi_expectation}~(a): for subsystems of size less than half, the complexity $\CC$ (\autoref{def:complexity}) grows in time and then decreases (\autoref{thm:complexity_growth_na_lt_nb}), never exceeding $O(\poly(n))$ complexity and equilibrating to $O(n_A)$ complexity (\autoref{thm:cotler2022fluctuations}, \autoref{cor:cotler2022fluctuations}, \autoref{thm:patchwork_purity}, \autoref{cor:patchwork_purity}).
    
    \item \autoref{fig:qi_expectation}~(b): for subsystems of size greater than half, the complexity grows linearly in time (\autoref{thm:complexity_growth_na_gt_nb}) and saturates at a value of $O(\exp(n))$.
\end{itemize}

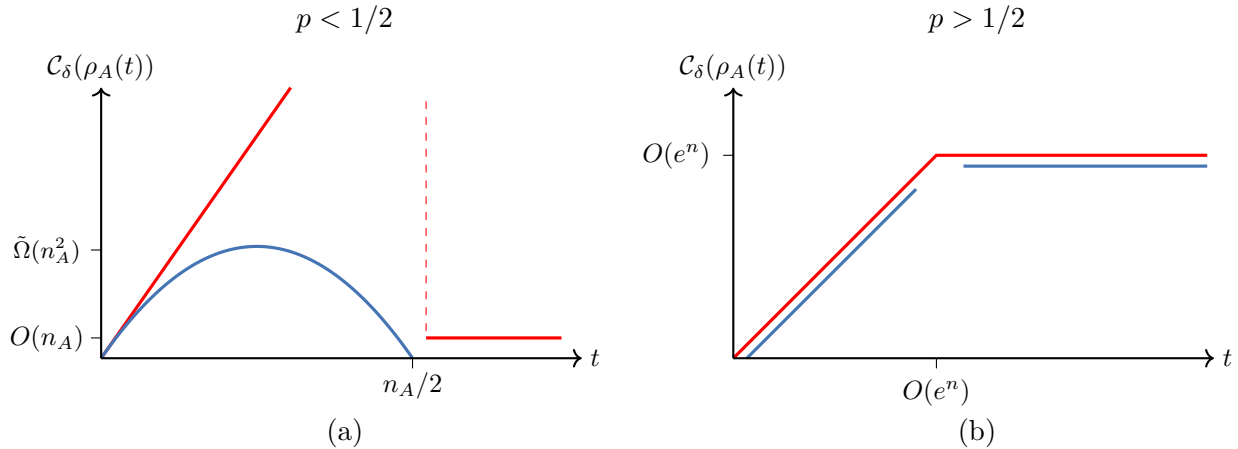
\begin{figure}[t!]
\definecolor{db}{rgb}{0.271, 0.459, 0.706}
\begin{tikzpicture}[scale=0.9,baseline=0mm]
\draw[red,very thick] (0,0.2) -- (4.75,4.4);
\draw[red,very thick] (4.75,0.3) -- (6.8,0.3);
\draw[red,dashed] (4.75,4.4) -- (4.75,0.3);
\draw[very thick,db] plot[domain=0:4.6, samples=201, variable=\t] (\t, {\t*(9.2 - 2*\t)/9});
\draw[thick,->] (0,0) -- (7.1,0);
\draw[thick,->] (0,0) -- (0,4);
\node at (7.3,0) {{\small $t$}};
\node at (0,4.3) {{\small $\CC(\rho_A(t))$}};
\draw (-0.15,0.3) -- (0,0.3);
\node at (-0.8,0.3) {\small $O(n_A)$};
\draw (-0.15,1.6) -- (0,1.6);
\node at (-0.8,1.6) {\small $\tilde \Omega(n_A^2)$};
\draw (4.6,-0.15) -- (4.6,0);
\node at (4.6,-0.4) {\small $n_A/2$};
\node at (3.6,5) {$p<1/2$};
\node at (3.6,-1.1) {(a)};
\end{tikzpicture}~~
\begin{tikzpicture}[scale=0.9,baseline=0mm]
\draw[red,very thick] (0,0) -- (3,3) -- (7,3);
\draw[db,very thick] (0.2,0) -- (2.7,2.5);
\draw[db,very thick] (3.4,2.84) -- (7,2.84);
\draw[thick,->] (0,0) -- (7.1,0);
\draw[thick,->] (0,0) -- (0,4);
\node at (7.3,0) {{\small $t$}};
\node at (0,4.3) {{\small $\CC(\rho_A(t))$}};
\node at (-0.84,3) {\small $O(e^n)$};
\draw (-0.15,3) -- (0,3);
\node at (3,-0.5) {\small $O(e^n)$};
\draw (3,-0.15) -- (3,0);
\node at (3.6,5) {$p>1/2$};
\node at (3.6,-1.1) {(b)};
\end{tikzpicture}

\caption{A schematic plot of our results for the subsystem complexity growth of random quantum circuits. On the left: complexity growth for $p<1/2$, less-than-half subsystem size ($\na < \nb$). On the right: complexity growth for $p>1/2$, greater-than-half subsystem size ($\na > \nb$). Red lines indicate circuit complexity {\it upper bounds}, whereas blue lines indicate the circuit complexity {\it lower bounds}. The colors in this plot should not be identified with the colors in \autoref{figtable}.}
\label{fig:qi_expectation}
\end{figure}

Taken together, our results evince a sharp transition in the saturation complexity of subsystems of random quantum circuits with respect to subsystem size fraction. This transition corresponds to that between the second and third rows in \autoref{figtable}. Below, we list our theorems informally and defer their rigorous formulations until \autoref{app:rqccomp}. First, we consider the subsystem complexity growth in the two regimes $\na > \nb$ (Informal~\autoref{thm:complexity_growth_na_gt_nb}) and $\na \linebreak[1] < \linebreak[1] \nb$ (Informal~\autoref{thm:complexity_growth_na_lt_nb}).  We write $U\sim \nu$ to indicate that $U$ is drawn from some distribution $\nu$ on the unitary group, and we let $\ket{\psi}$ denote an initially unentangled state.

\begin{itheorem}{Informal \autoref{thm:complexity_growth_na_gt_nb}}{Complexity growth for $\na > \nb$}
    Assume $A$ is a contiguous subsystem of a one-dimensional $n$-qubit system with periodic boundary conditions. The time-evolved states $\ssa(t) = \tr_{B}(U \ketbra{\psi} U^\dagger)$ of depth-$t$ random quantum circuits $U \sim \nu$ satisfy
    \begin{align}
        \CC(\ssa(t)) \gtrsim \frac{t}{\log(n) \log^7(t)}
    \end{align}
    with high probability over the choice of circuits $U \sim \nu$.
\end{itheorem}

\begin{itheorem}{Informal \autoref{thm:complexity_growth_na_lt_nb}}{Complexity growth for $\na < \nb$}
    Assume $A$ is a contiguous subsystem of a one-dimensional $n$-qubit system with periodic boundary conditions. The time-evolved states $\ssa(t) = \tr_{B}(U \ketbra{\psi} U^\dagger)$ of depth-$t$ random quantum circuits $U \sim \nu$ satisfy
    \begin{align}
        \CC(\ssa(t)) \gtrsim \frac{\na}{\log^2(n)  \log^7(t)} \left(t - \frac{2t^2}{\na}\right)
				\label{eq:lowerboundnalessnb}
    \end{align}
    with high probability over the choice of circuits $U \sim \nu$.
\end{itheorem}

\vspace{0.5\baselineskip}

At time $t = \na / 2$, the right-hand side of \autoref{eq:lowerboundnalessnb} goes to zero and the complexity lower bound becomes trivial, as depicted in \autoref{fig:qi_expectation}~(a). Furthermore, at time $t \approx \na / 2$, the late-time complexity upper bound tells us that the complexity remains small for all later times.

We note that a simple complexity upper bound for $\na > \nb$ is the complexity of the corresponding pure state, which is in turn upper-bounded by the smaller of: (1) the total number of gates in the circuit and (2) the upper bound on the number of gates required to construct an arbitrary unitary (Refs.~\cite{barenco1995elementary, knill1995approximation} combined with the Solovay-Kitaev theorem \cite{dawson2005solovaykitaev}). This complexity upper bound is shown in \autoref{fig:qi_expectation}~(b)  as a linearly growing red line, saturating at exponential complexity. We now provide a more detailed discussion of the complexity upper bounds for $\na < \nb$. At early times, we again have a growing upper bound simply from counting the number of gates in the light cone of the subsystem $A$. At later times, Informal \autoref{thm:cotler2022fluctuations} and \autoref{thm:patchwork_purity} address the closeness of subsystem states to the maximally mixed state.

\begin{itheorem}{Informal \autoref{thm:cotler2022fluctuations} and \autoref{thm:patchwork_purity}}{Distance to maximally mixed state for $\na < \nb$}
    Assume $A$ is a contiguous subsystem of a one-dimensional $n$-qudit system with periodic boundary conditions. For some $\err>0$, the time-evolved states $\ssa(t) = \tr_{B}(U \ketbra{\psi} U^\dagger)$ of depth-$t$ random quantum circuits $U \sim \nu$ are close to maximally mixed, i.e., they satisfy
    \begin{align}
        \Pr_{U \sim \nu}\Big( \dist\big(\ssa(t),\iden_A/\da\big) \geq \err\Big) \lesssim \frac{\da e^{-\lambda t}}{\err^2}\,,
    \end{align}
    where $\lambda$ is a constant.
\end{itheorem}

\begin{itheorem}{Informal \autoref{cor:cotler2022fluctuations} and \autoref{cor:patchwork_purity}}{Complexity saturation for $\na < \nb$}
    Assume $A$ is a contiguous subsystem of a one-dimensional $n$-qudit system with periodic boundary conditions. For some $\err>0$, the time-evolved states $\ssa(t) = \tr_{B}(U \ketbra{\psi} U^\dagger)$ of depth-$t = O(\na)$ random quantum circuits $U \sim \nu$ satisfy
    \begin{align}
        \CC(\ssa(t)) = \CC(\iden_A / \da)
    \end{align}
    with high probability.
\end{itheorem}

\vspace{0.5\baselineskip}

\autoref{cor:cotler2022fluctuations} and \autoref{cor:patchwork_purity} follow directly from \autoref{thm:cotler2022fluctuations} and \autoref{thm:patchwork_purity}, as the theorems conclude that states $\rho_A(t)$ are close to maximally mixed after some time. As we can prepare the maximally mixed state on $n_A$ qudits with $O(n_A)$ gates by making $n_A$ EPR pairs on $2n_A$ qudits and throwing away one qudit from each pair, it follows that random circuit states at late times have complexity at most $O(\na)$.

\autoref{thm:complexity_growth_na_gt_nb}, \autoref{cor:cotler2022fluctuations}, and \autoref{cor:patchwork_purity} demonstrate the same sharp complexity transition that we find in holography, specifically between the second and third rows of \autoref{figtable} as a function of the subsystem size. If the fraction is greater than $1/2$, then the complexity grows linearly for an exponential time, and if the fraction is below $1/2$, then the complexity saturates to that of the maximally mixed state.

In addition, the results illustrated in \autoref{fig:qi_expectation}~(a) indicate the possibility of a second sharp complexity transition with respect to the random quantum circuit time (instead of subsystem size fraction), similar to the sawtooth complexity drop in the rightmost column of \autoref{figtable}. In particular, for subsystems of size less than half, our results indicate the following: the subsystem complexity grows as a function of time to order quadratic in the subsystem size and then, at a time of order the subsystem size, the complexity drops to order linear in the subsystem size. This rise and fall is likely smooth for random quantum circuits on local qubits. But considering that a large number of degrees of freedom per site could conceivably capture more features of the holographic picture, we conjecture that this transition becomes sharp for random quantum circuits in the limit of large local qudit dimension $q\ra\infty$. If this sharp transition were proven, it would imply that subsystems of certain random quantum circuits thermalize instantly. Although our results cannot establish this sharp transition, the upper and lower bounds are consistent with this behavior. As an analogous sharp complexity transition in time is present in holography for less-than-half-sized subsystems, we conjecture that:
\begin{conjecture}
    \label{conj:second_transition}
    In the limit of large local dimension $q\ra\infty$, the complexity of an $\na$-qudit subsystem of an $n$-qudit random quantum circuit for $\na < n/2$ grows linearly to $\CC(\rho_A(t))=O(\na^2)$ at a time $t=O(\na)$, after which the complexity drops sharply to $\CC(\rho_A(t))=O(\na)$.
\end{conjecture}

We note that if the complexity lower bound in Informal \autoref{thm:complexity_growth_na_lt_nb} were such that the term $-2t^2/\na$ were absent at large local dimension for all times until order $O(\na)$, then \autoref{conj:second_transition} would follow. However, the $-2t^2/\na$ appears because we use the increasing rank of subsystems of random quantum circuits in our complexity lower bounds. Such terms may not be inherent to the problem and may only be an artifact of our proof technique. We leave this for future work.

\section{Discussion} \label{sec:discussion}

Using holographic arguments, we have identified three sharp transitions in the complexity dynamics of subsystems of quantum many-body systems, two in space and one in time:
\begin{enumerate}
\item[(1)] a transition in the long-time behavior of complexity as the subsystem size $p$ crosses the value $p = 1/2$, which corresponds to passing from the first two rows to the last two rows of \autoref{figtable};
\item[(2)] a transition in the early-time behavior of complexity as the subsystem size $p$ crosses the values $p = p_\text{crit}$ or $p = 1 - p_\text{crit}$, which corresponds to passing between the first two rows or between the last two rows of \autoref{figtable} (this is fundamentally a single transition between minimal-area extremal surfaces, which simply has different interpretations below and above $p = 1/2$);
\item[(3)] and a discontinuous transition in the complexity as a function of time, corresponding to the sawtooth profiles in the right column of \autoref{figtable} for $p_\text{crit} < p < 1/2$ and $1/2 < p < 1 - p_\text{crit}$.
\end{enumerate}
Using tools from quantum information, we have rigorously established the transition at $p = 1/2$ and provided evidence for the discontinuous transition in time.  We leave the transition at $p = p_\text{crit}$ as a conjecture for future work, which could foreseeably be resolved by studying the $\beta$-dependence of complexity in random quantum circuits with continuous symmetries.

Throughout our discussion, we have devoted more attention to the transition in $p$ at $p = p_\text{crit}$ and the transition in time for $p_\text{crit} < p < 1/2$ than to their counterparts at $p = 1 - p_\text{crit}$ and for $1/2 < p < 1 - p_\text{crit}$.  From the holographic point of view, this was simply a matter of convenience, as the transitions in space and time are equally manifest whether $p < 1/2$ or $p > 1/2$.  However, the transition in time for $1/2 < p < 1 - p_\text{crit}$ is far less clear in random quantum circuits than the equivalent one in time for $p_\text{crit} < p < 1/2$.  In RQCs, the time transition for $p < 1/2$ can be explained by the subsystem approaching the maximally mixed state and equilibrating to infinite temperature.  Whether there exists a physical justification for a sudden \emph{increase} in the value and growth rate of subsystem complexity for $p > 1/2$ remains to be seen. (Suggestively, our complexity lower bounds in \autoref{subsec:complexity_growth_na_gt_nb} depend on the dynamically growing rank of the subsystem of interest, which potentially captures a discontinuous change in slope when $p > 1/2$.) Finally, as already mentioned, the transitions at $p = p_\text{crit}$ and $p = 1 - p_\text{crit}$ are equally invisible in RQCs (for which $p_\text{crit} = 0$) but are presumably present in systems that equilibrate to finite temperature locally.

Looking ahead, a natural goal is to more fully understand the implications of our results for generic quantum systems.  This may entail achieving a more precise match between holographic and information-theoretic estimates for the saturation time, as well as understanding where this correspondence breaks down.  Various questions could be asked about the limitations of random quantum circuit models of black hole dynamics:
\begin{itemize}
\item Can they reproduce the expected scrambling behavior \cite{Chamon:2023lsv}?
\item Would an analogy between black holes and continuous random circuits be more appropriate?  Ref.~\cite{Magan:2024aet} argues in the affirmative by relaxing the requirement that the black hole evolve via a time-independent Hamiltonian.  But note that the opposite direction---keeping the time-independent Hamiltonian of the black hole while replacing the random quantum circuit model with a time-independent Hamiltonian model---would be closer in spirit to the original con\-jec\-ture of Brown and Susskind \cite{Brown:2017jil}.
\item Could classical gravity model quantum systems more reasonably than we have a right to expect?  Ref.~\cite{Milekhin:2024mce} uses circuit depth as a proxy for circuit complexity, which is justified for translationally invariant systems, and argues that the area of the connected HM surface of any subsystem (after suitable rescaling) provides a lower bound on the circuit complexity of the \emph{whole} system for all times, even beyond the time when the HM surface ceases to be the minimum-area extremal surface.  This bound is insensitive to the choice of subsystem and, therefore, to the complexity-theoretic significance of the HM transition.
\end{itemize}
It may be useful to go beyond the circuit model altogether, perhaps via random \cite{Hayden:2016cfa} or dynamical \cite{May:2016dgv, Osborne:2017woa} tensor networks.  For example, Ref.~\cite{Abt:2017pmf} gives a definition of subregion complexity for ten\-sor network states that reduces to counting tensors and examines how the holographic subregion complexity undergoes phase transitions as the boundary subregion varies in size and topology---i.e., phase transitions in space but not in time.

Other complexity measures offer alternatives to circuit complexity.  For example, what is the relation between subsystem complexity and complexity measures that assign costs only to multi-party gates \cite{Balasubramanian:2018hsu, Baiguera:2023bhm}?  Subsystem complexity is also meaningful in the context of evolution by a fixed Hamiltonian, e.g., Krylov complexity \cite{Parker:2018yvk},\footnote{We thank Moshe Rozali for this suggestion.} which can be studied in the broader setting of unitary circuit dynamics \cite{Suchsland:2023cmb, Sahu:2024urf}.\footnote{In fact, the narrowing of operator size under random unitary dynamics \cite{Schuster:2021uvg} suggests a sharp transition in Krylov complexity.}  However, the behavior of Krylov complexity in quantum field theory depends on the precise UV definition \cite{Kar:2021nbm, Avdoshkin:2022xuw}, and whether it is a good proxy for circuit or holographic complexity, particularly in relation to subsystems \cite{Das:2024zuu}, remains an open question.

This line of investigation opens the door to even broader questions.  How can we define the circuit complexity of mixed states without purification?  One option is to work with quantum channels directly \cite{Araiza:2023dlq}.  How can we generalize rigorous results about the growth of circuit complexity in random quantum circuits to Hamiltonian evolution?  The original Brown-Susskind conjecture concerns \emph{generic} Hamiltonians, i.e., chaotic many-body systems undergoing time-independent Hamiltonian evolution.  In the spirit of this conjecture, we would like a proof of circuit complexity growth in a time-independent Hamiltonian model \cite{Kotowski:2023ayh} rather than the time-dependent Markovian model exhibited by random quantum circuits.  Since time-independent Hamiltonians do not generate unitary designs, such a proof would require completely different techniques.  Finally, how can we use these Hamiltonian models to probe the dynamics of chaotic open quantum systems?

\subsubsection*{Acknowledgements}

We thank Scott Aaronson, Jeongwan Haah, Vedika Khemani, John Preskill, Daniel Ranard, Moshe Rozali, Thomas Schuster, and especially Juan Maldacena for helpful discussions. We also thank Jeongwan Haah and Douglas Stanford for coordinating arXiv submission of the related work \cite{Haah:2025hyf}. The work of YF was supported in part by NSF grant PHY-2210562. NHJ and SM acknowledge support in part from DOE grant DE-SC0025615. The work of AK was supported in part by DOE grant DE-SC0022021 and by a grant from the Simons Foundation (Grant 651678, AK). NHJ would like to thank the Simons Institute for the Theory of Computing for their hospitality during the completion of part of this work. Sandia National Laboratories is a multi-program laboratory managed and operated by National Technology and Engineering Solutions of Sandia, LLC, a wholly owned subsidiary of Honeywell International Inc., for the U.S. Department of Energy's National Nuclear Security Administration under contract DE-NA-0003525. All views and conclusions contained herein are those of the authors and should not be construed as representing the official views or policies of the U.S. Department of Energy or the U.S. Government.

\appendix

\section{Phases of Quantum Field Theory and Quantum Gravity} \label{app:hawkingpage}

For convenience, we recall some standard facts about how AdS/CFT relates phases of gravity to phases of quantum field theory \cite{Witten:1998zw}.

In quantum gravity, we define the thermal partition function as the path integral on a Euclidean manifold with the boundary condition that Euclidean time is a circle of proper size $\beta$:
\begin{equation}
Z(\beta) = \int Dg\, e^{-S_E[g]}, \qquad t_E\sim t_E + \beta, \qquad g_{tt}\to 1,
\end{equation}
as $r\to\infty$ (for pure gravity, we turn off all bulk fields aside from the metric).  Thermodynamic properties such as free energy, entropy, and energy are given by the standard formulas:
\begin{equation}
Z = e^{-\beta F}, \qquad S = (1 - \beta\partial_\beta)\log Z, \qquad E = -\partial_\beta\log Z.
\end{equation}
We approximate the path integral by expanding around classical solutions:
\begin{equation}
Z\approx \sum_{\bar{g}} e^{-S_E[\bar{g}]}.
\end{equation}
This approximation does not include semiclassical (one-loop) corrections.  $S_E$ contains a factor of $1/G_N$, and to leading order in $G_N$, the thermal free energy is the Euclidean on-shell action.

Thermal states in CFT are dual to black holes in AdS quantum gravity.  We impose boundary conditions appropriate for thermal field theory:
\begin{equation}
Z_\text{CFT}[M] = Z_\text{grav}[\partial = M], \qquad M = \Sigma_{d-1}\times S^1_\beta.
\end{equation}
We compute $Z_\text{grav}(\beta)$ classically as a sum over saddles and interpret the result in CFT; it gives a (de)con\-fine\-ment phase transition as a function of temperature.  To do so, we find all classical solutions that obey the thermal boundary condition \eqref{thermalBC} and evaluate their on-shell actions using the Einstein-Hilbert action of AdS$_{d+1}$.  There are three such solutions in pure gravity: small black holes, large black holes, and thermal AdS.  Large black holes have positive specific heat, while small black holes are unstable (like those in flat spacetime).  Thermal AdS has
\begin{equation}
f(r) = 1 + \frac{r^2}{\ell^2}
\end{equation}
(there is no horizon, and $\beta$ is a free parameter).  After summing over these allowed solutions, the thermodynamics in the field theory is determined by the solution with lowest free energy (the one that dominates the canonical ensemble at fixed $\beta$).  The large black hole always has lower free energy than the small black hole:
\begin{equation}
S_E^\text{small}(\beta)\geq S_E^\text{large}(\beta).
\end{equation}
Hence we compare the large black hole and thermal AdS:
\begin{align}
Z_\text{grav}(\beta) &\approx e^{-S_E^\text{large}} + e^{-S_E^\text{thermal}}, \\
\log Z_\text{grav}(\beta) &\approx \max(-S_E^\text{large}, -S_E^\text{thermal}).
\end{align}
The exponents are very large, of order $1/G_N$.  We find that the black hole dominates at large $r_h$ and that thermal AdS dominates at small $r_h$.  There is a first-order phase transition when the two solutions exchange dominance:
\begin{equation}
S_E^\text{large} = S_E^\text{thermal} \implies \beta_\text{crit} = \frac{2\pi\ell}{d - 1}, \text{ } r_h^\text{crit} = \ell.
\end{equation}
This is the Hawking-Page transition.

For the AdS black hole, since $\beta$ attains a maximum value as a function of $r_h$, the black hole only contributes to the thermodynamics if the temperature is sufficiently high ($\beta$ is sufficiently small).  The large-temperature limit can be obtained by taking $r_h\to 0$ or $r_h\to \infty$, but since the small black hole is thermodynamically disfavored, we must take $r_h\to \infty$ (corresponding to large $\mu$).  The low-temperature phase is thermal AdS; the high-temperature phase is the stable black hole.

The entropy $(1 - \beta\partial_\beta)\log Z$ of thermal AdS is zero at this order: that is, it is $O(G_N^0)$ due to quantum corrections.  Hence the full thermal entropy $S(\beta)$ is $O(G_N^0)$ at low temperatures and jumps to a large number $O(1/G_N)$ at $\beta_\text{crit}$.  The Hawking-Page transition implies that theories with a semiclassical gravity description must have a small number of states at low energy but a very large number at high energy, with a sharp transition in between.  This is reminiscent of confinement in nonabelian gauge theory.  In a confining $SU(N)$ gauge theory, $F = O(1)$ in the confined phase (the physical DOFs are color singlets) and $F = O(N^2)$ in the deconfined phase (the physical DOFs are gluons).  Hence the black hole corresponds to the deconfined phase, while thermal AdS corresponds to the confined phase.\footnote{Another geometric manifestation of this fact is the following.  In $SU(N)$ gauge theory, a Wilson loop wrapping the thermal circle (Polyakov loop) serves as an order parameter for the (de)confinement transition.  One can think of this temporal Wilson loop as a free quark: $\langle W\rangle = e^{-\beta F}$ is nonzero or zero depending on whether the quark has finite or infinite free energy.  One can also think of it as an order parameter for the spontaneous breaking of the $\mathbb{Z}_N$ center symmetry, which is a one-form symmetry \cite{Gaiotto:2014kfa} since the charged objects are lines.  Indeed, the holographic rule for computing $\langle W\rangle$ at large $N$ \cite{Maldacena:1998im} is to compute $e^{-S_\text{NG}}$ for a Euclidean string worldsheet ending on the contour of $W$.  Such a worldsheet exists when the thermal circle is contractible in the bulk, e.g., for the Euclidean black hole but not for thermal AdS.  This topological criterion captures the preservation or breaking of a one-form symmetry: a higher-form symmetry is spontaneously broken if its charged objects have nonzero VEV when they are large.  For example, a one-form symmetry in 4D is unbroken if its charged loop operators exhibit an area law and broken if they exhibit a perimeter or Coulomb law (because the latter two laws can be set to zero by a local counterterm, leading to a nonzero VEV).  In summary, $\langle W\rangle\neq 0$ indicates the breaking of center symmetry and hence deconfinement.}

One caveat is that in CFT, only the ratio $\ell/\beta$ matters, so high temperature means large $\ell$.  So the theory on $\mathbb{R}^3$ only has the black hole phase, and there is no Hawking-Page transition.  QCD has a (de)confinement transition in infinite volume, i.e., on $\mathbb{R}^3$; $\mathcal{N} = 4$ SYM does not, since it is a CFT.  Instead, $\mathcal{N} = 4$ SYM has a (de)confinement transition on $S^3$, but only in the $N\to\infty$ limit (kinematic confinement occurs at low energies due to the Gauss law).  Normally, phase transitions do not occur in finite volume since with finitely many DOFs, the free energy is analytic in $\beta$.

\section{Holography Details} \label{app:holodetails}

\subsection{Static RT Surface}

The late-time RT surface is static and described by a time-independent $r(\theta)$.  Fix a $t_b$; then the induced metric is
\begin{equation}
ds^2 = \left[r^2 + \frac{r'^2}{f(r)}\right]d\theta^2 + r^2\sin^2\theta\, d\Omega_{d-2}^2,
\end{equation}
and the area functional is
\begin{equation}
A[r(\theta)] = \omega_{d-2}\int d\theta\, (r\sin\theta)^{d-2}\sqrt{r^2 + \frac{r'^2}{f(r)}}.
\label{areafunctionalstatic}
\end{equation}
The integrand now has explicit $\theta$-dependence, so we are forced to solve the variational equation, which simplifies to
\begin{equation}
2(d - 2)(r'^2 + r^2 f(r))r'\cot\theta - 2(d - 1)r^3 f(r)^2 + 2rf(r)(rr'' - dr'^2) - r^2 r'^2 f'(r) = 0.
\end{equation}
This is a second-order nonlinear ODE.  We impose the boundary conditions $r(0) = r_0 > r_h$ and $r'(0) = 0$ at the innermost point of the extremal surface.  We make a few observations:
\begin{itemize}
\item One can show numerically that a solution for $r(\theta)$ exists for any $r_0 > r_h$.
\item The solution diverges at some value $\theta = \theta_0$ between 0 and $\pi$.  Up to a caveat that we address below, we can determine how $\theta_0$ changes as a function of $r_0$ (e.g., by extracting the location of the singularity as the upper limit on the domain of the interpolating function for the numerical solution).  This is easier than imposing the boundary condition $r(\theta_0) = \infty$ directly.
\item $\theta_0$ is a monotonically decreasing function of $r_0$ that ranges from $\pi$ to 0 as $r_0$ ranges from $r_h$ to $\infty$.  Therefore, a solution exists for any subregion size (determined by $\theta_0$).
\end{itemize}
Note that we expect to find two extremal surface solutions for any given $\theta_0$, with one of them having smaller area (except when $\theta_0 = \pi/2$, in which case both surfaces have the same area).  One of these solutions is the one described above; the other one is the reflected version of that solution for $\pi - \theta_0$.  For instance, for small $\theta_0$, the above solution lies near the boundary, while the other wraps nearly the entire horizon.

The solution described above, which has minimal area when $\theta_0 < \pi/2$, is always homologous to the boundary subregion defined by $[0, \theta_0]$.  So for a single-sided black hole, where the entanglement wedge of a subregion never contains the black hole, this solution is the one compatible with the homology constraint.  In the two-sided case, however, only the surface of smaller area is relevant.  This corresponds to the above solution when $\theta_0 < \pi/2$ and to the reflected solution when $\theta_0 > \pi/2$.  What changes at $\theta_0 = \pi/2$ is whether the entanglement wedge contains the black hole.  Therefore, we only need to consider the solutions with $\theta_0 < \pi/2$.

If one computes $\theta(r)$ instead of $r(\theta)$, then one can read off the asymptotic value $\theta = \theta_0$ for large $r$ rather than looking for a singularity.  To compute $\theta(r)$, we write the area functional as
\begin{equation}
A[\theta(r)] = \omega_{d-2}\int dr\, (r\sin\theta)^{d-2}\sqrt{\frac{1}{f(r)} + r^2\theta'^2}.
\label{areafunctionalstaticr}
\end{equation}
The corresponding variational equation is
\begin{equation}
2(d - 2)(1 + r^2 f(r)\theta'^2)\cot\theta - r(rf'(r) + 2df(r))\theta' - 2(d - 1)r^3 f(r)^2\theta'^3 - 2r^2 f(r)\theta'' = 0.
\end{equation}
We impose the initial condition $\theta'(r_0) = \infty$.

We now return to the caveat mentioned earlier.  When computing $\theta(r)$ numerically, we observe non-monotonicity for $d > 2$ (although not in the BTZ case).  Importantly, the non-monotonicity of $\theta(r)$ is an obstacle to solving for $r(\theta)$: it implies that $r(\theta)$ is a multivalued function, only one branch of which diverges at $\theta = \theta_0$.  To a numerical solver, it means that $r(\theta)$ breaks down not at $\theta_0$, but at the turnaround point.  Therefore, in practice, we use $\theta(r)$ to solve for $\theta_0$.

\subsection{Growing RT Surface}

Spherical symmetry guarantees that the growing RT (or HM) surface can be described by $r(t)$ and $\theta(t)$, or $t(r)$ and $\theta(r)$.  For any solution, $r$ ranges from some minimal radius $r_s < r_h$ to $\infty$.

In terms of $r(t)$ and $\theta(t)$, the induced metric is
\begin{equation}
ds^2 = \left[-f(r) + \frac{\dot{r}^2}{f(r)} + r^2\dot{\theta}^2\right]dt^2 + r^2\sin^2\theta\, d\Omega_{d-2}^2.
\end{equation}
The area functional is
\begin{equation}
A[r(t), \theta(t)] = \omega_{d-2}\int dt\, (r\sin\theta)^{d-2}\sqrt{-f(r) + \frac{\dot{r}^2}{f(r)} + r^2\dot{\theta}^2}.
\label{areafunctionalgrowingt}
\end{equation}
We could impose the boundary conditions
\begin{equation}
r(t_b) = \infty, \qquad \dot{r}(0) = 0, \qquad \theta(t_b) = \theta_0, \qquad \dot{\theta}(0) = 0.
\end{equation}
However, it is simpler to impose all conditions at $t = 0$:
\begin{equation}
r(0) = r_s, \qquad \dot{r}(0) = 0, \qquad \theta(0) = \theta_s, \qquad \dot{\theta}(0) = 0.
\end{equation}
This is as it should be: we specify two free constants as inputs to the ODEs that determine the asymptotic values $t_b$ and $\theta_0$ at $r = \infty$ (for this pair of second-order ODEs, the other two initial conditions are determined by symmetry).

In terms of $t(r)$ and $\theta(r)$, the induced metric is
\begin{equation}
ds^2 = \left[\frac{1}{f(r)} - f(r)t'^2 + r^2\theta'^2\right]dr^2 + r^2\sin^2\theta\, d\Omega_{d-2}^2,
\end{equation}
where primes denote $r$-derivatives.  The area functional is
\begin{equation}
A[t(r), \theta(r)] = \omega_{d-2}\int dr\, (r\sin\theta)^{d-2}\sqrt{\frac{1}{f(r)} - f(r)t'^2 + r^2\theta'^2}.
\label{areafunctionalgrowingr}
\end{equation}
We could impose the boundary conditions
\begin{equation}
t(r = \infty) = t_b, \qquad t'(r = \infty) = 0, \qquad \theta(r = \infty) = \theta_0, \qquad \theta'(r = \infty) = 0.
\end{equation}
However, it is simpler to impose all conditions at $r = r_s$ rather than $r = \infty$:
\begin{equation}
t(r_s) = 0, \qquad t'(r_s) = \infty, \qquad \theta(r_s) = \theta_s, \qquad \theta'(r_s) = \text{TBD}.
\end{equation}
To determine what $\theta'(r_s)$ should be, note that the boundary conditions for $r(t)$ and $\theta(t)$ imply that, near $t = 0$,
\begin{equation}
r(t) = r_s + \frac{1}{2}\ddot{r}(0)t^2 + \cdots, \qquad \theta(t) = \theta_s + \frac{1}{2}\ddot{\theta}(0)t^2 + \cdots.
\end{equation}
Hence we can write
\begin{equation}
t^2 = \frac{2(r - r_s)}{\ddot{r}(0)} + \cdots \quad\implies\quad \theta(r) = \theta_s + \frac{\ddot{\theta}(0)}{\ddot{r}(0)}(r - r_s) + \cdots.
\end{equation}
The boundary conditions ensure that there is no non-analytic $\sqrt{r - r_s}$ term in $\theta(r)$, and in turn that $\theta'(r)$ is nonsingular near $r_s$ (as long as $\ddot{r}(0)\neq 0$).  We are essentially using L'H\^{o}pital's rule to evaluate $d\theta/dr = \dot{\theta}/\dot{r}$ by going to the second derivatives when the first derivatives vanish:
\begin{equation}
\frac{d\theta}{dr}(r_s) = \frac{\frac{d\theta}{dt}(0)}{\frac{dr}{dt}(0)} = \frac{\ddot{\theta}(0)}{\ddot{r}(0)}.
\end{equation}
We can determine $\ddot{r}(0)$ and $\ddot{\theta}(0)$ analytically in terms of $r_s$ and $\theta_s$ using the EOMs.  To leading order in $t$, the $r(t)$ and $\theta(t)$ equations of motion give, respectively:
\begin{align}
2r_s\ddot{r}(0) + r_s f(r_s)f'(r_s) + 2(d - 2)f(r_s)^2 + O(t^2) &= 0, \\
r_s^2\ddot{\theta}(0) + (d - 2)f(r_s)\cot\theta_s + O(t^2) &= 0.
\end{align}
Hence we get
\begin{equation}
\theta'(r_s) = \frac{\ddot{\theta}(0)}{\ddot{r}(0)} = \frac{2(d - 2)\cot\theta_s}{2(d - 2)r_s f(r_s) + r_s^2 f'(r_s)}.
\end{equation}
This vanishes for the constant solution $\theta = \theta_0 = \pi/2$.

If the action is a function of $\dot{\theta}^2$ and independent of $\theta$, then $\theta = \text{constant}$ is a solution of the equation of motion for $\theta$.  More generally, if the action is a function of $\dot{\theta}^2$ and its variation with respect to $\theta$ vanishes for $\theta = \text{some constant}$, then $\theta = \text{that constant}$ is a solution of the equation of motion.  This is the case for us when $\theta = \pi/2$ (the hemisphere) because the action depends on $\theta$ through $\sin^{d-2}\theta$, so its variation contains a factor of $\cos\theta$.  The HM case (half-space in Poincar\'e coordinates) can be viewed as the limit of an infinitely large sphere, where any plane corresponds to a great circle.  However, $\theta = \theta_0$ is not a solution of the equation of motion for arbitrary $\theta_0$.  In particular, for $t_b = 0$, the surface lies strictly at $t = 0$ and the area functional reduces to that for the static surface, \eqref{areafunctionalstatic} or \eqref{areafunctionalstaticr}; to obtain the HM surface, we would solve the same equation of motion with horizon-crossing boundary conditions.

The equations of motion for $(r, \theta)$ or $(t, \theta)$ are coupled second-order nonlinear ODEs.  We still have a conserved quantity:
\begin{equation}
\frac{(r\sin\theta)^{d-2}f(r)}{\sqrt{-f(r) + \frac{\dot{r}^2}{f(r)} + r^2\dot{\theta}^2}} = \frac{(r\sin\theta)^{d-2}f(r)t'}{\sqrt{\frac{1}{f(r)} - f(r)t'^2 + r^2\theta'^2}} = C.
\end{equation}
However, it does not fix the solution.  To determine the HM surface, our strategy is to solve for $t(r)$ and $\theta(r)$ numerically for fixed $r_s$ and $\theta_s$, inferring the boundary values $t_b$ and $\theta_0$ at $r = \infty$.  We must understand the behavior of the solutions at $r = r_h$, where we expect a singularity.  We choose to work with $t(r)$ rather than $r(t)$ because $t(r)$ has a singularity at $r = r_h$, which means that $r(t)$ would be multivalued.

\begin{figure}[!htb]
\centering
\includegraphics[width=0.7\textwidth]{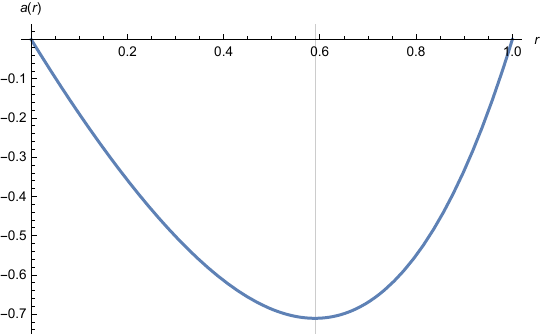}
\caption{$a(r)$ for $r\in [0, r_h]$ with black hole parameters $d = 3$, $\ell = 1$, $\mu = 2$.  $a(r)$ vanishes at $r = 0$ and $r = r_h = 1$; it achieves a minimum at $r = r_m\approx 0.589755$.}
\label{fig:ar}
\end{figure}

We first consider the instructive example with constant $\theta = \theta_0 = \pi/2$.  This requires solving for $t(r)$ only.  In this case, one can argue for the linear growth analytically (as in \cite{Hartman:2013qma}).  Specifying the growing RT surface by a function $r(t)$, the induced metric (with $\sin\theta_0 = 1$) is
\begin{equation}
ds^2 = \left[-f(r) + \frac{\dot{r}^2}{f(r)}\right]dt^2 + r^2\, d\Omega_{d-2}^2.
\end{equation}
The corresponding ``area'' functional is very similar to \eqref{volume}:
\begin{equation}
A[r(t)] = \omega_{d-2}\int dt\, r^{d-2}\sqrt{-f(r) + \frac{\dot{r}^2}{f(r)}}.
\label{areafunctional}
\end{equation}
We obtain the conserved quantity
\begin{equation}
\frac{r^{d-2}f(r)}{\sqrt{-f(r) + \frac{\dot{r}^2}{f(r)}}} = C,
\end{equation}
where the constant $C$ is determined by boundary conditions (i.e., the boundary time $t_b = t(r = \infty)$).  Solving this equation gives
\begin{equation}
\dot{r} = -f(r)\sqrt{1 + \frac{a(r)}{C^2}}, \qquad a(r) := r^{2(d-2)}f(r).
\end{equation}
Note that $f(r) < 0$ (and hence $\dot{r} > 0$) for $r < r_h$ while $f(r) > 0$ (and hence $\dot{r} < 0$) for $r > r_h$, reflecting the change in the spacetime signature across the horizon.  The overall sign is chosen so that $t_b$ increases as the innermost value of $r$ attained by the solution decreases.  We expect the RT surface to be reflection-symmetric under $t\leftrightarrow -t$, and therefore that $\dot{r} = 0$ at $t = 0$.  At this symmetric point, we have $r = r_s$ where
\begin{equation}
a(r_s) = -C^2.
\label{symmetricpoint}
\end{equation}
Since this requires $f(r_s) < 0$, we have $r_s < r_h$.  Note that $a(r)$ attains a single minimum at some $r\in (0, r_h)$, which gives an upper bound on the values of $C^2$ for which a solution to \eqref{symmetricpoint} exists (see \autoref{fig:ar}).  In view of the boundary condition $t(r_s) = 0$, we have
\begin{align}
t(r) &= -\int_{r_s}^r \frac{dr'}{f(r')\sqrt{1 - \frac{a(r')}{a(r_s)}}} \\
&:= \lim_{\epsilon\to 0}\left(-\int_{r_s}^{r_h - \epsilon} \frac{dr'}{f(r')\sqrt{1 - \frac{a(r')}{a(r_s)}}} - \int_{r_h + \epsilon}^r \frac{dr'}{f(r')\sqrt{1 - \frac{a(r')}{a(r_s)}}}\right), \label{regulatedtime}
\end{align}
where we have defined $t(r)$ as a principal value integral due to the pole of the integrand at $r = r_h$.  The logarithmic divergences cancel, as the quantity
\begin{equation}
\int_{r_h - \epsilon}^{r_h + \epsilon} \frac{dr'}{f(r')\sqrt{1 - \frac{a(r')}{a(r_s)}}} = \frac{1}{f'(r_h)}\int_{-\epsilon}^\epsilon \frac{d\epsilon'}{\epsilon'} + O(\epsilon) = O(\epsilon)
\end{equation}
is independent of $r_s$ to leading order in $\epsilon$.  Suppose that $a(r)$ is minimized at
\begin{equation}
r = r_m < r_s < r_h.
\end{equation}
For $r' > r_h$, we have $a(r')/a(r_s) < 0$, and the integrand of \eqref{regulatedtime} is well-behaved.  For $r' < r_h$, we have $a(r')/a(r_s) > 0$, and the integrand of \eqref{regulatedtime} diverges at $r' = r_s$.  However, this is a mild (square-root) divergence, which is still integrable: setting $r' = r_s + \epsilon$, we have
\begin{equation}
-\int_{r_s} \frac{dr'}{f(r')\sqrt{1 - \frac{a(r')}{a(r_s)}}} = -\int_0 \frac{d\epsilon}{f(r_s)\sqrt{-\frac{a'(r_s)}{a(r_s)}\epsilon}}(1 + O(\epsilon))
\end{equation}
(note that $a(r_s) < 0$ while $a'(r_s) > 0$).  The trouble arises when $r_s\to r_m$ (and hence $a'(r_s)\to 0$): in this limit, the integral diverges.  This is precisely the limit that $t_b\to\infty$.  In this limit, the main contribution to \eqref{regulatedtime} comes from the part of the integral near $r_m$.  Formally, we write
\begin{equation}
t_b(r_s) = t(r = \infty)\big|_{t(r_s) = 0} = \lim_{\epsilon\to 0}\left(-\int_{r_s}^{r_h - \epsilon} \frac{dr}{f(r)\sqrt{1 - \frac{a(r)}{a(r_s)}}} - \int_{r_h + \epsilon}^\infty \frac{dr}{f(r)\sqrt{1 - \frac{a(r)}{a(r_s)}}}\right).
\end{equation}
We find that $t_b(r_s)$ decreases monotonically from $t_b(r_m) = \infty$ to
\begin{equation}
t_b(r_h) = \lim_{\delta\to 0}\lim_{\epsilon\to 0}\left(-\int_{r_h - \delta}^{r_h - \epsilon} \frac{dr}{f(r)\sqrt{1 - \frac{a(r)}{a(r_h - \delta)}}} - \int_{r_h + \epsilon}^\infty \frac{dr}{f(r)\sqrt{1 - \frac{a(r)}{a(r_h - \delta)}}}\right) = 0
\end{equation}
as $r_s$ ranges from $r_m$ to $r_h$.  Note that $t(r)$ does not increase monotonically as $r$ ranges from $r_s$ to $\infty$.  We can compute the area by integrating from $r = r_s$ to $r = \infty$ and multiplying by two, using the $t\leftrightarrow -t$ symmetry, as follows:
\begin{align}
A_\text{HM}\big|_{\theta_0 = \pi/2} &= 2\omega_{d-2}\int_{r_s}^\infty \frac{dr}{|\dot{r}|}\, r^{d-2}\sqrt{-f(r) + \frac{\dot{r}^2}{f(r)}} \\
&= 2\omega_{d-2}\int_{r_s}^\infty \frac{r^{2(d-2)}\, dr}{\sqrt{a(r) - a(r_s)}}.
\end{align}
For large $t_b$, we can approximate \eqref{areafunctional} by setting $\dot{r} = 0$ and $r = r_m$, finding:
\begin{equation}
A_\text{HM}\big|_{\theta_0 = \pi/2} = 2\omega_{d-2}\sqrt{-a(r_m)}t_b + \text{constant},
\end{equation}
where the constant includes a universal ($t_b$-independent) large-$r$ divergence.  We can subtract this divergence explicitly by writing
\begin{equation}
A_\text{HM}\big|_{\theta_0 = \pi/2} = 2\omega_{d-2}\lim_{R\to\infty}\left[\int_{r_s}^R \frac{r^{2(d-2)}\, dr}{\sqrt{a(r) - a(r_s)}} - \frac{\ell R^{d-2}}{d - 2}\right],
\end{equation}
where $R$ is a radial cutoff and we have neglected a possible universal constant.  See \autoref{fig:tbanalytical}.

\begin{figure}[!htb]
\centering
\includegraphics[width=0.49\textwidth]{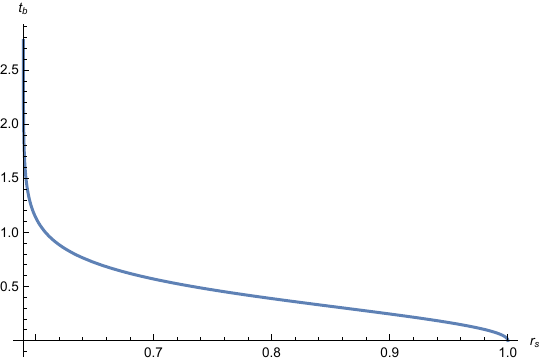}
\hfill
\includegraphics[width=0.49\textwidth]{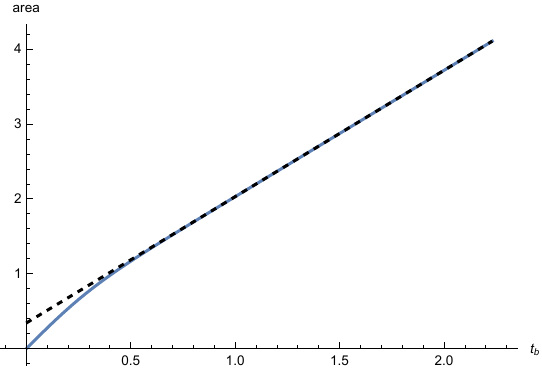}
\caption{Properties of analytical solution with $\theta_0 = \pi/2$ and the same black hole parameters as in \autoref{fig:ar}.  Left: boundary time $t_b$ versus $r_s$; $t_b$ ranges from $\infty$ to 0 as $r_s$ ranges from $r_m$ to $r_h$.  Right: area (computed with a radial cutoff) versus $t_b$, showing asymptotically linear growth.}
\label{fig:tbanalytical}
\end{figure}

Of course, this behavior cannot continue to hold for \emph{arbitrarily} large $t_b$ if the boundary of the spatial subregion is compact, as it is in our case.  This is due to the dominance of a disconnected extremal surface (with smaller area) at late times, where each component lies outside the horizon but close to it.  As noted in \cite{Hartman:2013qma}, these new surfaces live purely at $t_b$ and asymptote to the horizon at late times.

Let us understand the singular behavior of the $\theta_0 = \pi/2$ solution near $r = r_h$.  Since $t(r)$ is computed via a principal value integral, it is symmetric about $r = r_h$ for $r$ very close to $r_h$:
\begin{equation}
\lim_{\epsilon\to 0} [t(r_h - \epsilon) - t(r_h + \epsilon)] = \lim_{\epsilon\to 0} [t'(r_h - \epsilon) + t'(r_h + \epsilon)] = 0.
\end{equation}
This gives us a prescription for analytic continuation across (i.e., around) $r = r_h$.  More specifically, we can derive the matching condition by solving the EOM in a series expansion around $r = r_h$.  In terms of $t(r)$, the area functional is
\begin{equation}
A[t(r)] = \omega_{d-2}\int dr\, r^{d-2}\sqrt{\frac{1}{f(r)} - f(r)t'^2}.
\end{equation}
The equation of motion for $t(r)$ is
\begin{equation}
(2(d - 2)f(r) + 3rf'(r))t' - f(r)^2(2(d - 2)f(r) + rf'(r))t'^3 + 2rf(r)t'' = 0.
\end{equation}
Setting $r = r_h + \epsilon$ gives
\begin{align}
0 &= (3r_h f'(r_h) + O(\epsilon))t'(r_h + \epsilon) \nonumber \\
&\quad + (2r_h f'(r_h)\epsilon + O(\epsilon^2))t''(r_h + \epsilon) \\
&\quad - (r_h f'(r_h)^3\epsilon^2 + O(\epsilon^3))t'(r_h + \epsilon)^3. \nonumber
\end{align}
Cancellation of the leading terms in the EOM near $r_h$ requires (up to an overall sign) that
\begin{equation}
t'(r_h + \epsilon) = -\frac{1}{f'(r_h)\epsilon} + O(1).
\label{EOMsingularity}
\end{equation}
This is consistent with the exact solution, which implies that $t(r_h + \epsilon)$ diverges logarithmically as $\epsilon\to 0$.  Indeed, slightly inside the horizon, we have
\begin{equation}
t(r_h - \epsilon) = -\int_{r_s}^{r_h - \epsilon} \frac{dr'}{f(r')\sqrt{1 - \frac{f(r')}{f(r_s)}}},
\end{equation}
where we have substituted $a(r) = r^{2(d-2)}f(r)$.  This gives
\begin{equation}
t'(r_h - \epsilon) = -\frac{1}{f(r_h - \epsilon)\sqrt{1 - \frac{f(r_h - \epsilon)}{f(r_s)}}} = \frac{1}{f'(r_h)\epsilon} + O(1).
\end{equation}
The principal value prescription then stipulates \eqref{EOMsingularity}.

We therefore find that the function $t(r) + \frac{\log|r - r_h|}{f'(r_h)}$ is regular on all of $[r_s, \infty)$.  This gives a route to a numerical solution: we solve the EOM (starting at $r = r_s$) for
\begin{equation}
t_\text{reg}(r) := t(r) + \frac{\log(r_h - r)}{f'(r_h)},
\end{equation}
where the initial conditions $t(r_s) = 0$ and $t'(r_s) = \infty$ become
\begin{equation}
t_\text{reg}(r_s) = \frac{\log(r_h - r_s)}{f'(r_h)}, \qquad t_\text{reg}'(r_s) = \infty.
\end{equation}
We then infer $t(r)$ from $t_\text{reg}(r)$, where the latter is nonsingular at $r = r_h$.  See \autoref{fig:innerouter}.

\begin{figure}[!htb]
\centering
\includegraphics[width=0.49\textwidth]{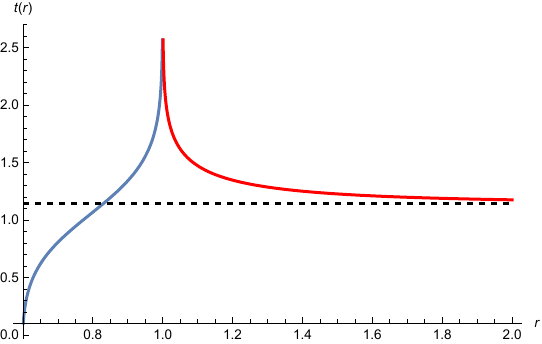}
\hfill
\includegraphics[width=0.49\textwidth]{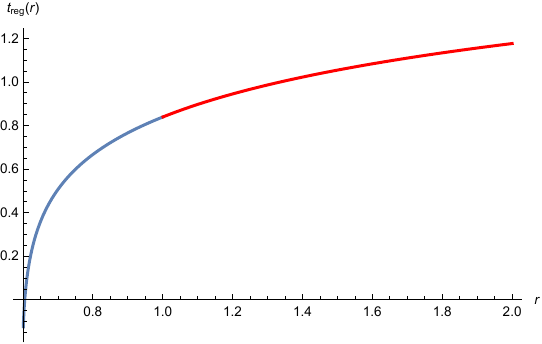}
\caption{Example of an analytical solution for $r_s = 0.6$, with the same parameter values as in \autoref{fig:ar}.  Left: inner (blue) and outer (red) solutions for $t(r)$; the asymptote is shown as a dashed black line.  Right: inner (blue) and outer (red) solutions for $t_\text{reg}(r)$, after subtracting the logarithmic divergence from $t(r)$ at $r = r_h$.}
\label{fig:innerouter}
\end{figure}

Returning to the case of arbitrary $\theta_0$, the equation of motion for $t(r)$ is
\begin{gather}
(3rf'(r) + 2r^2 f(r)((d - 3)f(r) + rf'(r))\theta'^2 - 2r^3 f(r)^2\theta'\theta'' \nonumber \\
{} + 2(d - 2)f(r)(1 + r(1 + r^2 f(r)\theta'^2)\theta'\cot\theta))t' \\
{} - f(r)^2(2(d - 2)f(r)(1 + r\theta'\cot\theta) + rf'(r))t'^3 + 2rf(r)(1 + r^2 f(r)\theta'^2)t'' = 0, \nonumber
\end{gather}
and the equation of motion for $\theta(r)$ is
\begin{gather}
2(d - 2)(1 - f(r)^2 t'^2)(1 - f(r)^2 t'^2 + r^2 f(r)\theta'^2)\cot\theta \nonumber \\
{} - r(rf'(r) + 2df(r) + f(r)^2(rf'(r) - 2df(r))t'^2 + 2rf(r)^3 t't'')\theta' \\
{} - 2(d - 1)r^3 f(r)^2\theta'^3 - 2r^2 f(r)(1 - f(r)^2 t'^2)\theta'' = 0. \nonumber
\end{gather}
Note that $\theta(r)$ cannot diverge at $r = r_h$ since the range of $\theta$ is bounded.  Assuming that $\theta(r)$ is regular and setting $r = r_h + \epsilon$ in the equation of motion for $t(r)$ gives exactly the same leading behavior for $t(r_h + \epsilon)$ as in \eqref{EOMsingularity}.  As a consistency check, note that all powers of $t'$ or $t''$ in the equation of motion for $\theta(r)$ are compensated for by powers of $f(r)$, so all terms in the equation are regular near $r_h$.  Therefore, we use the same strategy of subtracting the pole in $t(r)$ at $r = r_h$ to solve the modified EOMs for a completely regular $t_\text{reg}(r)$.

\subsection{Comments on Volumes}

In terms of $r(t)$ and $\theta(t)$, the area functional for the growing HM surface is \eqref{areafunctionalgrowingt}.  In terms of $\theta(r)$, the area functional for the static RT surface is \eqref{areafunctionalstaticr}.  Let us write down two corresponding volume functionals.  Making $\theta$ explicit, the general metric is
\begin{equation}
ds^2 = -f(r)\, dt^2 + \frac{dr^2}{f(r)} + r^2(d\theta^2 + \sin^2\theta\, d\Omega_{d-2}^2).
\end{equation}
For the growing surface, we fix $r(t)$ but integrate $\theta$ from 0 to $\theta(t)$, so the induced metric is
\begin{equation}
ds^2 = \left[-f(r) + \frac{\dot{r}^2}{f(r)}\right]dt^2 + r^2(d\theta^2 + \sin^2\theta\, d\Omega_{d-2}^2).
\end{equation}
The volume functional is then
\begin{equation}
V[r(t), \theta(t)] = \omega_{d-2}\int dt\, r(t)^{d-1}\sqrt{-f(r(t)) + \frac{\dot{r}(t)^2}{f(r(t))}}\int_0^{\theta(t)} d\vartheta\, \sin^{d-2}\vartheta.
\end{equation}
For the static surface, we fix $t = t_b$ but integrate $\theta$ from 0 to $\theta(r)$, so the induced metric is
\begin{equation}
ds^2 = \frac{dr^2}{f(r)} + r^2(d\theta^2 + \sin^2\theta\, d\Omega_{d-2}^2).
\end{equation}
The corresponding volume functional is
\begin{equation}
V[\theta(r)] = \omega_{d-2}\int dr\, \frac{r^{d-1}}{\sqrt{f(r)}}\int_0^{\theta(r)} d\vartheta\, \sin^{d-2}\vartheta.
\end{equation}
Following \cite{Susskind:2014moa}, we estimate the volume of the growing surface.  For simplicity, we take $\theta_0 = \pi/2$.  Then, as $t_b\to\infty$, the solution approaches one with constant $r = r_f$ and constant $\theta = \theta_0 = \pi/2$.  Indeed, the variational equations (for the area functional) for constant $r$ and $\theta$ simplify to
\begin{equation}
2(d - 2)f(r_f) + r_f f'(r_f) = 0, \qquad \cot\theta_0 = 0.
\end{equation}
Note that $r_f < r_h$.  The corresponding area is
\begin{equation}
A(\theta_0) = \omega_{d-2}r_f^{d-2}\sqrt{-f(r_f)}\int dt,
\end{equation}
and the corresponding volume is
\begin{align}
V(\theta_0) &= \omega_{d-2}r_f^{d-1}\sqrt{-f(r_f)}\int_0^{\pi/2} d\theta\, \sin^{d-2}\theta\int dt \\
&= \frac{1}{2}\omega_{d-1}r_f^{d-1}\sqrt{-f(r_f)}\int dt.
\end{align}
We have used
\begin{equation}
\omega_{d-1} = \frac{2\pi^{d/2}}{\Gamma(\frac{d}{2})} = \frac{\pi^{1/2}\Gamma(\frac{d - 1}{2})}{\Gamma(\frac{d}{2})}\frac{2\pi^{(d-1)/2}}{\Gamma(\frac{d - 1}{2})} = \left(\int_0^\pi d\theta\, \sin^{d-2}\theta\right)\omega_{d-2}.
\end{equation}
Therefore, at late times $t_b$, we have
\begin{equation}
V\approx \omega_{d-1}r_f^{d-1}\sqrt{-f(r_f)}t_b,
\end{equation}
where we have used $t\leftrightarrow -t$ symmetry to evaluate the integral.  This is because most of the HM surface lies at $r = r_f$; it only deviates near the ends.  The transition occurs at $t_b\sim L\sim \ell$, so
\begin{equation}
V_\text{crit}\approx \omega_{d-1}r_f^{d-1}\sqrt{-f(r_f)}\ell.
\end{equation}
This quantity is implicitly regularized.

The discussion above is only schematic, as $V[r(t), \theta(t)]$ does not truly account for the volume of the \emph{maximal-volume} spatial slice contained in the entanglement wedge.  Suppose we denote the volumes of interest by $\operatorname{vol}_\text{HM}(t)$ and $\operatorname{vol}_\text{RT}$, where the latter is time-independent and follows from $V[\theta(r)]$.  One can compute the differences $\operatorname{vol}_\text{HM}(t) - \operatorname{vol}_\text{RT}$, but comparing multiplicative factors between $\operatorname{vol}_\text{HM}(t)$ and $\operatorname{vol}_\text{RT}$ requires computing absolute rather than relative volumes.  A phys\-i\-cal\-ly meaningful regularization would be to compare $\operatorname{vol}_\text{HM}(t) - \operatorname{vol}_\text{HM}(0)$ and $\operatorname{vol}_\text{RT} \linebreak[1] - \linebreak[1] \operatorname{vol}_\text{HM}(0)$.  By setting $t$ to be the transition time, we would see how pronounced the sawtooth is.

\subsection{Numerics}

Our numerical routines for both the static and growing RT surfaces involve tolerances and approximations (e.g., finite truncations of infinite parameter ranges).  The optimal values of the numerical tolerances and ranges depend on the black hole parameters.

For illustration, we give examples of our numerical routines for $d = 3$, $\ell = 1$, $\mu = 2$, which is one of the (borderline) large black holes at the Hawking-Page threshold considered in the main text.  For these parameter values, $r_h = 1$.

For the static surface, we use a dimensionless tolerance $\epsilon$ as a proxy for 0 to avoid singularities.  We fix an $r_\text{max}$.  For a given $r_0 > r_h$, we solve for $\theta(r)$ for $r\in [r_0, r_\text{max}]$ by imposing the initial conditions $\theta(r_0) = \epsilon$ and $\theta'(r_0) = 1/\epsilon$.  We (approximately) determine the corresponding boundary value of $\theta$ (i.e., $\theta_0$) by evaluating $\theta(r_\text{max})$, as well as the corresponding area by evaluating the area functional on the solution $\theta(r)$ and integrating up to some cutoff on $r$.  See \autoref{fig:examplestatic}.

\begin{figure}
\centering
\includegraphics[width=0.9\textwidth]{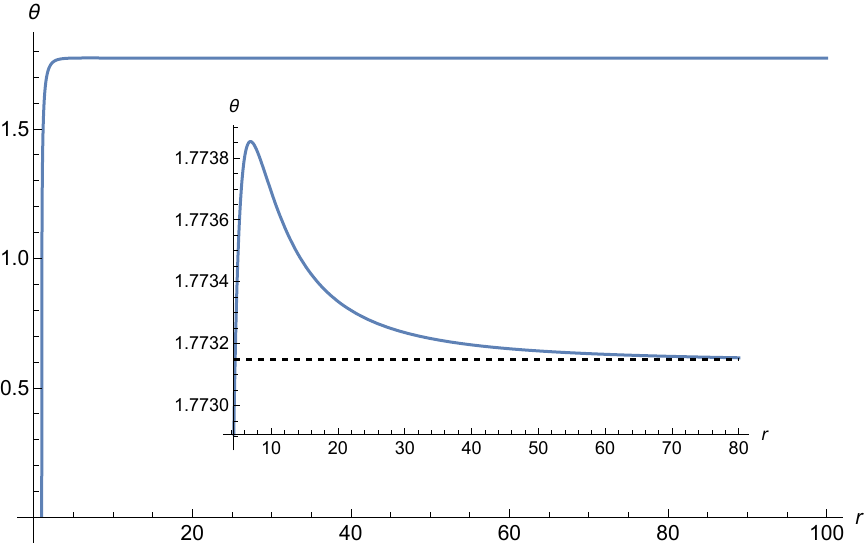}
\caption{Example of a static solution for $\theta(r)$ with $r_0 = r_h + 0.005 = 1.005$ and an asymptotic value of $\theta_0\approx 1.77315$.  Zooming in to the onset of the plateau reveals non-monotonicity (magnified in the inset, with asymptote shown).}
\label{fig:examplestatic}
\end{figure}

Using this $\theta(r)$ routine, we compute an interpolating function for $\theta_0$ as a function of $r_0$ for some range $r_0\in [r_0^\text{min}, r_0^\text{max}]$, where $r_0^\text{min}$ is slightly above $r_h$.  The $\theta(r)$ routine breaks down if $r_0$ is too large, which sets a limit on how large $r_0^\text{max}$ can be.  See \autoref{fig:th0versusr0}.

\begin{figure}
\centering
\includegraphics[width=0.8\textwidth]{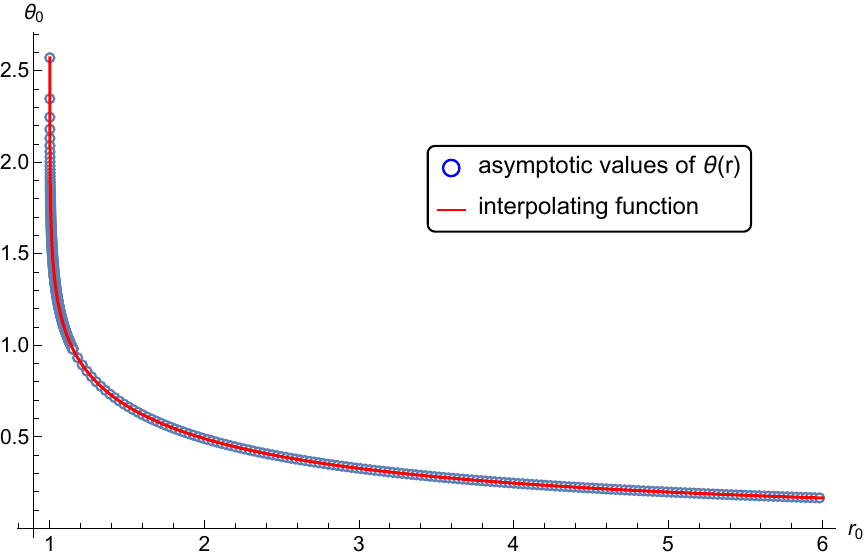}
\caption{$\theta_0$ versus $r_0$, computed via numerical solutions for $\theta(r)$ (data points in blue and interpolating function in red).  In this case, $r_0$ ranges from slightly above $r_h$ to $r_h + 5$; $\theta_0$ ranges from $\pi$ to 0 as $r_0$ ranges from $r_h$ to $\infty$.}
\label{fig:th0versusr0}
\end{figure}

For the growing surface, we compute both $t(r)$ and $\theta(r)$ for given $r_s < r_h$ and $\theta_s$.  To obtain the interior solutions, it suffices to solve the equations of motion for $r\in [r_s, r_h)$.  To obtain the solutions on all of $[r_s, \infty)$ (or rather, on $[r_s, r_\text{max}]$), we instead solve the modified equations of motion for $t_\text{reg}(r)$, which is $t(r)$ with the singular term at $r = r_h$ subtracted.  We use a dimensionless cutoff $\Lambda$ as a proxy for $\infty$ and impose that $t_\text{reg}'(r_s) = \Lambda$.  The $t_\text{reg}(r)$ routine breaks down if $\Lambda$ is too large.  See \autoref{fig:examplegrowing}.

Finally, we compare the areas of the static and growing surfaces.  We fix a radial cutoff $R < r_\text{max}$ (chosen sufficiently large that the area difference is stable).

To compute the static area as a function of boundary angle $\theta_0$, we evaluate the static area functional (integrated up to $R$) on the solution for $\theta(r)$ computed for $r_0(\theta_0)$, where $r_0(\theta_0)$ is computed by inverting the interpolating function for $\theta_0(r_0)$.

To compute the growing area as a function of $r_s$ and $\theta_s$, we evaluate the growing area functional (integrated up to $R$) on the solutions for $t(r)$ and $\theta(r)$, where $t(r)$ is obtained from $t_\text{reg}(r)$ by adding back the singular term.  We compute approximations to the corresponding boundary time $t_b$ and boundary angle $\theta_0$ by evaluating $t(r_\text{max})$ and $\theta(r_\text{max})$.

\begin{figure}[!htb]
\centering
\begin{subfigure}[c]{0.49\textwidth}
\centering
\includegraphics[width=\textwidth]{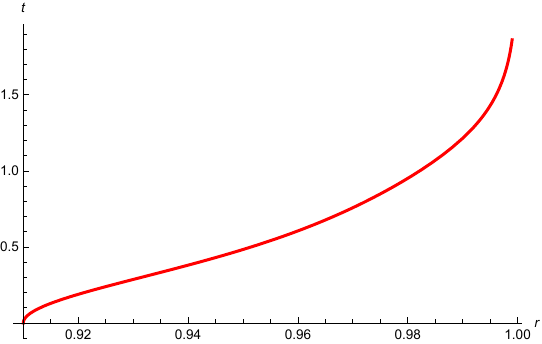}
\caption{Inner solution for $t(r)$.}
\end{subfigure}
\hfill
\begin{subfigure}[c]{0.49\textwidth}
\centering
\includegraphics[width=\textwidth]{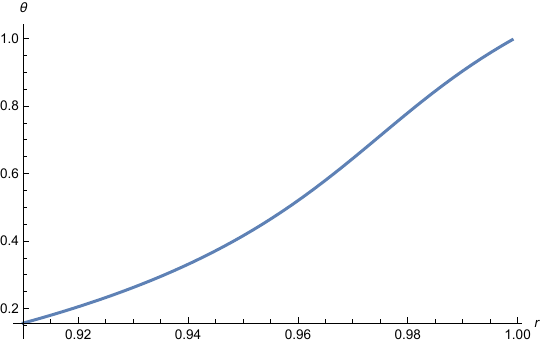}
\caption{Inner solution for $\theta(r)$.}
\end{subfigure}
\\
\begin{subfigure}[c]{0.49\textwidth}
\centering
\includegraphics[width=\textwidth]{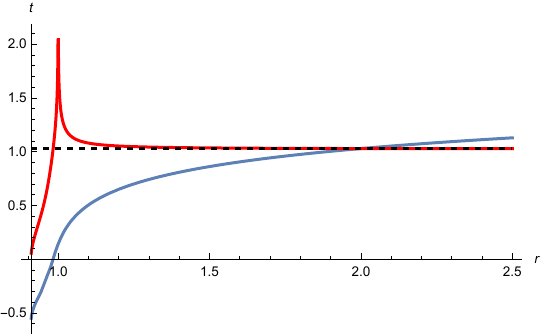}
\caption{Full solution for $t(r)$.}
\end{subfigure}
\hfill
\begin{subfigure}[c]{0.49\textwidth}
\centering
\includegraphics[width=\textwidth]{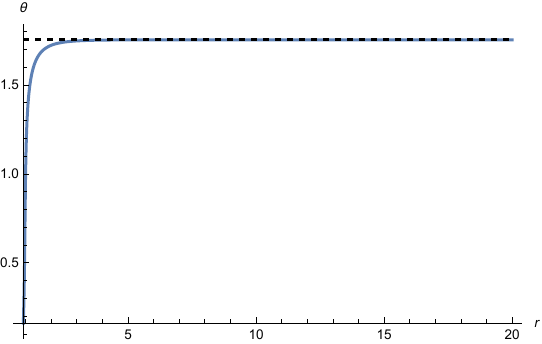}
\caption{Full solution for $\theta(r)$.}
\end{subfigure}
\caption{Example of a growing solution with $r_s = r_h - 0.09$ and $\theta_s = 0.1(\pi/2)$.  Top row: inner solutions for $t(r)$ and $\theta(r)$ with $r\in [r_s, r_h)$; $t(r)$ diverges at $r = r_h$.  Bottom row: full solutions for $t_\text{reg}(r)$ and $\theta(r)$, with $t_\text{reg}(r)$ in blue and the corresponding $t(r)$ in red; the asymptotic values $t_b\approx 1.02586$ and $\theta_0\approx 1.75316$ are easily read off.}
\label{fig:examplegrowing}
\end{figure}

For a given $(r_s, \theta_s)$, we compute $t_b(r_s, \theta_s)$ and $\theta_0(r_s, \theta_s)$, as well as the difference between the growing area and the static area corresponding to $\theta_0(r_s, \theta_s)$.  For fixed $r_s$, we scan over values of $\theta_s$ that produce close to zero area difference (within some tolerance).  Having found an appropriate $\theta_s$, the transition time for the boundary angle $\theta_0(r_s, \theta_s)$ is taken to be $t_b(r_s, \theta_s)$.  Thus we obtain, pointwise, the transition time as a function of $\theta_0$.  We convert this into the transition time as a function of subsystem fraction $p$.

\section{Complexity of Subsystems of Random Quantum Circuits} \label{app:rqccomp}

Here, we precisely define some of the quantum information-theoretic concepts central to our discussion and present proofs of the theorems stated in the main text.

\subsection{Quantum Information Preliminaries}

The local degrees of freedom are $q$-dimensional qudits with Hilbert space $\mathbb{C}^q$. We will often consider systems of $n$ qudits with corresponding tensor product Hilbert space $(\mathbb{C}^{q})^{\otimes n} = \mathbb{C}^{q^n}$. Sometimes, we will allow for $O(\poly(n))$ ancilla qudits alongside $n$ system qudits. We will work with matrix norms to measure distances between density matrices on the systems of interest.
If $X$ is an $m\times m$ complex matrix, i.e., if $X \in \mathbb{C}^{m \times m}$, then the Schatten $p$-norm of $X$ is denoted by $\|X \|_p$ and defined as $\|X \|_p := (\tr(|X|^p))^{1/p}$, where $|X| := \sqrt{X^\dagger X}$. Equivalently, if the singular values of $X$ are $\{s_i\}_{i=1} ^m$, then the Schatten $p$-norm is the vector $\ell_p$-norm of the singular values: $\|X\|_p = \big(\sum_i s_i^p\big)^{1/p}$. The trace distance $\dist(\rho, \sigma)$ between two equal-dimensional density matrices $\rho$ and $\sigma$ is defined to be $\dist(\rho, \sigma) := \frac{1}{2}\lVert\rho - \sigma\rVert_1$. Finally, we use $O(f(x))$ and $\Omega(f(x))$ to denote asymptotic upper bounds and asymptotic lower bounds, respectively. Specifically, when we write $f(x) = O(g(x))$, we mean that there exist universal constants $C > 0$ and $x_0$ such that for all $x > x_0$, $f(x) \leq C g(x)$. Likewise, when we write $f(x) = \Omega(g(x))$, we mean that there exist universal constants $c > 0$ and $x_0$ such that for all $x > x_0$, $f(x) \geq c g(x)$.

The dynamical models we consider are $n$-qudit random quantum circuits. They are composed of quantum gates, which are unitary operators from the group $U(q^2)$. By definition, the action of a two-qudit gate on an $n$-qudit state is the natural action of the unitary operator on two qudits (which must be specified) and the identity operator on all other qudits. Random quantum circuits are an ensemble (a distribution) of unitaries on $U(q^n)$ with  probability equal to the convolution of the probabilities of picking specific two-qudit gates at specific locations in the circuit. The property of random quantum circuits that makes them analytically tractable is their closeness to Haar-random unitaries (the uniform distribution over the unitary group). The closeness of the two distributions is captured by the notion of an $\varepsilon$-approximate unitary $k$-design. Let $\mu_H$ denote the Haar measure on the unitary group $U(D)$. Colloquially speak\-ing, distributions $\nu$ on $U(D)$ are $\varepsilon$-approximate unitary $k$-designs if the $k^\text{th}$ ``moments'' of $\nu$ are $\varepsilon$-``close'' to those of $\mu_H$. The moments are defined as $k$-fold quantum channels, and the notion of closeness is captured by the complete positivity of scaled differences of those channels, as we define below. For any distribution $\nu$ on the unitary group $U(D)$, we define the $k$-fold channel of $\nu$ on an operator $\op$ acting on $(\C^D)^{\otimes k}$ as
\begin{equation}
    \label{def:moment_channel}
    \Phi^{(k)}_\nu (\op) := \int d\nu\, U^{\otimes k}\op(U^\dagger)^{\otimes k}\,.
\end{equation}
An exact unitary $k$-design is a distribution $\nu$ such that $\Phi^{(k)}_\nu = \Phi^{(k)}_{\mu_H}$, i.e., the $k^\text{th}$ moments exactly match those of the Haar measure.

\begin{definition}[$\varepsilon$-approximate unitary $k$-designs]
    A distribution $\nu$ on the unitary group $U(D)$ is a relative-error $\varepsilon$-approximate unitary $k$-design if the $k$-fold channels are such that for any $\varepsilon > 0$, both
    \begin{equation}
        \Phi_{\nu} ^{(k)} - (1 - \varepsilon) \Phi_{\mu_H} ^{(k)} \succeq 0  \quad{\rm and}\quad (1 + \varepsilon) \Phi_{\mu_H} ^{(k)} - \Phi_{\nu} ^{(k)} \succeq 0
    \end{equation}
    are simultaneously true, where $\Phi \succeq 0$ means that the channel $\Phi$ is completely positive.
\end{definition}

We now define the random quantum circuit models that we consider. To simplify these definitions, we will assume that $n$ is even. We label the $n$ qudits by $\{0, 1, \dots, n-1\}$.

\begin{definition}[Brickwork random quantum circuits]
    \label{def:bwrqc}
    Let $\bwrqc$ denote a depth-$t$ brickwork random quantum circuit on $n$ qudits defined by a $t$-fold composition of $U_{0,1}\otimes U_{2,3}\otimes \cdots \otimes U_{n-2, n-1}$ at odd and $U_{1,2}\otimes U_{3,4}\otimes \cdots \otimes U_{n-3,n-2} \otimes B_{n-1,0}$ at even iterations in the composition, where each $U_{i,i+1}$ is a Haar-random two-qudit gate drawn with respect to the Haar measure $\mu_H(U(q^2))$ that acts nontrivally on qudits $i$ and $i+1$, i.e., $U_{i,i+1} = (U)_{i,i+1}\otimes \iden_{[n]\setminus \{i,i+1\}}$ and $U\sim \mu_H(U(q^2))$. Depending on whether we consider open or periodic boundary conditions, $B_{n-1,0} = \iden_{n-1,0}$ or $U_{n-1,0}$, respectively.
\end{definition}

Before we define the next model, we cover some prerequisite notation. Let $\xi, d \in \mathbb{N}$. Let $\xi$ divide $n$ into $m$ parts where $m$ is even, i.e., $n/\xi = m \in \mathbb{N}$ such that $m\bmod 2=0$. We arrange the qudits in a periodic 1D geometry. We define a ``patch'' labeled by $p \in \{0, 1, \dots, m-1\}$ as a contiguous collection of $2\xi$ qudits. The $p^{\text{th}}$ patch contains the qudits with labels $\{p \xi, p \xi + 1, \dots, (p+2) \xi - 1\}$. Let $U_{p}^{(d)}$ denote a depth-$d$ brickwork random quantum circuit with open boundary conditions (cf.\ \autoref{def:bwrqc}) on the qudits in patch $p$.

\begin{definition}[Patchwork random quantum circuits]
    \label{def:pwrqc}
    Patchwork random quantum circuits, denoted by $\pwrqc$, of depth $2d$ and patch size $2\xi$ on $n$ qudits are distributions over $U(q^n)$ of the unitaries
    \[
    (U_{1}^{(d)} \otimes U_{3}^{(d)} \otimes \cdots \otimes U_{m-1}^{(d)}) \cdot (U_{0}^{(d)} \otimes U_{2}^{(d)} \otimes \cdots \otimes U_{m-2}^{(d)}).
    \]
\end{definition}
\begin{remark}
    Note that in \autoref{def:pwrqc}, the patch labeled by $p=m-1$ contains the qudits labeled by $\{n - \xi, \dots, n-1, 0, \dots, \xi - 1\}$, and hence this model can be thought of as a locally interacting, one-dimensional random quantum circuit with periodic boundary conditions.
\end{remark}

\begin{remark}
\label{rem:time}
The notion of depth or time in patchwork random quantum circuits is different from that in brickwork random quantum circuits. Recall that in brickwork random quantum circuits (\autoref{def:bwrqc}), the location of gates is independent of the choice of time or depth. At an even depth $t$, a brickwork random quantum circuit is $t/2$ repeated applications of the same arrangement of gates. In contrast, in patchwork random quantum circuits, the location of gates is specified by the depth $d$. In other words, depth-$d$ patchwork random quantum circuits are not any repeated applications of the same layers of gates. For this reason, we define a family of patchwork random quantum circuits $\{\pwrqc(t)\}_t$ explicitly indexed by the depth, highlighting that for different $t$, the structure of the random circuit with respect to the placement of the gates is different. When we refer to a depth-$t$ patchwork random quantum circuit, we mean $\pwrqc(t)$. Contrast this with a brickwork random quantum circuit $\bwrqc^{*t}$ of even depth $t$.
\end{remark}

We consider bipartitions of the $n$ qudits in a random quantum circuit into subsystems $A$ and its complement $B$. The number of qudits in $A$ and $B$ and the Hilbert space dimensions of those subsystems are denoted by $\na$ and $\nb$, and $\da := q^{\na}$ and $\db := q^{\nb}$, respectively. We are interested in the quantum state complexity of states that are time-evolved by random quantum circuits restricted to subsystem $A$. To define quantum state complexity, we first let $\mathsf{G}$ denote a universal gate set on $n$ qudits. For example, $\mathsf{G}$ could consist of the Hadamard, $T$, and CNOT gates for states of $n$ qubits. Let $\mathsf{G}_{\cxl}$ denote the set of all size-$r$ circuits constructed from $\mathsf{G}$, i.e., all unitaries that are products of $r$ gates from $\mathsf{G}$ applied to different qudits. Then the $\err$-circuit complexity of a mixed state is defined as follows.
\begin{definition}[Mixed state complexity]
    \label{def:complexity}
    The $\err$-mixed state complexity, denoted $\CC(\rho_A)$, of a mixed state $\rho_A$ on system $A$ of $n_A$ qudits is at most $\cxl$ if there exists a unitary $U \in \mathsf{G}_r$ on $n_A$ qudits and $O(\poly(n_A))$ ancillas such that
    \begin{equation}
        \dist\big(\rho_A,\tr_{A^c}(U \ketbra{\vphi} U^\dagger)\big) \leq \delta\,,
    \end{equation}
    where $\dist$ denotes the trace distance, $\ket{\vphi}$ is some fixed unentangled state on the $n_A$ qudits and $O(\poly(n_A))$ ancillas, and $\tr_{A^c}$ is the trace over the ancillas. Alternatively, the $\err$-mixed state complexity $\CC(\rho_A)$ is at least $\cxl$ if for all unitaries $U \in \mathsf{G}_r$,
    \begin{equation}
        \dist\big(\rho_A,\tr_{A^c}(U \ketbra{\vphi} U^\dagger)\big) > \delta\,.
    \end{equation}
\end{definition}

In the following, we will always work with this definition of complexity. Hence we will simply refer to it as ``complexity'' rather than as ``$\err$-mixed state complexity.'' We will also take the qudit dimension to be $q = 2$, considering systems of local qubits. This is because previous results in \cite{chen2024incompressibility, schuster2025random}, on which we rely, are proven specifically for $q=2$.

Our goal is to prove precise theorems about growth and saturation of complexity of marginals of output states of random quantum circuits. We address complexity growth in the two regimes $\na > \nb$ and $\na < \nb$ via \autoref{thm:complexity_growth_na_gt_nb} and \autoref{thm:complexity_growth_na_lt_nb}, respectively. We address complexity saturation by quoting known results in the literature (including those of \cite{knill1995approximation} and \autoref{thm:cotler2022fluctuations}, which is Theorem 1 of \cite{cotler2022fluctuations}) and by providing our own \autoref{thm:patchwork_purity}. In the following theorems and proofs, we denote the probability of an event $E(x)$ of a random variable $x$ sampled from the distribution $X$ by $\pr_{x \sim X}(E(x))$.

\begin{theorem}[Complexity growth for $\na > \nb$]
    \label{thm:complexity_growth_na_gt_nb}
    Assume $A$ is a contiguous subsystem of a one-dimensional $n$-qubit system with periodic boundary conditions. For some $\err>0$, the time-evolved state $\ssa(t) = \tr_{B}(U \ketbra{\psi} U^\dagger)$ of a depth-$t$ brickwork random quantum circuit $U \sim \bwrqc$ (recall \autoref{def:bwrqc}) obeys 
    \begin{equation}
        \pr_{U \sim \bwrqc}\big(\CC(\rho_A(t)) > \cxl(t) + 1\big) \geq 1 - O(\exp(-n)),
    \end{equation}
    where
    \begin{align}
        r(t) =  \Omega\left(\frac{t}{\log(n) \log^7(t)}\right).
    \end{align}
    
\end{theorem}

\begin{theorem}[Complexity growth for $\na < \nb$]
    \label{thm:complexity_growth_na_lt_nb}
    Assume $A$ is a contiguous subsystem of a one-dimensional $n$-qubit system with periodic boundary conditions. For some $\err>0$, the time-evolved state $\ssa(t) = \tr_{B}(U \ketbra{\psi} U^\dagger)$ of a depth-$t$ patchwork random quantum circuit $U \sim \pwrqc$ (recall \autoref{def:pwrqc}) obeys
    \begin{equation}
        \pr_{U \sim \bwrqc}\big(\CC(\rho_A(t)) > \cxl(t) + 1\big) \geq 1 - O(\exp(-n)),
    \end{equation}
    where
    \begin{align}
        r(t) = \Omega\left(\frac{\na}{\log^2(n)  \log^7(t)} \left(\left(t - \frac{2t^2}{\na}\right) - \frac{t}{\na} \log(t)\right)\right).
    \end{align}
\end{theorem}

We begin by proving \autoref{thm:complexity_growth_na_gt_nb} in detail. Thereafter, we prove \autoref{thm:complexity_growth_na_lt_nb}, borrowing largely from the proof of \autoref{thm:complexity_growth_na_gt_nb}. The only (and the crucial) change in the proof is from a constant to a time-dependent dimensional factor that appears in the calculations from application of H\"older's inequality. This factor is the rank of the subsystem density matrix, which dynamically saturates to its maximum value and is crucial in proving nontrivial subsystem complexity lower bounds for $\na < \nb$. Although this factor of the rank dynamically evolves for all choices of subsystem size, for $\na > \nb$, this early-time ($O(\nb)$-time) feature is irrelevant for proving the interesting feature of linear complexity growth for $O(\exp(\na))$ timescale. However, for $\na < \nb$, since complexity growth occurs for short ($O(\na)$) times, the rank feature is significant in that time regime.

Our results for complexity growth will rely heavily on the $\varepsilon$-approximate unitary $k$-design property of random quantum circuits. This choice of proof technique is also the reason for us to introduce the more complicated patchwork random quantum circuits, as opposed to only the simplest brickwork random quantum circuits in \autoref{def:bwrqc}. Our complexity growth and saturation results for $\na < \nb$ rely on sublinear-depth $\varepsilon$-approximate unitary $k$-designs, which are available for patchwork random quantum circuits (\autoref{def:pwrqc}). We end this subsection with a short review of recent results on the same.

Recall that a ``random quantum circuit'' on $n$ qubits is not a single circuit, but rather a distribution over $n$-qubit circuits (or, equivalently, the unitary group $U(2^n)$) such that each instance of a circuit constructed via the recipe in \autoref{def:bwrqc} or \autoref{def:pwrqc} is an element of $U(2^n)$ that occurs with a probability equal to the probability of picking the specific gates in their specific locations in that circuit. The random quantum circuits introduced in \autoref{def:bwrqc} and \autoref{def:pwrqc} are relative-error $\varepsilon$-approximate unitary $k$-designs. They achieve this property when the circuit depth (circuit time) is sufficiently large. The following theorems are restatements of Corollary 1.7 from \cite{chen2024incompressibility} and Corollary 1 from \cite{schuster2025random}, which provide upper bounds on the minimum depth required to form relative-error $\varepsilon$-approximate unitary $k$-designs. (The convergence of RQCs to unitary designs is well-studied \cite{HL08, brandao2016local, harrow2023approximate, HHJ20, Haferkampt5, chen2024incompressibility}.) We will use these in our proof of the complexity growth of subsystems of the random quantum circuits introduced in \autoref{def:bwrqc} and \autoref{def:pwrqc}.

\begin{itheorem}{Brickwork designs}{Restatement of Corollary 1.7 in \cite{chen2024incompressibility}}
    An $n$-qubit depth-$t$ brickwork random quantum circuit, as defined in \autoref{def:bwrqc}, forms a relative-error $\varepsilon$-approximate unitary $k$-design for
    \begin{align}
        \label{thm:copy_incom}
        t = O\big((nk + \log(1/\varepsilon))\log^7(k)\big)\,,
    \end{align}
    when $k \leq O(2^{2n/5})$.
\end{itheorem}

\begin{itheorem}{Patchwork designs}{Restatement of Corollary 1 in \cite{schuster2025random}}
    An $n$-qubit depth-$t$ patchwork random quantum circuit with patch size $2\xi \geq 2\log_2(nk^2/\varepsilon)$, as defined in \autoref{def:pwrqc}, forms a relative-error $\varepsilon$-approximate unitary $k$-design for
    \begin{align}
        \label{thm:copy_shall}
        t = O\big(\log(nk/\varepsilon) k\log^7(k)\big)\,.
    \end{align}
\end{itheorem}

\subsection{Complexity for Greater-Than-Half System Size \texorpdfstring{$\na > \nb$}{nA > nB}}
\label{subsec:complexity_growth_na_gt_nb}

Our proof of \autoref{thm:complexity_growth_na_gt_nb} follows a similar line of reasoning as was previously given in Ref.~\cite{brandao2021models} to prove linear complexity growth of pure-state outputs of random quantum circuits. The difference between our calculations here and those in Ref.~\cite{brandao2021models} is that here, we will compute moments of overlaps of mixed states as opposed to those of pure states. To that end, we first clarify what we precisely mean by ``overlaps'' of mixed states.

Recall that the trace distance and fidelity between two quantum states $\rho$ and $\sigma$ are given by
\begin{align}
    \dist(\rho, \sigma) := \frac{1}{2}\lVert \rho - \sigma\rVert_1 \quad\text{and}\quad \F(\rho,\sigma):= \lVert\sqrt{\rho}\sqrt{\sigma}\rVert_1^2\,.
\end{align}
While trace distance is a distance measure on states, fidelity is not (because it does not satisfy the triangle inequality). However, they are related as follows:
\begin{align}
    1-\sqrt{F(\rho,\sigma)} \leq \dist(\rho,\sigma) \leq \sqrt{1 - F(\rho,\sigma)}\,.
\end{align}
Should the trace distance be less than $\delta$, we then have the following implications:
\begin{equation}
    \label{eq:change_random_variables}
    \dist(\rho,\sigma)\leq \delta \implies \F(\rho,\sigma)\geq (1-\delta)^2 \implies \db \tr(\rho\sigma)\geq (1-\delta)^2\,.
\end{equation}
The first inequality follows from the lower bound $1-\sqrt{F(\rho,\sigma)} \leq \dist(\rho,\sigma)$. The second inequality follows from the observation that
\begin{equation}
    \F(\rho,\sigma) = \|\sqrt{\rho}\sqrt{\sigma}\|_1^2 \leq {\rm rank}\big(\sqrt{\rho}\sqrt{\sigma}\big) \|\sqrt{\rho}\sqrt{\sigma}\|_2^2\leq \db \tr(\rho\sigma)\,,
\end{equation}
where we used H\"older's inequality for Schatten norms, $\|X\|_1\leq \sqrt{{\rm rank}(X)}\, \|X\|_2$, and that the rank of $\sqrt{\rho}\sqrt{\sigma}$ is bounded as ${\rm rank}\big(\sqrt{\rho}\sqrt{\sigma}\big)\leq {\rm rank}\big(\sqrt{\rho}\big) = {\rm rank}(\rho) = \min\{\da,\db\} = \db$ for $\db \leq \da$. We refer to $\tr(\rho\sigma)$ as the overlap of density matrices. In the proof below, we will take $\dist(\rho, \sigma)$ to be our random variable and use Markov's inequality to bound the probability that trace distances are large. We will find it operationally easier to pass from that random variable to the new random variable $\tr(\rho\sigma)$ using \autoref{eq:change_random_variables} and then to use Markov's inequality on the overlaps.

\begin{proof}[Proof of \autoref{thm:complexity_growth_na_gt_nb}]
    Using the definition of complexity, we rewrite the probability that the complexity is at most $r$ in terms of trace distance. Consider the collection of events
    \begin{align}
        \bigcup_{r' \leq \cxl} \bigcup_{V \in \mathsf{G}_{r'}} \big\{\dist(\ssa, \tr_{A^c} (V\ketbra{\vphi}V^\dagger)) \leq \err\big\},
    \end{align}
    where $\rho_A = \tr_B(U \ketbra{\psi} U^{\dagger})$, $\ket\psi$ is some fixed initial state, and $U$ is drawn from a distribution of random quantum circuits denoted by $\nu$, i.e., $U \sim \nu$. Recall $\ket\psi$ is our initial $n$-qubit state which we evolve by a random circuit and $\ket\vphi$ is the $(n+\poly(n))$-qubit state with which we construct our approximation to $\rho_A$. 
    If any of these events occurs (i.e., if for any $r'$ and $V$, the corresponding trace distance is below $\err$), then the complexity of $\ssa$ is at most $\cxl\geq r'$. Therefore, 
    \begin{align}
        \pr_{U \sim \nu}(\CC(\ssa) \leq \cxl) &= \pr_{U \sim \nu} \bigg(\bigcup_{r' \leq \cxl} \bigcup_{V \in \mathsf{G}_{r'}} \big\{\dist(\ssa, \tr_{A^c} (V\ketbra{\vphi}V^\dagger)) \leq \err\big\}\bigg)\,.
    \end{align}
    Note: we do not have to account for different choices of subsystems of the qudits plus ancillas as, without loss of generality, we can fix the partial trace to be over all but the first $n_A$ qudits. This is because a size-$r$ circuit preparing a mixed state on any subsystem of $n_A$ qudits is equivalent to a size-$r$ circuit preparing the same state on the first $n_A$ qudits up to a relabeling of qudits and thus also a size-$r$ circuit contained in $\mathsf{G}_r$. Using the union bound for probabilities, we find
    \begin{align}
        \pr_{U \sim \nu}(\CC(\ssa) \leq \cxl) &\leq \sum_{r' \leq \cxl} \sum_{V \in \mathsf{G}_{r'}} \pr_{U \sim \nu} \left(\dist(\ssa, \tr_{A^c} (V\ketbra{\vphi}V^\dagger)) \leq \err\right).
    \end{align}
    Moving forward, any specific choice of $V$ and hence of the state $\tr_{A^c} (V\ketbra{\vphi}V^\dagger)$ will not be important, so we simplify the notation by denoting the mixed state as $\sigma=\tr_{A^c} (V\ketbra{\vphi}V^\dagger)$. Proceeding, we upper-bound the probability that the state $\rho_A$ is close to some $\sigma$ using the fidelity as
    \begin{align}
        \pr_{U \sim \nu} \left(\dist(\ssa, \sigma) \leq \err\right) &\leq \pr_{U \sim \nu} \left(\F(\ssa, \sigma) \geq (1-\err)^2 \right) \\
        &\leq \pr_{U \sim \nu} \left(\rank(\sqrt{\ssa} \sqrt{\sigma}) \lVert \sqrt{\ssa} \sqrt{\sigma} \rVert_2 ^2 \geq (1-\err)^2 \right) \\
        &= \pr_{U \sim \nu} \left(\tr(\ssa\sigma) \geq \frac{(1-\err)^2}{\rank(\sqrt{\ssa} \sqrt{\sigma})} \right).
    \end{align}
    Denoting a general upper bound on the rank (independent of $\rho_A$ and $\sigma$) by $R$, we find
    \begin{align}
        \pr_{U \sim \nu} \left(\tr(\ssa\sigma) \geq \frac{(1-\err)^2}{R} \right) &= \pr_{U \sim \nu} \left((\tr\big(\ssa \sigma)\big)^{k} \geq \frac{(1-\err)^{2k}}{R^k} \right) \\
        &\leq \frac{R^k \Ex_{\nu}\Big[ \big(\tr(\ssa\sigma)\big)^k \Big]}{(1-\err)^{2k}}
    \end{align}
    for any $k$, where the last inequality is Markov's inequality. In order to proceed, we need to upper-bound the moments of the overlaps of $\rho_A$ and $\sigma$, for which we use the following proposition.
    \begin{proposition}
        \label{prop:rhosigma_moment}
        Let $\nu$ be a relative-error $\ep$-approximate unitary $k$-design, and let $\sigma$ be any fixed quantum state. The moments of $\tr(\rho\sigma)$, where $\rho = \tr_B (U\ketbra{\psi}U^\dagger)$ and $U$ is drawn from $\nu$, obey 
        \begin{equation}
            \Ex_\nu\Big[ \big(\tr(\rho\sigma)\big)^k \Big] \leq (1+\varepsilon)\frac{k!}{\da^k}\,.
        \end{equation}
    \end{proposition}
    Before we present the proof of \autoref{prop:rhosigma_moment}, we complete the proof of \autoref{thm:complexity_growth_na_gt_nb} as follows:
    \begin{equation}
        \pr_{U \sim \nu}(\CC(\ssa) \leq \cxl) \leq \sum_{r' \leq \cxl} \sum_{V \in \mathsf{G}_{r'}} \frac{R^k (1 + \varepsilon)k!}{\da^k (1-\err)^{2k}} \leq |\mathsf{G}_a|^{r} \frac{R^k (1 + \varepsilon)k!}{\da^k (1-\err)^{2k}} \,,
    \end{equation}
    where we define $|\mathsf{G}_a|:= O(\poly(n)) |\mathsf{G}|$, the cardinality of our universal gate set $\mathsf{G}$ on $n$ qubits and $O(\poly(n))$ ancillas. 
    This upper bound on the number of $\sigma$'s preparable with size-$r$ circuits using a polynomial number of ancillas is proved as follows. There are $|\mathsf{G}|$ gates that can be applied to any pair of qubits. If there are $n$ qubits and $O(\poly(n))$ ancillas, then there are $\binom{O(\poly(n))}{2} < O(\poly(n))$ pairs of qubits. Therefore, there are at most $O(\poly(n)) |\mathsf{G}|$ possible unitaries from $U(2^{O(\poly(n))})$ that we can construct using our gate set. Lastly, we observe that $\sum_{r'\leq r} x^{r'} \leq x^{r+1}$ for $x\geq 2$, which follows from summing the geometric series, and $(O(\poly(n)) |\mathsf{G}|)^{r+1} = (O(\poly(n)) |\mathsf{G}|)^r$ for constant-sized universal gate sets. Therefore, $\sum_{r' \leq \cxl} \sum_{V \in \mathsf{G}_{r'}} 1 \leq (O(\poly(n)) |\mathsf{G}|)^r$ and
    \begin{align}
        \label{eq:repeat_proof_till_here}
        \pr_{U \sim \nu}(\CC(\ssa) \leq \cxl) \leq |\mathsf{G}_a|^r \frac{R^k (1 + \varepsilon)k!}{\da^k (1-\err)^{2k}}.
    \end{align}
    In the case that $\na > \nb$, $R \leq \db$ because the Schmidt rank across a bipartition is at most the Hilbert space dimension of the smaller part in the partition.
    Demanding that the probability upper bound remain exponentially suppressed in $\na$, i.e.,
    \begin{align}
        |\mathsf{G}_a|^r \frac{R^k (1 + \varepsilon)k!}{\da^k (1-\err)^{2k}} \leq \frac{1}{q^{\na}}
    \end{align}
    and using Stirling's approximation for $k!$, we determine the following bound on $\cxl$ (the complexity lower bound) for all $\varepsilon < 1$:
    \begin{equation}
        r \leq \frac{1}{\log(|\mathsf{G}_a|)}\Big((\na - \nb) k \log(q) - k\log(k) + 2k \log(1-\delta) - \na \log(q)\Big)\,.
    \end{equation}
    Finally, we may replace $k$ by the circuit depth $t$ to form an $\varepsilon$-approximate unitary $k$-design in relative error. Let $k = k(t)$; then the complexity lower bound is
    \begin{align}
        \label{eq:insert_koft_nagtnb}
        \frac{\na}{\log(|\mathsf{G}_a|)} \left((2 - n/\na) k(t) \log(q) - \frac{k(t)}{\na} \log(k(t)) + \frac{2k(t)}{\na} \log(1-\delta) - \log(q)\right).
    \end{align}
    We refer to Corollary 1.7 of \cite{chen2024incompressibility}, reproduced in \autoref{thm:copy_incom}, for the asymptotic form of $k(t)$:
    \begin{align}
        t &= C' (nk \log(2) + \log(1/\varepsilon)) \log^{7}(k)\\
        \label{eq:invert_lambertw_1}
        &= C n k \log^{7}(k)
    \end{align}
    for a constant fixed $\varepsilon > 0$ and corresponding large-enough constants $C', C > 0$. We invert \autoref{eq:invert_lambertw_1} to find $k$ as a function of $t$ using the Lambert $W$-function, whose defining equation is $x = W(x) e^{W(x)}$, as was done in \cite{oszmaniec2024saturation}. To follow the procedure there, we define $\tau := t / (Cn)$. Then, we note that \autoref{eq:invert_lambertw_1} is equivalent to $\tau = k \log^\alpha(k)$. We can find $k$ as a function of $\tau$, i.e., $k = k(\tau)$. In general, $k(\tau) := \tau / \alpha^\alpha W^{\alpha}(\tau^{1/\alpha} / \alpha)$ solves the equation $\tau = k(\tau) \log^\alpha(k(\tau))$. In particular, we invert \autoref{eq:invert_lambertw_1} to find that
    \begin{align}
        k(t) = \left\lfloor \frac{t}{7^7 Cn W^7\left(\frac{1}{7}\left(\frac{t}{Cn}\right)^{1/7}\right)}\right\rfloor.
    \end{align}
    Finally, we give a lower bound on $k$ by realizing that for $x > e$, $W(x) \leq \log(x)$:
    \begin{equation}
        k(t) \geq \left\lfloor \frac{t}{7^7 Cn  \log^7\left(\frac{1}{7}\left(\frac{t}{Cn}\right)^{1/7}\right)}\right\rfloor = \left\lfloor \frac{t}{Cn  \log^7\left(\frac{t}{7^7 Cn}\right)}\right\rfloor\,.
        \label{eq:koft_nagtnb}
    \end{equation}
    Before we attempt the final step of the proof by inserting $k(t)$ into \autoref{eq:insert_koft_nagtnb}, we need to make sure that we do so in the regime where \autoref{eq:insert_koft_nagtnb} is a growing function of $k$, because \autoref{eq:koft_nagtnb} is a lower bound on $k$. Taking the derivative of \autoref{eq:insert_koft_nagtnb}, we notice that as long as $\na$ is a constant fraction greater than $1/2$ of the total system size $n$, then unless $k = O(2^{\na ^2})$, the derivative remains positive asymptotically. However, asymptotically, the complexity saturates before $k = O(2^{\na ^2})$. Therefore, we may safely insert \autoref{eq:koft_nagtnb} into \autoref{eq:insert_koft_nagtnb} to complete the proof of \autoref{thm:complexity_growth_na_gt_nb}, which says that the complexity corresponding to a subsystem of a depth-$t$ random circuit grows as
    \begin{align}
        \Omega\left(\frac{t}{\log(n) \log^7(t)}\right).
    \end{align}
\end{proof}

\begin{proof}[Proof of \autoref{prop:rhosigma_moment}]
First, we note that the Haar-averaged moments of $\tr(\rho\sigma)$ are bounded as
\begin{equation}
    \Ex_{\mu_H}\Big[ \big(\tr(\rho\sigma)\big)^k \Big] = \int d\mu_H(U)\, \tr\big(\big(\tr_B U\ketbra{\psi}U^\dagger\big)^{\otimes k}\sigma^{\otimes k}\big) = \int d\mu_S(\phi)\, \tr\big(\big(\tr_B \ketbra{\phi}\big)^{\otimes k}\sigma^{\otimes k}\big)\,,
\end{equation}
as the distribution of states $U\ketbra{\psi}U^\dagger$, for $U\sim\mu_H$ and any fixed state $\ket\psi$, is the Haar measure on the space of states $\mu_S$, induced from the Haar measure on the space of unitaries. For Haar-random states $\ket\phi\sim \mu_S$, the average of the $k^\text{th}$ moments of the states is equal to the projector onto the symmetric subspace \cite{harrowchurch}
\begin{equation}
    \Ex_{\ket\phi \sim \mu_S}\big[ \ketbra{\phi}^{\otimes k}\big] = \binom{D+k-1}{k}^{-1}\Pi_{\rm sym}\,, \, \text{ where } \, \Pi_{\rm sym} = \frac{1}{k!}\sum_{\pi\in S_k} P_\pi
\end{equation}
and where $P_\pi$ is a permutation operator on the $k$-fold space and $S_k$ is the symmetric group on $k$ elements. Applying this formula, we can compute the Haar moments of $\tr(\rho\sigma)$ as
\begin{align}
    \Ex_{\mu_H}\Big[ \big(\tr(\rho\sigma)\big)^k \Big] &= \frac{1}{\binom{D+k-1}{k}}\frac{1}{k!} \sum_{\pi\in S_k} \tr\big(P_\pi^{(A)} \sigma^{\otimes k}\big) \tr\big(P_\pi^{(B)}\big)\\
    &= \frac{1}{\binom{D+k-1}{k}}\frac{1}{k!} \sum_{\pi\in S_k} \db^{\ell(\pi)} \prod_{i=1}^{\ell(\pi)} \tr(\sigma^{\lambda_i(\pi)})\\
    &= \frac{1}{\prod_{i=1}^{k}(D+k-i)}\sum_{\pi\in S_k} \db^{\ell(\pi)} \prod_{i=1}^{\ell(\pi)} \tr(\sigma^{\lambda_i(\pi)})\,,
\end{align}
where, for a permutation $\pi\in S_k$, $\lambda(\pi)$ is the cycle type of the permutation, i.e., $\lambda\vdash k$ is an integer partition of $k$ given by $\lambda = (\lambda_1, \lambda_2, \ldots, \lambda_{\ell(\pi)})$ where $\ell(\pi)$ is the length of the cycle type of the permutation. As $\sigma$ is Hermitian, nonnegative, and trace-one, we have $\tr(\sigma^i)\leq 1$. As $\ell(\pi)\leq k$, we have the bound $\sum_{\pi\in S_k} \db^{\ell(\pi)} \leq k!\, \db^k$, from which we conclude that
\begin{equation}
    \label{eq:final_rel_err_moments}
    \Ex_{\mu_H}\Big[ \big(\tr(\rho\sigma)\big)^k \Big] \leq \frac{\sum_{\pi\in S_k} \db^{\ell(\pi)}}{D(D+1)\cdots(D+k-1)}\leq k! \frac{\db^k}{D^k} = \frac{k!}{\da^k}\,.
\end{equation}

So far, we have considered Haar moments. To bound the moments of $\tr(\rho\sigma)$ over an \emph{approximate} unitary design $\nu$, we consider
\begin{equation}
    \Ex_\nu\Big[ \big(\tr(\rho\sigma)\big)^k \Big] = \Ex_\nu\Big[ \big(\tr(\rho\sigma)\big)^k \Big] - \Ex_{\mu_H}\Big[ \big(\tr(\rho\sigma)\big)^k \Big] + \Ex_{\mu_H}\Big[ \big(\tr(\rho\sigma)\big)^k \Big]\,.
\end{equation}
Let $M := (\sigma \otimes \iden_{B})^{\otimes k} = \sum_{\lambda} \lambda \ketbra{\lambda}$, where all $\lambda \geq 0$. We find that
\begin{align}
    \Ex_\nu\Big[ \big(\tr(\rho\sigma)\big)^k \Big] - \Ex_{\mu_H}\Big[ \big(\tr(\rho\sigma)\big)^k \Big] &= \tr \left(\Big( \Phi_\nu^{(k)}\big(\ketbra{\psi}^{\otimes k}\big) - \Phi_{\mu_H}^{(k)}\big(\ketbra{\psi}^{\otimes k}\big) \Big) M\right) \\
    &= \sum_{\lambda} \lambda \bra{\lambda} \Big( \Phi_\nu^{(k)}\big(\ketbra{\psi}^{\otimes k}\big) - \Phi_{\mu_H}^{(k)}\big(\ketbra{\psi}^{\otimes k}\big) \Big) \ket{\lambda} \\
    &\leq \sum_{\lambda} \lambda \bra{\lambda} \ep\,\Phi_{\mu_H}^{(k)}\big(\ketbra{\psi}^{\otimes k}\big) \ket{\lambda} \\
    &\leq \tr \left(\ep \, \Phi_{\mu_H}^{(k)}\big(\ketbra{\psi}^{\otimes k}\big) M \right) \\
    &\leq \varepsilon \Ex_{\mu_H}\Big[ \big(\tr(\rho\sigma)\big)^k \Big].
\end{align}
Finally, using \autoref{eq:final_rel_err_moments} (moment bounds for $k$-designs), the moments of $\tr(\rho\sigma)$ for ap\-prox\-i\-mate unitary $k$-designs are bounded as
\begin{equation}
    \Ex_\nu\Big[ \big(\tr(\rho\sigma)\big)^k \Big] \leq (1+\varepsilon)\frac{k!}{\da^k}\,.
\end{equation}
\end{proof}

Lastly, we prove that at exponential times, the subsystem complexity lower bound is exponential. We can do this using exponentially high-degree designs. 
\begin{proposition}
    Consider a subsystem $A$ of an $n$-qubit brickwork random quantum circuit of depth $t$, with $n_A>n_B$. When the circuit depth is $t=\Omega(2^{3n} n^5)$, the complexity of the time-evolved state $\rho_A(t) = \tr_B(U\ketbra{\psi}U^\dagger)$ is exponential, $\CC(\rho_A(t))=\Omega(2^{n_A-n_B})$.
\end{proposition}
\begin{proof}
This follows from a slight modification of the bound in \autoref{thm:complexity_growth_na_gt_nb}. First, Ref.~\cite{brandao2016local} proved an exponentially small but $k$-independent lower bound on the spectral gap of random quantum circuits. This was improved in Ref.~\cite{Haferkampt5} to $\Delta(H)\geq C n^{-4} 2^{-2n}$, where $C<1$ is a constant. The gap $g_{\rm bw}$ of brickwork RQCs is then $1-g_{\rm bw}\geq C' \Delta(H)$. Combining Lemmas 19 and 20 in Ref.~\cite{brandao2016local} and taking $k=2^n$ and $\ep=1/2^n$, we find that when $t=\Omega(2^{3n} n^5)$, RQCs form $k=2^n$ designs. Now recall that in the proof of \autoref{thm:complexity_growth_na_gt_nb}, we found
\begin{align}
    \pr_{U \sim \nu}(\CC(\ssa) \leq \cxl) \leq |\mathsf{G}_a|^r \frac{\db^k (1 + \varepsilon)k!}{\da^k(1-\err)^{2k}}\leq |\mathsf{G}_a|^r \frac{(1 + \varepsilon)}{(1-\err)^{2k}} \left(\frac{k\db}{2\da}\right)^k\,,
\end{align}
using $k!\leq (k/2)^k$ for $k\geq 6$. Taking $k=\da/\db$, we find that $\pr_{U \sim \nu}(\CC(\ssa) \leq \cxl) \lesssim |\mathsf{G}_a|^r 2^{-\da/\db}$. Thus with extremely high probability, when the circuits are $t=\Omega(2^{3n} n^5)$ and form $k=\da/\db$ designs, the complexity of the subsystem is exponential: $\CC(\ssa) =\Omega(2^{n_A-n_B})$.
\end{proof}

We end the subsection by commenting on possible improvements to \autoref{prop:rhosigma_moment} and \autoref{thm:complexity_growth_na_gt_nb}. Our bound can be tightened a bit by considering centered moments $\Ex_\nu\Big[ \big|\tr(\rho\sigma) - 1/\da \big|^k \Big]$, which turn out to concentrate better. For example, for subexponential moments, we can show that
\begin{equation}
    \Pr\Big( \dist(\rho,\sigma)\leq \delta\Big) \leq \Pr\bigg( \Big|\tr(\rho\sigma)-\frac{1}{\da}\Big|^{2k} \geq \frac{\Delta^{2k}}{\db^{2k}}\bigg)
    \leq \frac{\db^k}{\da^{2k}}\frac{(2k)!(1+\ep)}{\Delta^{2k}},
\end{equation}
where $\Delta := 1 - \frac{1}{\da} - \delta^2$. However, these improvements in concentration do not qualitatively change our results, and thus we opted for the simpler concentration result in \autoref{prop:rhosigma_moment}. The improved dimension scaling above corresponds to a linear complexity growth with slightly larger slope. We believe that even tighter concentration should follow from centered moments of the fidelity. This might allow one to see the correct slope of the linearly growing subsystem complexity.

\subsection{Complexity for Less-Than-Half System Size \texorpdfstring{$\na < \nb$}{nA < nB}}

\begin{proof}[Proof of \autoref{thm:complexity_growth_na_lt_nb}] 
    Here, we borrow largely from the proof of \autoref{thm:complexity_growth_na_gt_nb} presented in \autoref{subsec:complexity_growth_na_gt_nb}. In fact, the proof is identical up to \autoref{eq:repeat_proof_till_here}. Thereafter, the proof differs as follows. In the proof of \autoref{thm:complexity_growth_na_gt_nb}, we replaced the rank $R$ by $\db$, the maximum Schmidt rank of a bipartition of a pure state, which is equal to the Hilbert space dimension of the smaller part in the partition. Now note that for any bipartition, the rank of $\ssa$ increases by at most a factor of $q^4$ after two layers (an even and an odd layer) of gates. If the gate set were exactly equal to $U(q^2)$, then this would be true with probability 1 (refer to Appendix A of \cite{nahum2017quantum}). For simplicity, consider even $t$ such that the rank increases by a factor of $\q^{2t}$ after $t$ layers, corresponding to $t/2$ repetitions of even and odd layers of the circuit. We then find that
    \begin{align}
        r \leq \frac{1}{\log(|\mathsf{G}_a|)}\big((\na - 2t) k \log(q) - k\log(k/2) + 2k \log(1-\delta) - \na \log(q)\big)\,,
    \end{align}
    where again $|\mathsf{G}_a| = O(\poly(n))|\mathsf{G}|$.
    Finally, we may do as we did earlier and replace $k$ by the circuit depth $t$ (which is naturally the number of convolutions in the unitary design results; see Figure~1 of \cite{chen2024incompressibility}) to form $\varepsilon$-approximate unitary $k$-designs in relative error from \cite{schuster2025random}. Let $k = k(t)$; then the complexity lower bound is
    \begin{align}
        \label{eq:insert_koft_naltnb}
        \frac{\na}{\log(|\mathsf{G}_a|)} \left((1 - 2t/\na) k(t) \log(q) - \frac{k(t)}{\na} \log(k(t)) + \frac{2k(t)}{\na} \log(1-\delta) - \log(q)\right).
    \end{align}
    In Corollary 1 of \cite{schuster2025random}, reproduced in \autoref{thm:copy_shall}, the authors showed that
    \begin{equation}
        t = C' \log(nk/\varepsilon) k \log^7(k)
    \end{equation}
    for a sufficiently large constant $C' > 0$. In the case of $\na < \nb$, we will only be interested in $k \linebreak[1] = \linebreak[1] O(\poly(n))$ because beyond that, superpolynomial circuit size is required, by which time the complexity has already saturated. Therefore, in this case, for a sufficiently large constant $C > 0$,
    \begin{align}
        \label{eq:invert_lambertw_2}
        t = C \log(n) k \log^7 (k).
    \end{align}
    As we did in the proof of \autoref{thm:complexity_growth_na_gt_nb}, we invert \autoref{eq:invert_lambertw_2} to find the lower bound
    \begin{align}
        k \geq\left\lfloor \frac{t}{C \log(n)  \log^7\left(\frac{t}{7^7 C\log(n)}\right)} \right\rfloor,
    \end{align}
    which we can insert into \autoref{eq:insert_koft_naltnb} as long as that equation has a positive derivative with respect to $k$. The latter property is available (asymptotically in $n$) as long as $t$ is a constant fraction of $\na / 2$. We conclude that the complexity lower bound is
    \begin{align}
        \Omega\left(\frac{\na}{\log^2(n)  \log^7(t)} \left(\left(t - \frac{2t^2}{\na}\right) - \frac{t}{\na} \log(t)\right)\right)\approx \Omega\left(\frac{\na}{\log^2(n)  \log^7(t)} \left(t - \frac{2t^2}{\na}\right)\right)\,.
    \end{align}
    The factor of 2 arises from the fact that after two layers (an even and an odd layer), the rank of the subsystem grows by at most a factor of $q^4$, thus $q^2$ per layer on average. For odd time steps, the exact rank increase will depend on the choice of subsystem, but at even time steps, after an equal number of even and odd layers, the rank is $q^{2t}$. Thus we quote our results for even times.
\end{proof}

As written, our proof for subsystem complexity growth for less-than-half subsystems uses log depth designs and holds for patchwork RQCs. The same bounds would hold for brickwork RQCs if it were proven that they form relative-error approximate $k$-designs in $\log(n)$ depth. Even establishing that brickwork RQCs formed $k$-designs in depth $n^\alpha$ for some $\alpha<1$ would give a nontrivial complexity lower bound, larger than the late-time complexity value and establishing that the complexity must rise and then fall. We note that recent work in Ref.~\cite{Laracuente25} shows that brickwork RQCs have relative entropy decay after log depth, giving some evidence that brickwork RQCs are also log depth designs.

\subsection{Saturation Complexity and Timescales}

\subsubsection*{Complexity Saturation for Brickwork RQCs}

\begin{theorem}[Theorem~1 of \cite{cotler2022fluctuations}]
    \label{thm:cotler2022fluctuations}
    Assume $A$ is a contiguous subsystem of a one-dimensional $n$-qubit system with periodic boundary conditions. For some $\err>0$, the time-evolved state $\ssa = \tr_{B}(U \linebreak[1] \ketbra{\psi} \linebreak[1] U^\dagger)$ of a depth-$t$ brickwork random quantum circuit $U \sim \bwrqc$ obeys
    \begin{equation}
        \Pr_{U\sim\bwrqc}\Big( \dist\big(\ssa,\iden_A/\da\big) \geq \err\Big) \leq \frac{\da e^{-\lambda (t-1)}}{4\err^2-\frac{\da}{\db}}\,,
    \end{equation}
    where $\lambda = 2\log\big(\frac{5}{4}\big)\approx 0.45$, we require that $\delta^2\geq \frac{D_A}{4D_B}$, and $\ket{\psi}$ is an arbitrary pure state.
\end{theorem}

\begin{proof}
The proof of this statement follows from a computation of the expected purity of brickwork random quantum circuits. First, we note that
\begin{equation}
    \bigg\|\ssa - \frac{\iden_A}{\da}\bigg\|_1^2 \leq \da \bigg\|\ssa - \frac{\iden_A}{\da}\bigg\|_2^2 = \da\tr\big(\ssa^2\big) - 1\,,
\end{equation}
using the relation between Schatten norms $\|X\|_1 \leq \sqrt{\rank(X)} \|X\|_2$. Next, we apply this inequality in a statement about probabilities:
\begin{align}
    \label{eq:purity_follow_steps_start}
    \Pr_{\bwrqc}\Big( \dist\big(\ssa,\iden_A/\da \big) \geq \err\Big) &= \Pr_{\bwrqc}\bigg(\bigg\|\ssa - \frac{\iden_A}{\da}\bigg\|_1^2 \geq 4\err^2\bigg)\\ 
    &\leq \Pr_{\bwrqc}\Big(\da\tr\big(\ssa^2\big) - 1 \geq 4\err^2\Big)\\
    &= \Pr_{\bwrqc}\bigg(\da\tr\big(\ssa^2\big) - 1 - \frac{\da}{\db} \geq 4\err^2 - \frac{\da}{\db} \bigg)\\
    \label{eq:sub_avg_bw_purity}
    &\leq \frac{\da\Ex_{\bwrqc}\big[\tr\big(\ssa^2\big)\big] - 1 - \frac{\da}{\db}}{4\err^2 - \frac{\da}{\db}}\,,
\end{align}
where in the last step we used Markov's inequality and require that $\delta^2\geq \frac{\da}{4\db}$.

Assuming that $A$ is a contiguous subsystem of a 1D brickwork random quantum circuit, we then quote Proposition~1 of \cite{cotler2022fluctuations},
\begin{equation}
    \Ex_{\bwrqc}\big[ \tr(\ssa^2)\big]\leq \frac{1}{\da} + \frac{1}{\db} + \bigg(\frac{2q}{q^2+1}\bigg)^{2(t-1)}\,.
    \label{eq:avgpur}
\end{equation}
Substituting this expectation value into \autoref{eq:sub_avg_bw_purity} and considering local qubits with $q=2$, we complete the proof:
\begin{equation}
    \Pr_{\bwrqc}\Big( \dist\big(\ssa,\iden_A/\da\big) \geq \err\Big) \leq \frac{\da e^{-\lambda (t-1)}}{4\err^2 - \frac{\da}{\db}}\,.
\end{equation}
\end{proof}

\autoref{thm:cotler2022fluctuations} clarifies two points about complexity in the regime $\na < \nb$: (i) the complexity saturates at timescales of order $O(\na)$, and (ii) the saturation complexity is that of preparing the maximally mixed state on $A$.
\begin{corollary}
    \label{cor:cotler2022fluctuations}
    For the same scenario as in \autoref{thm:cotler2022fluctuations}, after a depth $t = O(\na)$, the complexity of subsystem $A$ is equal to that of $\iden_A / \da$, i.e., $\CC(\ssa) = \CC(\iden_A / \da) = O(n_A)$ with high probability.
\end{corollary}
The complexity of the maximally mixed state is exactly $n_A$, prepared from an initial product state using a circuit of $n_A$ Hadamards and $n_A$ CNOTs on $2n_A$ qubits. For any universal gate set, the complexity is $\CC(\iden_A / \da) = O(n_A)$.

\subsubsection*{Comments}

The complexity of the time-evolved subsystem $\rho_A(t)$ becomes trivial at time $t\approx n_A\log(2)/\lambda$, which is $t\approx 1.55\, n_A$ for local qubits. But for $q$-dimensional local qudits, $\lambda = 2\log\big(\frac{q^2+1}{2q}\big)$, and thus the time $t\approx n_A\log(q)/\lambda$ it takes for the subsystem to become maximally mixed is $t\approx n_A/2$ as $q\ra\infty$. We note that this precisely coincides with the time at which the complexity lower bound for less-than-half subsystems in \autoref{thm:complexity_growth_na_lt_nb} becomes trivial. We conjecture that at large local dimension, the complexity drops sharply at time $t\approx n_A/2$. Unfortunately, our current proof technique, using moments of the density matrix to bound the complexity, is largely agnostic to the local dimension. Proving subsystem complexity lower bounds which peak at $t>n_A/4$ requires a new approach.

We further remark that the purity decay for brickwork RQCs not only tells us when the complexity of the subsystems has saturated, but can also give us lower bounds on the subsystem complexity growth. The expected purity lower-bounds the rank of the marginal as $\Ex[S_0(\rho_A)]\geq \Ex[S_2(\rho_A)]\geq -\log \Ex[\tr(\rho_A^2)]$. Thus from the purity decay at a given time step, say after a single layer of the circuit, we can conclude that the rank must have increased, and therefore that $\CC(\rho(t))$ has increased by an $O(1)$ amount. Strictly speaking, the rank can increase without changing the complexity of the state so long as we stay within $\delta$ of the $\rho$ in trace distance. Thus we should consider the smoothed rank, the smallest Schmidt rank of a state $\sigma$ which approximates $\rho$. As the purity of a fixed subsystem $A$ decays, the number of gates that spanned the bipartition becomes $\Omega(n_A)$. But the averaged purity in \autoref{eq:avgpur} holds for any contiguous subsystem: thus the rank must be increasing across all contiguous bipartitions of our 1D system. We can repeat this argument to get a subsystem complexity growth to $\Omega(n_A^2)$, at which point the complexity will start to decrease as the light cones of the boundaries of the subsystem $A$ meet and overlap.

\subsubsection*{Complexity Saturation for Patchwork RQCs}

\autoref{cor:cotler2022fluctuations} would have been sufficient for our discussion had we known that a complexity lower bound (\autoref{thm:complexity_growth_na_lt_nb}) also holds for the similar setup of brickwork random quantum circuits. However, for technical reasons originating from \cite{schuster2025random}, we can only prove the lower bound in \autoref{thm:complexity_growth_na_lt_nb} for patchwork random quantum circuits (see \autoref{def:pwrqc}), and thus, we need versions of \autoref{thm:cotler2022fluctuations} and \autoref{cor:cotler2022fluctuations} for patchwork random quantum circuits. We state and prove those analogues in \autoref{thm:patchwork_purity} and \autoref{cor:patchwork_purity}.

\begin{theorem}
    \label{thm:patchwork_purity}
    Assume $A$ is a contiguous subsystem of a one-dimensional $n$-qubit system with periodic boundary conditions. For some $\err>0$, the time-evolved state $\ssa = \tr_{B}(U \ketbra{\psi} U^\dagger)$ of a depth-$t$ patchwork random quantum circuit $U \sim \pwrqc$ obeys
    \begin{equation}
        \pr_{U\sim\pwrqc}\Big( \dist\big(\ssa,\iden_A/\da\big) \geq \err\Big) \leq \frac{\da e^{-t/C + \log(n)}}{\err^2 - \frac{\da}{4\db}}
    \end{equation}
    where $\ket{\psi}$ is an arbitrary pure state and $C > 0$ is a constant. For $\epsilon \in (0, 1)$ concentration,
    \begin{equation}
        \pr_{U\sim\pwrqc}\Big( \dist\big(\ssa,\iden_A/\da\big) \geq \err\Big) \leq \epsilon\,,
    \end{equation}
    we demand that $n_A \leq \frac{n}{2} - \log_2\Big(\frac{1}{\epsilon}\Big) - \log_2\Big( \frac{8n}{4\delta^2 - \frac{D_A}{D_B}} \Big)$.
\end{theorem}

\begin{proof}
    The proof relies on the result of \cite{schuster2025random} that patchwork random quantum circuits of depth $t \linebreak[1] = \linebreak[1] \widetilde{O}(\ln(nk/\varepsilon) k\log^7(k))$ form $\varepsilon$-approximate unitary $k$-designs. Following the steps \autoref{eq:purity_follow_steps_start}--\autoref{eq:sub_avg_bw_purity} identically from the proof of \autoref{thm:cotler2022fluctuations} with only $\bwrqc$ replaced by $\pwrqc$, we find
    \begin{align}
        \label{eq:return_markov}
        \pr_{\pwrqc}\Big( \dist\big(\ssa,\iden_A/\da\big) \geq \delta\Big) &\leq \frac{\da\Ex_{\pwrqc}\big[\tr\big(\ssa^2\big)\big] - 1 - \frac{\da}{\db}}{4\err^2 - \frac{\da}{\db}}
    \end{align}
    for $\delta \geq \frac{\da}{4\db}$. We proceed by using the swap trick to express the purity $\tr(\ssa^2)$ of $\ssa$ as a trace of two copies of the pure state $U\ket{\psi}$, where $\rho_A = \tr_B(U\ketbra{\psi}U^\dagger)$. Consider a doubled Hilbert space (on $2n$ qubits) and the state $(U\ket{\psi})^{\otimes 2}$. We identify the systems $A$ and $B$ on the two copies, and we refer to them as $A_1$, $A_2$, $B_1$, $B_2$. We define a unitary operator $\mathbb{S}$ on the doubled Hilbert space such that
    \begin{align}
        \mathbb{S}\ket{\psi_{A_1}, \psi_{B_1}, \psi_{A_2}, \psi_{B_2}} := \ket{\psi_{A_2}, \psi_{B_1}, \psi_{A_1}, \psi_{B_2}},
    \end{align}
    where $\ket{\psi_{A_1}, \psi_{B_1}, \psi_{A_2}, \psi_{B_2}}$ is an arbitrary tensor product state across the four systems $A_1$, $B_1$, $A_2$, $B_2$. In words, $\mathbb{S}$ swaps the state across the systems $A_1$ and $A_2$. With the help of this operator, we can express the purity $\tr(\ssa^2)$, where the trace is over system $A$ in the single-copy picture, as $\tr(\mathbb{S} U^{\otimes 2} \ketbra{\psi}^{\otimes 2} U^{\dagger \otimes 2})$, where the trace is over systems $A_1 A_2 B_1 B_2$ in the double-copy picture. We now return to the Markov inequality in \autoref{eq:return_markov}:
    \begin{align}
        \pr_{\pwrqc}\Big( &\dist\big(\ssa,\iden_A/\da\big) \geq \delta\Big) \nonumber \\
        &\leq \frac{\da}{4\err^2 - \frac{\da}{\db}} \left(\Ex_{\pwrqc} \left[\tr(\mathbb{S} U^{\otimes 2} \ketbra{\psi}^{\otimes 2} U^{\dagger \otimes 2})\right] - \frac{1}{\da} - \frac{1}{\db}\right)\\
        &= \frac{\da}{4\err^2 - \frac{\da}{\db}} \left(\tr(\mathbb{S} \Ex_{\pwrqc} [U^{\otimes 2} \ketbra{\psi}^{\otimes 2} U^{\dagger \otimes 2}]) - \frac{1}{\da} - \frac{1}{\db}\right)\\
        &= \frac{\da}{4\err^2 - \frac{\da}{\db}} \left(\tr(\mathbb{S} \Phi_{\pwrqc}^{(2)} (\ketbra{\psi})) - \frac{1}{\da} - \frac{1}{\db}\right)\\
        &= \frac{\da}{4\err^2 - \frac{\da}{\db}} \left(\tr(\mathbb{S}( \Phi_{\pwrqc}^{(2)} - \Phi_{\mu_H}^{(2)}) (\ketbra{\psi})) + \tr(\mathbb{S} \Phi_{\mu_H}^{(2)} (\ketbra{\psi})) - \frac{1}{\da} - \frac{1}{\db}\right)
    \end{align}
    where, in the second-to-last expression, we packaged the expectation values into $2$-fold channels (refer to \autoref{def:moment_channel}) and, in the last expression, we added and subtracted the $2$-fold Haar-random channels. Next, we use two facts to simplify the right-hand side: (1) the trace of $\mathbb{S}$ times the difference of the channels is bounded above by the maximum distinguishability of the channels as measured by the diamond norm, and (2) the purity of subsystems of Haar-random states is given by $\frac{\da + \db}{\da \db + 1} \leq \frac{1}{\da} + \frac{1}{\db}$. Using these two points, we find
    \begin{align}
        \pr_{\pwrqc}\Big( \dist\big(\ssa,\iden_A/\da\big) \geq \delta\Big) &\leq \frac{\da}{4\err^2 - \frac{\da}{\db}} \big\lVert \Phi_{\pwrqc}^{(2)} - \Phi_{\mu_H}^{(2)} \big\rVert_{\diamond}.
    \end{align}
    In particular, suppose $\pwrqc$ forms relative-error 2-designs with error $\varepsilon$; then the relative error bounds the diamond norm error at the cost of a factor of two (Lemma 3 in \cite{brandao2016local}):
    \begin{align}
        \label{eq:sub_err_ub}
        \pr_{\pwrqc}\Big( \dist\big(\ssa,\iden_A/\da\big) \geq \delta\Big) \leq \frac{2\da \varepsilon}{4\err^2 - \frac{\da}{\db}}.
    \end{align}
    Finally, we bound the error by the size of the random circuit. The depth $t$ of such a patchwork random quantum circuit that provides relative-error approximate unitary designs, derived in \cite{schuster2025random}, is given by
    \begin{align}
        &t \leq C \ln(2n/\varepsilon) \implies \varepsilon \leq 2n \exp(-t/C)
    \end{align}
    for some constant $C > 0$. Recall from \autoref{rem:time} that increasing the time $t$ in patchwork random quantum circuits corresponds to picking an entirely new circuit from a family of circuits. We substitute this upper bound on the error $\varepsilon$ into \autoref{eq:sub_err_ub} and complete the proof:
    \begin{equation}
        \frac{2\da \varepsilon}{4\err^2 - \frac{\da}{\db}} \leq \frac{4\da n \exp(-t/C)}{4\err^2 - \frac{\da}{\db}} = \frac{\da \exp(-t/C + \log(n))}{\delta^2 - \frac{\da}{4\db}}\,.
    \end{equation}
    A subtlety arises when considering patchwork random quantum circuits with regards to the error of the approximate unitary design $\varepsilon$ and the size of the patches $2\xi$ (refer to \autoref{def:pwrqc}). Namely, if we demand an exponentially small $\varepsilon$, then the patch size $2\xi = 2 \log_2(nk^2 / \varepsilon)$ required for sublinear-depth $\ep$-approximate unitary $k$-designs from Ref.~\cite{schuster2025random} could exceed the total system size $n$, which is illogical. This issue is significant for the current theorem,  \autoref{thm:patchwork_purity}, and the following \autoref{cor:patchwork_purity}. In particular, to guarantee an $\epsilon \in (0, 1)$ concentration in \autoref{eq:sub_err_ub}, we demand that $\varepsilon \leq \epsilon (4\delta^2 - \frac{D_A}{D_B}) / 2 D_A$, which constrains $n_A$ by in turn requiring that the patch size $2\xi = 2 \log_2(nk^2 / \varepsilon)$ remain less than or equal to $n$ for patchwork $\ep$-approximate unitary $2$-designs as follows:
\begin{equation}
    \label{eq:patch_constraint}
    2 \log_2(nk^2 / \varepsilon) \leq n \,\implies\, 2 \log_2\Bigg(\frac{8 n D_A}{\epsilon\Big(4\delta^2 - \frac{D_A}{D_B}\Big)}\Bigg) \leq n \,\implies\, n_A \leq \frac{n}{2} - \log_2\Big(\frac{1}{\epsilon}\Big) - \log_2\Bigg(\frac{8 n}{4\delta^2 - \frac{D_A}{D_B}}\Bigg) \,.
\end{equation}
\end{proof}

\begin{corollary}
    \label{cor:patchwork_purity}
    For the same scenario as in \autoref{thm:patchwork_purity} with $\epsilon \in (0, 1)$ concentration, after a patchwork random quantum circuit depth $t = O(\na)$, the complexity of subsystem $A$ is equal to that of $\iden_A / \da$, i.e., $\CC(\ssa) = \CC(\iden_A / \da)$, with probability at least $1-\epsilon$.
\end{corollary}

In \autoref{cor:patchwork_purity}, by ``patchwork random quantum circuit depth,'' we are referring to distinct circuits at different times, as explained in \autoref{rem:time}. This completes our discussion of the saturation complexity for the case $\na < \nb$.

\subsubsection*{Comments}

Finally, we explain how the constraints on $n_A$ in \autoref{thm:patchwork_purity} and \autoref{cor:patchwork_purity}, coming from logical patch sizes for patchwork $\ep$-approximate unitary $2$-designs, limit our results using \autoref{fig:qi_expectation}~(a). We studied two similar models of 1D random quantum circuits: patchwork and brickwork circuits. The blue line in \autoref{fig:qi_expectation}~(a) is a complexity lower bound for all constant $p < 1/2$, and the red line after $t \geq \na / 2$ is a complexity upper bound for $p$ only as in \autoref{eq:patch_constraint} for patchwork circuits. For brickwork circuits, we can only prove complexity upper bounds. Thus, the red lines in \autoref{fig:qi_expectation}~(a) denote complexity upper bounds for subsystems of those circuits for all $p < 1/2$. In other words, we introduced patchwork circuits to prove complexity lower bounds that we could not prove with brickwork circuits because we do not yet know of a proof of sublinear-depth approximate unitary designs on those circuits. We expect this to be true, so we encourage the reader to read the red lines in  \autoref{fig:qi_expectation}~(a) as provable complexity upper bounds for brickwork circuits and the blue line as a conjectured complexity lower bound for the same, but which we can prove for patchwork circuits.

\newpage
\bibliographystyle{utphys}
\bibliography{refs}
\end{document}